%% file: main.tex
\pdfoutput=1

\documentclass[sigplan,screen]{acmart}

\setcopyright{acmlicensed}
\acmPrice{15.00}
\acmDOI{10.1145/3519939.3523713}
\acmYear{2022}
\copyrightyear{2022}
\acmSubmissionID{pldi22main-p389-p}
\acmISBN{978-1-4503-9265-5/22/06}
\acmConference[PLDI '22]{Proceedings of the 43rd ACM SIGPLAN International Conference on Programming Language Design and Implementation}{June 13--17, 2022}{San Diego, CA, USA}
\acmBooktitle{Proceedings of the 43rd ACM SIGPLAN International Conference on Programming Language Design and Implementation (PLDI '22), June 13--17, 2022, San Diego, CA, USA}

\usepackage{booktabs}   
\usepackage{subcaption} 
\usepackage{proof}
\usepackage{bbm}
\usepackage{mathtools}
\usepackage{xcolor}
\usepackage{amsmath}
\usepackage{balance}

\usepackage{multicol}
\usepackage[inline]{enumitem}

\usepackage{amsthm,mathrsfs}

\usepackage{graphicx}

\newcommand{\nc}{\newcommand}
\nc{\ketbra}[2]{|#1\rangle\!\langle#2|}
\nc{\braket}[2]{\langle#1|#2\rangle}
\nc{\proj}[1]{| #1\rangle\!\langle #1 |}

\newtheorem{theorem}{Theorem}

\newtheorem{lemma}[theorem]{Lemma}

\newtheorem{corollary}[theorem]{Corollary}

\newtheorem{remark}{Remark}
\newtheorem{definition}[theorem]{Definition}

\numberwithin{equation}{subsection}
\numberwithin{theorem}{section}
\numberwithin{example}{section}
\numberwithin{remark}{section}

\usepackage{tikz}
\usetikzlibrary{automata,positioning}

\begin{document}

\title[Algebraic Reasoning of Quantum Programs via Non-idempotent Kleene Algebra (Extended Version)]{Algebraic Reasoning of Quantum Programs via Non-idempotent Kleene Algebra (Extended Version)}

\author{Yuxiang Peng}
\affiliation[obeypunctuation=true]{
  \institution{University of Maryland}, \country{USA}
}

\author{Mingsheng Ying}
\affiliation[obeypunctuation=true]{
\institution{Chinese Academy of Sciences},
\country{China}
}
\affiliation[obeypunctuation=true]{
\institution{Tsinghua University},
\country{China}
}

\author{Xiaodi Wu}
\affiliation[obeypunctuation=true]{
  \institution{University of Maryland}, \country{USA}
}

\begin{abstract}
We investigate the \emph{algebraic} reasoning of quantum programs inspired by the success of classical program analysis based on Kleene algebra. 
One prominent example of such is the famous Kleene Algebra with Tests (KAT), which has furnished both theoretical insights and practical tools. 
The succinctness of algebraic reasoning would  
be especially desirable for scalable analysis of quantum programs, given 
the involvement of exponential-size matrices in most of the existing methods. 
A few key features of KAT including the idempotent law and the nice properties of classical tests, however, 
fail to hold in the context of quantum programs due to their unique quantum features, especially in branching. 
We propose Non-idempotent Kleene Algebra (NKA) as a natural alternative and identify complete and sound semantic models for NKA as well as their quantum interpretations.
In light of applications of KAT, we demonstrate algebraic proofs in NKA of quantum compiler optimization and the normal form of quantum \textbf{while}-programs. 
Moreover, we extend NKA with Tests (i.e., NKAT), where tests model quantum predicates following  effect algebra, 
and illustrate how to encode propositional quantum Hoare logic as NKAT theorems. 
\end{abstract}


\newif\iftechrep\techreptrue 


\newtheorem{construction}[theorem]{Construction}

\newcommand{\eq}[1]{\hyperref[eq:#1]{(\ref*{eq:#1})}}
\renewcommand{\sec}[1]{\hyperref[sec:#1]{Section~\ref*{sec:#1}}}
\newcommand{\thm}[1]{\hyperref[thm:#1]{Theorem~\ref*{thm:#1}}}
\newcommand{\rdef}[1]{\hyperref[def:#1]{Definition~\ref*{def:#1}}}
\newcommand{\rprop}[1]{\hyperref[prop:#1]{Proposition~\ref*{prop:#1}}}
\newcommand{\rsdp}[1]{\hyperref[sdp:#1]{Invariant-SDP~\ref*{sdp:#1}}}
\newcommand{\lem}[1]{\hyperref[lem:#1]{Lemma~\ref*{lem:#1}}}
\newcommand{\cor}[1]{\hyperref[cor:#1]{Corollary~\ref*{cor:#1}}}
\newcommand{\itm}[1]{\hyperref[itm:#1]{\ref*{itm:#1}}}
\newcommand{\tab}[1]{\hyperref[tab:#1]{Tabel~\ref*{tab:#1}}}

\iftechrep
\newcommand{\app}[1]{\hyperref[app:#1]{Appendix~\ref*{app:#1}}}
\else 
\newcommand{\app}[1]{\hyperref[app:#1]{Supplementary Materials~\ref*{app:#1}}}
\fi

\newcommand{\fig}[1]{\hyperref[fig:#1]{Figure~\ref*{fig:#1}}}
\newcommand{\ex}[1]{\hyperref[ex:#1]{Example~\ref*{ex:#1}}}
\newcommand{\defn}[1]{\hyperref[defn:#1]{Definition~\ref*{defn:#1}}}
\newcommand{\cstr}[1]{\hyperref[cstr:#1]{Construction~\ref*{cstr:#1}}}

\newcommand{\at}[0]{\mathrm{at}}
\newcommand{\1}[0]{\mathbbm{1}}
\newcommand{\A}[0]{\mathcal{A}}
\newcommand{\B}[0]{\mathcal{B}}
\newcommand{\D}[0]{\mathcal{D}}
\newcommand{\F}[0]{\mathcal{F}}
\renewcommand{\H}[0]{\mathcal{H}}
\renewcommand{\C}[0]{\mathcal{C}}
\renewcommand{\L}[0]{\mathcal{L}}
\renewcommand{\O}[0]{\mathcal{O}}
\renewcommand{\P}[0]{\mathcal{P}}
\newcommand{\Q}[0]{\mathcal{Q}}
\renewcommand{\S}[0]{\mathcal{S}}
\newcommand{\T}[0]{\mathcal{A}}
\newcommand{\X}[0]{\mathcal{X}}
\newcommand{\Y}[0]{\mathcal{Y}}
\newcommand{\K}[0]{\mathcal{K}}
\newcommand{\N}[0]{\mathcal{N}}
\newcommand{\M}[0]{\mathcal{M}}
\newcommand{\nat}[0]{\mathbb{N}}
\newcommand{\natinf}[0]{\overline{\mathbb{N}}}
\newcommand{\cmpl}{\mathbb{C}}
\newcommand{\llrrangle}[1]{\langle\! \langle #1\rangle\! \rangle}
\newcommand{\ratps}[0]{\natinf^{\mathrm{rat}}\llrrangle{\Sigma^*}}
\newcommand{\ps}[0]{\natinf\llrrangle{\Sigma^*}}
\newcommand{\rpsfin}[0]{\nat^{\mathrm{rat}}\llrrangle{\Sigma^*}}
\newcommand{\psfin}{\nat\llrrangle{\Sigma^*}}
\newcommand{\psft}{\natinf\langle\Sigma^*\rangle}
\newcommand{\fps}[1]{\mathbf{ #1 }}
\newcommand{\fpsc}[2]{\mathbf{ #1 }[#2]}
\newcommand{\CP}[0]{\mathcal{QC}}
\newcommand{\I}[0]{\mathcal{I}}
\newcommand{\empstr}[0]{\epsilon}
\newcommand{\eval}[0]{\mathrm{eval}}
\newcommand{\pred}[0]{\LT_{\mathrm{Pred}}}
\newcommand{\meas}[0]{\LT_{\mathrm{Meas}}}
\newcommand{\cmeas}[0]{\LT_{\mathrm{CMeas}}}
\newcommand{\cp}[1]{\langle #1 \rangle^{\uparrow}}
\renewcommand{\neg}[1]{\overline{ #1 }}
\newcommand{\PO}{\mathcal{PO}}
\newcommand{\cpPO}{\cp{\PO}}
\newcommand{\cpCP}{\cp{\CP}}
\newcommand{\PS}{\mathcal{PS}}
\newcommand{\cpPS}{\cp{\PS}}
\newcommand{\KATmodel}{\vdash_{\mathrm{KAT}}}
\newcommand{\NKAmodel}{\vdash_{\mathrm{NKA}}}
\newcommand{\NKAT}{\mathrm{NKAT}}
\newcommand{\NKA}{\mathrm{NKA}}
\newcommand{\KAmodel}{\vdash_{\mathrm{KA}}}
\newcommand{\circdag}{\diamond}
\newcommand{\enc}{\mathrm{Enc}}
\newcommand{\expsig}{\mathrm{Exp}_{\Sigma}}
\newcommand{\intp}{\mathrm{int}}
\newcommand{\Qint}{\Q_{\intp}}
\newcommand{\Kint}{\K_{\intp}}
\newcommand{\Qintdag}{\Q_{\intp}^{\dagger}}

\newcommand{\bra}[1]{\langle #1 \vert}
\newcommand{\ket}[1]{\vert #1 \rangle}
\newcommand{\tr}[0]{\mathrm{tr}}
\newcommand{\Tr}[0]{\mathrm{tr}}
\newcommand{\sem}[1]{\left\llbracket #1 \right\rrbracket}
\newcommand{\norm}[1]{\left\lVert #1 \right\rVert}
\newcommand{\rpssem}[1]{\{\!\!\{ #1 \}\!\!\}}

\newcommand{\cskip}[0]{{\mathbf{skip}}}
\newcommand{\cfail}[0]{{\mathbf{abort}}}
\newcommand{\cif}[3]{{\mathbf{if}~#1~\mathbf{then}~#2~\mathbf{else}~#3}}
\newcommand{\qif}[1]{\mathbf{case}~#1~\mathbf{end}}
\newcommand{\qwhile}[2]{\mathbf{while}~#1~\mathbf{do}~#2~\mathbf{done}}
\newcommand{\cwhile}[2]{\mathbf{while}~#1~\mathbf{do}~#2~\mathbf{done}}
\newcommand{\qskip}[0]{\mathbf{skip}}
\newcommand{\triple}[3]{\{#1\}#2\{#3\}}
\newcommand{\atriple}[3]{\{#1\}#2\{#3\}}
\newcommand{\nm}[1]{\| #1 \|}
\newcommand{\E}[0]{\mathcal{E}}
\def \quwhile {quantum \textbf{while}}
\newcommand{\reduce}[2]{t_{#2}\left(#1\right )}
\newcommand{\eqv}[1]{\left[ #1 \right]}
\newcommand{\Seqv}[0]{\PO_{\infty}}
\newcommand{\LT}[0]{\P}

\maketitle

\input{Intro.tex}

\input{NKA.tex}

\input{QPath.tex}

\input{QInt.tex}

\input{Application.tex}

\input{NKAT.tex}

\balance
\bibliography{refs,qrefs}

\newpage

\appendix

\input{Appendix.tex}

\end{document}

%% file: Intro.tex
\section{Introduction}

\subsection{Background and Motivation}
Kleene algebra (KA)~\cite{Kleene56} that establishes the equivalence of regular expressions and finite automata is an important connection built between programming languages and abstract machines with a wide range of applications.
One very successful extension of KA, called Kleene algebra with tests (KAT), was introduced by~\citet{K97c} that combines KA with Boolean algebra (BA) to model the fundamental constructs arising in programs: sequencing, branching, iteration, etc. 
More importantly, the equational theory of KAT, which can be finitely axiomatized~\cite{KS96a}, allows \emph{algebraic reasoning} about corresponding classical programs.  

The mathematical elegance and succinctness of algebraic reasoning with KAT have furnished deep theoretical insights as well as practical tools. 
A lot of topics can be investigated with KAT including, e.g., 
program transformations~\cite{AK01a}, compiler optimization~\cite{KP00}, Hoare logic~\cite{K00a}, and so on. 
An important recent application of KAT is NetKAT~\cite{AFGJKSW13a} that reasons about the packet-forwarding behavior of software-defined networks,
with both a solid theoretical foundation~\cite{FKMST15a} and scalable practical performance~\cite{AFGJKSW13a}.
An efficient fragment of KAT, called Guarded KAT (GKAT), has also been identified~\cite{GKAT} to model typical imperative programs with an almost linear time equational theory. 
In contrast, KAT's equational theory is {\bf PSPACE}-complete~\cite{CKS96a}.

Quantum computation has been a topic of significant recent interest. 
With breakthroughs in experimental quantum computing 
and the introduction of many quantum programming languages such as Quipper \citep{Green2013}, Scaffold \citep{Sca12}, QWIRE \citep{PRZ2017}, 
Microsoft's Q\#~\citep{Svore:2018}, IBM's Qiskit \cite{aleksandrowicz2019qiskit}, Google's Cirq \cite{Cirq18}, Rigetti's Forest \cite{ForestRig},
there is an imperative need for the analysis and verification of quantum programs. 

Indeed, program analysis and verification have been a central topic ever since the seminal work on quantum programming languages~\citep{Om03,SZ00, Sabry-Haskel, Se04, AG05}. 
There have been many attempts of developing Hoare-like logic~\cite{Hoare:1969} for verification of quantum programs \citep{BJ04,CHADHA200619, Baltag2011, Feng:2007, Kaku09, Yin11}. 
In particular, \citet{DP2006} proposed the notions of quantum predicate and weakest precondition. \citet{Yin11} established the quantum Hoare logic with (relative) completeness for reasoning about a quantum extension of the \textbf{while}-language with many subsequent  developments \cite{YYW17,Li:2017,ZYY19}. 
We refer curious readers to surveys~\cite{Selinger04, Gay:2006, Ying2019} for details. 

Quantum \textbf{while}-programs have similar (yet semantically different) fundamental constructs (e.g., sequencing, branching, iterations) like classical ones, 
which gives rise to a natural question of the possibility of using KA/KAT to algebraically reason about quantum programs. 
Existing methods for quantum program analysis and verification usually involve exponential-size matrices in terms of the system size, which hence significantly limits the scalability. 
In contrast, a succinct KA-based algebraic reasoning, if possible, would greatly increase the scalability of such analyses for quantum programs due to its mathematical succinctness. 

\subsection{Research Challenges and Solutions}

Let us first revisit KAT-based algebraic reasoning and highlight the challenges in extending the framework to the quantum setting. We assume a few self-explanatory quantum notations with detailed quantum preliminaries in \sec{qInt:prelim}. 

\vspace{1mm} \noindent \textbf{KAT-based Reasoning.} A typical reasoning framework based on KAT, similarly for NetKAT and GKAT, will establish that KAT models the targeted computation by showing 
\begin{align} \label{eqn:KAT}
        \KATmodel e=f \quad \Leftrightarrow \quad \forall \intp, \Kint(e)=\Kint(f),
\end{align}
where $\Kint$ is an interpretation mapping from expressions to a language (or semantic) model of the desired computation. 
In reasoning about while programs, one encodes them as KAT expressions as in Propositional Dynamic Logic~\cite{FISCHER1979194}: 
\begin{align}
    p;q &:= pq \label{eqn:int-seq}\\
    \cif{b}{p}{q} &:= bp+\overline{b}q \label{eqn:cbranch}\\
    \cwhile{b}{p} &:= (bp)^*\overline{b} \label{eqn:int-while},
\end{align}
where $b$ is a classical guard/test and $\overline{b}$ is its Boolean negation. 

Intuitively, if one can derive the equivalence of encodings of two classical programs in KAT, then through the soundness direction ($\Rightarrow$), one can also establish the equivalence between the semantics of the original programs by applying an appropriate interpretation.

\vspace{1mm} \noindent \textbf{Quantum Branching.} 
One \emph{critical} difference between quantum and classical programs lies in the \emph{branching} statement. 
The quantum branching statement,  
\begin{align}
    \qif{M[q]\rightarrow^{i} P_i},
\end{align}
refers to a \emph{probabilistic} procedure to execute branch $P_i$ depending on the outcome of quantum measurement $M$ on quantum variable $q$ ({of which the state is} denoted by a density operator $\rho$). 
Consider the two-branching case ($i$=0,1), and let $M=\{M_0, M_1\}$ be the quantum measurement operators. 
Measurement $M$ will \emph{collapse} $\rho$ to the state $\rho_0= M_0\rho M_0^\dagger/\Tr(M_0\rho M_0^\dagger)$ with probability $p_0=\Tr(M_0\rho M_0^\dagger)$, and  the state $\rho_1=M_1\rho M_1^\dagger/\Tr(M_1\rho M_1^\dagger)$ with probability $p_1= \Tr(M_1\rho M_1^\dagger)$ respectively (here $\Tr(\cdot)$ is the matrix trace). 
After the measurement $M$, the program will execute $P_i$ on state $\rho_i$ with probability $p_i$ $(i=0,1)$. 

There are two important differences between quantum and classical branching. The \emph{first} is that quantum branching allows probabilistic choices over different branches. 
Even though random choices also appear in probabilistic programs, the probabilistic choices in quantum branching are due to quantum mechanics (i.e., measurements). 
In particular, their distributions are determined by the underlying quantum states and the corresponding quantum measurements, and hence \emph{implicit} in the syntax of quantum programs, whereas specific probabilities are usually \emph{explicitly} encoded in the syntax of probabilistic programs. 
Moreover, different quantum measurements do not necessarily commute with each other, which could hence lead to more complex probability distributions in quantum branching than ones allowed in classical probability theory and hence probabilistic programs.  

The \emph{second} difference lies in the different roles played by classical guards and quantum measurements in branching. 
Note that classical guards serve two functionalities simultaneously: (1) first, their values are used to choose the branches before the control; (2) second, they can also be deemed as property tests (i.e. logical  propositions) on the state of the program after the control but before executing each branch. These two points might be so natural that one tends to forget that they are based on \emph{an assumption that observing the guard won't change the state of the program}, which is also naturally held classically. 
The classical  guards, when deemed as tests in KAT, enjoy further the Boolean algebraic properties so that they can be conveniently manipulated. 

This natural assumption, however, fails to hold in quantum branching since  quantum measurements will change underlying states in the branching statement. 
This is mathematically evident as we see $\rho$ is collapsed to either $\rho_0$ or $\rho_1$ for different branches. 
Therefore, it is conceivable that quantum branching (and hence quantum programs) should refer to a different semantic model and quantum measurements should be deemed different from the tests in KAT. 

\vspace{1mm} \noindent \textbf{Issues with directly adopting KAT/KA.} Aforementioned differences make it hard to directly work with KAT/KA for quantum programs. 
First, there is a well-known issue when combining 
non-determinism, which is native to KAT, with probabilistic choices~\cite{MISLOVE2006,varacca_winskel_2006}, 
the latter of which is however essential in quantum branching. 
A similar issue also showed up in the probabilistic extension of NetKAT~\cite{FKMRS15a}, which does not satisfy all the KAT rules, especially the \emph{idempotent} law. 
One might wonder about the possibility of using GKAT~\cite{GKAT}, which is designed to mitigate this issue by restricting KAT with guarded structures.
Unfortunately, the classical guarded structure modeled in GKAT is semantically different from quantum branching, which makes it hard to connect GKAT with appropriate quantum models. 

\vspace{1mm} \noindent \textbf{Solution with NKA and NKAT.}
Our strategy is to work with the variant of KA without the idempotent law, namely, the \emph{non-idempotent Kleene algebra} (NKA).
This change will help model the probabilistic nature of quantum programs in a natural way, however, at the cost of losing properties implied by the idempotent law. 
Fortunately, thanks to the existing research on NKA \cite{kuich2012semirings, esik2004inductive}, many properties of KA are recovered in NKA for its applications to quantum programs.

Since there is no single "test" in quantum programs that can serve two purposes like classical guards, we simply separate the treatments for them. 
The branching functionality of quantum measurements can hence be expressed in NKA by treating them as normal program statements. Precisely, any quantum two-branching can be encoded as 
\begin{align} \label{eqn:qbranch}
    m_0 p_0 + m_1 p_1, 
\end{align}
where $m_{0/1}$ are encodings of measurements and $p_{0/1}$ are encodings of programs in each branch. 
Comparing with the classical encoding \eqref{eqn:cbranch}, $m_{0/1}$ no longer enjoy the Boolean algebraic properties and should be treated separately. 

It turns out that many classical applications of KAT such as compiler optimization~\cite{KP00} and the proof of the normal form of \textbf{while}-programs~\cite{K97c} can be implemented in NKA for quantum programs with branching functionality only. 

However, one needs to extend NKA to recover other applications of KAT which makes essential use of the proposition functionality of tests. 
A prominent example in KAT is its application to propositional Hoare logic~\cite{K00a}. 
Indeed, a typical Hoare triple $\{b\} p \{c\}$ asserts that whenever $b$ holds before the execution of the program $p$, then if and when $p$ halts, $c$ will hold of the output state,  where $b, c$ are both tests in KAT leveraging their proposition functionality. 

A similar triple $\{A\} P \{B\}$ is also used in quantum Hoare logic~\cite{Yin11}, where $P$ is the quantum program and $A, B$ become \emph{quantum predicates}~\cite{DP2006}.
To encode quantum Hoare logic, we extend NKA with the "test", denoted NKAT, which mimics the behavior of quantum predicates following the effect algebra~\cite{foulis1994effect}.
With quantum predicates, we develop a more delicate description of measurements in quantum branching, called \emph{partitions}, which  allow us to reason about the relationship among quantum branches caused by the same quantum measurement, e.g., the $m_0$ and $m_1$ branches in \eqref{eqn:qbranch}. 

\vspace{1mm} \noindent \textbf{Quantum Path Model.} One of our main technical contributions is the identification of the so-called \emph{quantum path model}, a complete and sound semantic model for NKA. Namely, 
\begin{align}
        \NKAmodel e=f \quad \Leftrightarrow \quad \forall \intp, \Qint(e)=\Qint(f),
\end{align}
where $\Qint$ is an interpretation mapping from NKA expressions to quantum path actions, which can be deemed as quantum evolution in the path integral formulation of quantum mechanics. 
$\Qint$ will connect the NKA encoding of any quantum  program $P$ with its denotational semantics $\sem{P}$.
\footnote{Since we relate NKA to quantum models which imply the probabilistic feature inherently, there is no need to explicitly add probability to NKA.}

The key motivation of the quantum path model is to address the infinity issue in NKA.  
For an intuitive understanding, one can deem any KA or NKA expression as a collection of potentially infinitely many traces, where "infinitely many" is caused by $*$ operations. 
In the case of KA, by the idempotent law, every single trace is either in or out of the collection. 
However, in the case of NKA, each trace is associated with a  weight, which by itself could be infinite.
To distinguish between nonequivalent NKA expressions, one needs to build a semantic model that can characterize a collection of weighted traces with potentially infinite weights. 
We also require the quantum nature of this semantic model for connection with the denotational semantics of quantum programs. 

The path integral formulation becomes very natural in this regard: it formulates quantum evolution as the accumulative effect of a collection of evolutions on individual trajectories. 
Our quantum path model basically characterizes the accumulative quantum evolution over a collection of potentially infinite evolutions over individual traces.
By identifying quantum path actions representing quantum predicates and quantum measurements in the quantum path model, a soundness theorem is proved for NKAT as well. 

\vspace{1mm} \noindent \textbf{Quantum-Classical differences as exhibited in NKA and NKAT.} 
The quantum-classical difference is not explicit in the syntax of NKA, as there is no special symbol for quantum measurements.
This is also reflected in the proof of the completeness of NKA where an interpretation of essentially classical probabilistic processes is constructed (Remark~\ref{rm:completeness_NKA}). 

However, the difference becomes explicit in NKAT: 
the two functionalities of the quantum guards are characterized separately by \emph{effects} and \emph{partitions}, in contrast with the classical guards in KAT.
The general noncommutativity of quantum measurements in NKAT demonstrates its quantumness and distinguishes itself from any classical model.

\begin{figure*}
\scalebox{0.9}{
    \centering
\begin{tikzpicture}[shorten >=1pt,on grid,very thick,
mainnode/.style={rectangle, rounded corners, draw=black, very thick, minimum size=5mm, align=center},
subnode/.style={rectangle, rounded corners, draw=blue, very thick, minimum size=7mm},
]
\node[mainnode] (NKA) {Non-idempotent Kleene algebra \\ \sec{NKA}};
\node[mainnode] (QPM) [below left=2cm and 3 cm of NKA] {Quantum path model \\ \sec{path_model}};
\node[mainnode] (QP)  [below right=2cm and 3 cm of NKA] {Quantum program \\ \sec{qInt}};

\draw[<->] (QPM) -- (NKA) node[midway,left=0.3cm, align=center] {\emph{sound and complete} \\ \thm{q-model-sound-complete}} ;
\draw[->] (NKA) -- (QP) node[midway,right=0.4cm, align=center] {\emph{encode} \\ \thm{enc-recovery}} ;
\draw[->] (QP) -- (QPM) node[midway,below, align=center] {\emph{embedded in} \\ \lem{PL-embed}} ;

\node (App) [right=of NKA.north, xshift=4.5cm, yshift=0.5cm]{\textbf{Applications}};
\draw [draw=black] (App.north west) rectangle +(5.8cm, -4.2cm);
\draw[->] (NKA.east) -- (NKA.east-|App.west);

\node[mainnode] (Comp) [below=0.8cm of App.west, anchor=west, xshift=2mm, align=center] {Compiler optimizations rules \\ \sec{app:optimization}};
\node[mainnode] (NF) [below=of Comp.west, anchor=west, align=center, yshift=-0.2cm] {Normal form theorem \\ \sec{app:normal}};
\node[mainnode] (NKAT) [below=of NF.west, anchor=west, align=center, yshift=-0.2cm] {NKAT \\ \sec{NKAT:NKAT}};
\node[mainnode] (pQHL) [right=of NKAT, align=center, xshift=2cm] {Propositional QHL \\ \sec{NKAT:pqhl}};
\draw[->] (NKAT) -- (pQHL)  ;

\end{tikzpicture}
}
    \caption{The structure and main results of this paper. }
    \label{fig:flowchart}
\end{figure*}
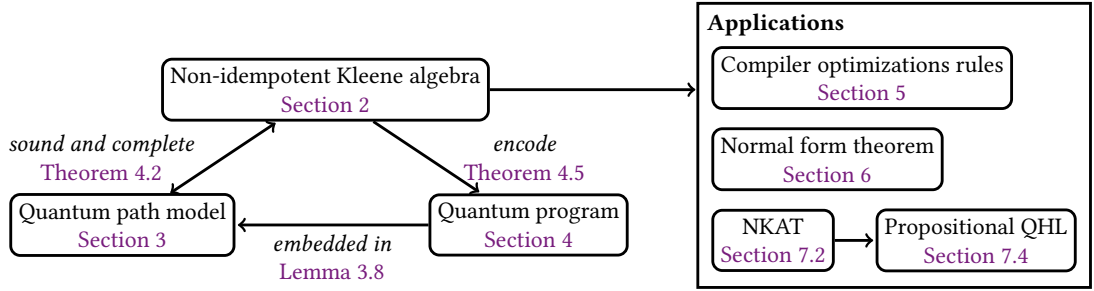

\vspace{-2mm}
\subsection{Contributions} 
To our best knowledge, we contribute the first investigation of Kleene-like algebraic reasoning of quantum programs and demonstrate its feasibility. 
We introduce the non-idempotent Kleene algebra (NKA) and existing results on the semantic model of NKA in  \sec{NKA}.
Our contributions include: 

\begin{itemize} [leftmargin=4mm]
    \item We illustrate the quantum path model and its relation with normal quantum superoperators in \sec{path_model}. 
    \item 
We prove that the NKA axioms are sound and complete with respect to the quantum path model, given encodings of quantum programs in NKA and an appropriate  interpretation of NKA to the quantum path model in \sec{qInt}. 
\item We demonstrate several applications of NKA for quantum programs, including: (1) the verification of optimization in quantum compilers
(\sec{app:optimization}); (2) an algebraic equational proof of the quantum counterpart of the classic B\"{o}hm-Jacopini theorem~\cite{BJ66} 
(\sec{app:normal}). 
\item  We extend NKA with the effect algebra to obtain the Non-idempotent Kleene Algebra with Tests (NKAT), which is proven sound for the quantum path model. 
We also encode the entire propositional quantum Hoare logic as NKAT theorems in light of \citet{K00a} (\sec{NKAT}). 
\end{itemize}

\vspace{1mm} \noindent \textbf{Main Theorem.} 
Our main theorem presented below formally guarantees that quantum program equivalences are implied if we can algebraically derive the corresponding NKA theorems. This approach is similar to deriving classical program equivalence via KAT. 
\begin{theorem}
\label{thm:main}
Given two quantum programs $P, Q$ and sub-program pairs $\{\langle P_i, Q_i\rangle\}$ where $\sem{P_i}=\sem{Q_i}$, if Horn theorem
\begin{align*}
    &\NKAmodel \left(\bigwedge\nolimits_i \enc(P_i)=\enc(Q_i)\right)\rightarrow \enc(P)=\enc(Q)
\end{align*}
is derivable, then we have $\sem{P}=\sem{Q}.$
Here $\enc$ is the encoding of quantum program in a similar manner of \eqref{eqn:int-seq}-\eqref{eqn:int-while}.
\end{theorem}
We display the essential concepts leading to this theorem in \fig{flowchart}, illustrating how our efforts in later sections connect to it, and its applications and extensions. 

\vspace{2mm} \noindent \textbf{Related Works.} It is worthwhile comparing quantum algebraic reasoning based on NKA with other techniques on quantum program analysis, e.g., quantum Hoare logic~\cite{Yin11}. As we see, classical algebraic reasoning is extremely good at certain tasks (e.g, equational proofs). However, since it abstracts away a lot of semantic information, it cannot tell about detailed specifications on the state of programs, which can otherwise be reasoned by Hoare logic~\cite{Hoare:1969}.

Our quantum algebraic reasoning inherits the advantages and disadvantages of its classical counterpart. 
It allows elegant applications in Section~\ref{sec:app:optimization} \& \ref{sec:app:normal}, which is very hard (e.g., involving exponential-size matrices) to solve with the quantum Hoare logic~\cite{Yin11} or its relational variants~\cite{Unruh-POPL-19,popl20-relational}. 
However, it cannot replace quantum Hoare logic to reason about, e.g., specifications on the state of quantum programs either. 

A recent result of quantum abstract interpretation~\cite{YuPalsberg21}
contributes to another promising approach to verifying quantum assertions with succinct proofs, although its applicability and technique are incomparable to ours. 

There are many other verification tools developed for quantum programs. \citet{hietala2019verified} built VOQC, an infrastructure for quantum circuits in Coq with numerous verified programs and compiler optimization rules. Another theory for equational reasoning of quantum circuits is introduced in \cite{staton2015algebraic}. They serve as good complements of our framework when loops are absent. 

\vspace{2mm} \noindent \textbf{Future Directions.}
One interesting question is the automation related to NKA, e.g., through co-algebra and bi-simulation techniques, in light of~\cite{K17a,GKAT,BONCHI201277, silva2010kleene}. 
This could lead to efficient symbolic algorithms for algebraic reasoning of quantum programs in light of~\cite{Pous15}. \citet{kiefer2013complexity} proposed an algorithm checking $\mathbb{Q}-$weighted automata equivalences, which works for NKA when no infinity presents.

Another direction is to include quantum-specific rules to NKA to ease the expression of practical quantum applications. For example, one may embed unitary superoperators into NKA as a group to encode their reversibility.

Given the promising applications of KAT in network programming (e.g., NetKAT~\cite{FKMRS15a}), 
an exciting opportunity is to investigate the possibility of a quantum version of NetKAT in the software-defined model of the emerging quantum internet (e.g., \cite{q-internet,qnetwork}) based on our work. 

%% file: NKA.tex
\section{Non-idempotent Kleene Algebra} \label{sec:NKA}

In this section, we introduce the theory of a Kleene algebraic system without the idempotent law, which is called non-idempotent Kleene algebra (NKA). 

\begin{figure*}[t]
\scalebox{0.95}{
  \hspace{-8mm}
  \begin{subfigure}{.7\linewidth}
  \begin{minipage}{.45\linewidth}
  \begin{align*}
    \label{fp-rule} &1+pp^*=1+p^*p=p^*  \tag{fixed-point} \\
    \label{star-mono} &p\leq q\rightarrow p^*\leq q^* \tag{monotone-star} \\
    \label{product-star-rule} &1+p(qp)^*q=(pq)^* \tag{product-star}
  \end{align*}
  \end{minipage}
  ~~
  \begin{minipage}{.53\linewidth}
  \begin{align*}
    \label{sliding-rule} &(pq)^*p=p(qp)^* \tag{sliding} \\
    \label{denesting-rule} &(p+q)^*=(p^*q)^*p^*=p^*(qp^*)^* \tag{denesting} \\
    \label{pos-rule} &0\leq p \tag{positivity}
  \end{align*}
  \end{minipage}
  \caption{Commonly used theorems of NKA}
  \label{fig:NKAtheorems}
  \end{subfigure}
  ~~
  \begin{subfigure}{.32\linewidth}
  \centering
  \begin{align*}
    \label{unr-rule} &(pp)^*(1+p)=p^*   \tag{unrolling} \\
    \label{comstar-rule} &pq=qp\rightarrow p^*q=qp^*  \tag{swap-star} \\
    \label{star-rew-rule} &pq=rp\rightarrow pq^*=r^*p  \tag{star-rewrite}
  \end{align*}
  \caption{Several theorems of NKA for applications}
  \label{fig:optlem}
  \end{subfigure}
}
\vspace{-2mm}
  \caption{Derivable formulae in NKA.}
\end{figure*}

We inherit Kozen's axiomatization for Kleene algebra (KA) in \cite{kozen1990completeness} with several weakenings.

\begin{figure}[b]
\centering
\scalebox{0.95}{
    \begin{minipage}{.4\linewidth}
    \centering
    {\textbf{Axioms of KA}}\\
    \begin{align*}
    \hline
    &\text{\textsc{Semiring Laws}} \\
    &p+(q+r)=(p+q)+r; \\
    &p+q=q+p; \\
    &p+0=p; \\
    &p(qr)=(pq)r; \\
    &1p=p1=p; \\
    &0p=p0=0; \\
    &p(q+r)=pq+pr; \\
    &(p+q)r=pr+qr; \\
    &{\color{red} p+p=p;} \\
    &\\
    \hline
    &\text{\textsc{Partial Order Laws}} \\
    &{\color{red} p\leq q\leftrightarrow p+q=q;} \\
    & \\
    & \\
    & \\
    & \\
    & \\
    \hline
    &\text{\textsc{Star Laws}} \\
    &1+pp^*\leq p^*; \\
    &q+pr\leq r \rightarrow p^*q\leq r; \\
    &q+rp\leq r \rightarrow qp^*\leq r;
    \end{align*}
    \end{minipage}
    \qquad
    \begin{minipage}{.5\linewidth}
    \centering
    {\textbf{Axioms of NKA}}\\
    \begin{align*}
    \hline
    &\text{\textsc{Semiring Laws:}} \\
    &p+(q+r)=(p+q)+r; \\
    &p+q=q+p; \\
    &p+0=p; \\
    &p(qr)=(pq)r; \\
    &1p=p1=p; \\
    &0p=p0=0; \\
    &p(q+r)=pq+pr; \\
    &(p+q)r=pr+qr; \\
    & \\
    & \\
    \hline
    &\text{\textsc{Partial Order Laws}} \\
    &{\color{blue} p\leq p;} \\
    &{\color{blue} p\leq q\wedge q\leq p \rightarrow p=q;} \\
    &{\color{blue} p\leq q\wedge q\leq r \rightarrow p\leq r;} \\
    &{\color{blue} p\leq q\wedge r\leq s \rightarrow p+r\leq q+s;} \\
    &{\color{blue} p\leq q\wedge r\leq s \rightarrow pr\leq qs;} \\
    & \\
    \hline
    &\text{\textsc{Star Laws}} \\
    &1+pp^*\leq p^*; \\
    &q+pr\leq r \rightarrow p^*q\leq r; \\
    &q+rp\leq r \rightarrow qp^*\leq r;
    \end{align*}
    \end{minipage}
}
    \caption{Axioms of KA and NKA. Axioms marked in {\color{blue} blue} ({\color{red} red}) only present in NKA (KA).}
    \label{fig:axiom-ka-nka}
\end{figure}

\begin{definition}
  A non-idempotent Kleene algebra (NKA) is a 7-tuple $(\K, +, \cdot, *, \leq, 0, 1)$, where $+$ and $\cdot$ are binary operations, $*$ is a unary operation, and $\leq$ is a binary relation. It satisfies the axioms in \fig{axiom-ka-nka}.
\end{definition}

The most essential weakening is the deletion of the idempotent law. The partial order in KA cannot directly fit in the scenario when the  idempotent law is absent.
We hence generalize the KA partial order to any partial order that is preserved by $+$ and $\cdot$. Therefore,  $*$ also preserves this partial order. Moreover, we did not include the symmetric fixed point inequality $1+p^*p\leq p^*$ because it is derivable by other axioms, both in KA and in NKA \cite{esik2004inductive}.

\begin{definition}
For an alphabet $\Sigma$, an expression over $\Sigma$ is inductively defined by:
\begin{align}
    e::=0~|~1~|~a~|~e_1+e_2~|~e_1\cdot e_2~|~e_1^*,
\end{align}
where $a\in \Sigma$. We denote all the expressions over $\Sigma$ by $\expsig$.

A Horn formula $\phi$ is defined as the form $(\bigwedge_i e_i\leq f_i)\rightarrow e\leq f$. One may also substitute equation for inequality in $\phi$ since $e=f \leftrightarrow e\leq f \wedge f\leq e$.

We write $\NKAmodel \phi$ if $\phi$ is derivable in NKA with equational logic. Any derivable formula in NKA is a \emph{theorem} of NKA.
\end{definition}

Apparently, every theorem in NKA is derivable in KA, since the partial order in KA is monotone. The reverse direction is not true in general.
Indeed, the idempotent law, for example, is nowhere derivable from the NKA axioms. It is thus natural to ask what important theorems in KA are still derivable in NKA. We provide affirmative answers to many of them in the following. (Proofs in \app{nka-rules}.)

\begin{lemma}
The following formulae are derivable in NKA.
\label{lem:iter-star-rules}
  \begin{enumerate}
      \item The formulae in \fig{NKAtheorems} due to \cite{esik2004inductive}. 
      \item The formulae in \fig{optlem}.
  \end{enumerate}
\end{lemma}

It is known that NKA also has a natural semantic model, called rational power series, which is a special class of formal power series over $\natinf=\nat\cup\{\infty\}$. We present a brief introduction to them in \app{NKA:FPS} for interested readers.

\begin{remark}[Complexity related to NKA]
\citet{bloom2009axiomatizing} have proposed an algorithm to determine the equivalence of two rational power series, so the equational theory of NKA is decidable.
Meanwhile, a subset $1^*\K=\{1^*p ~:~ p\in\K\}$ satisfies the Kleene algebra axioms, and the equational theory of KA is \textsc{PSpace}-complete \cite{stockmeyer1973word}, thus equational theory of NKA is also \textsc{PSpace}-hard. However, by linking formal power series to weighted finite automata, \citet{eilenberg1974automata} shows that it is undecidable whether a given inequality $e\leq f$ holds in NKA. 
\end{remark}

%% file: QPath.tex
\section{Quantum Path Model} \label{sec:path_model}

To address the infinity issue, we introduce a generalization of quantum superoperators in this section, named quantum path model, a sound model of NKA.
We include detailed quantum preliminaries in \sec{qInt:prelim}, introduce extended positive operators as a generalization of quantum states in \sec{qPath:epo}, define the quantum path model as an analog of the path integral in quantum mechanics in \sec{qPath:ap}, and embed quantum superoperators in the quantum path model in \sec{qPath:embed}. 
We recommend that first-time readers skip technical construction details in this section.

\subsection{Quantum Preliminaries} \label{sec:qInt:prelim}

We review basic notations from quantum information that are used in this paper. Curious readers should refer to \cite{nielsen_chuang_2010, watrous_2018} for more details. 

An $n$-dimensional Hilbert space $\H$
is essentially the space $\mathbb{C}^n$ of complex vectors.
We use Dirac's notation, $\ket{\psi}$, to denote a complex vector in $\mathbb{C}^n$. The inner product of $\ket{\psi}$ and $\ket{\varphi}$ is denoted by $\langle\psi|\varphi\rangle$,
which is the product of the Hermitian conjugate of $\ket{\psi}$, denoted by $\bra{\psi}$, and the vector $\ket{\varphi}$.

Linear operators between $n$-dimensional Hilbert spaces are represented by $n\times n$ matrices.
For example, the zero operator $O_{\H}$ and the identity operator $I_\H$ can be identified by the zero matrix and the identity matrix on $\H$.
The Hermitian conjugate of operator $A$ is denoted by $A^\dag$.
Operator $A$ is \emph{positive semidefinite} if for all vectors $\ket{\psi}\in\H$, $\bra{\psi}A\ket{\psi}\geq 0$.
The set of positive semidefinite operators over $\H$ is denoted by $\PO(\H).$
This gives rise to the \emph{L\"owner order} $\sqsubseteq$ among operators: $A\sqsubseteq B$ $\Leftrightarrow$ $B-A$ is positive semidefinite.

A \emph{density operator} $\rho$ is a positive semidefinite operator $\rho=\sum_i p_i\proj{\psi_i}$ where $\sum_i p_i=1, p_i>0$.
A special case $\rho=\proj{\psi}$ is conventionally denoted as $\ket{\psi}$.
A positive semidefinite operator $\rho$ on $\H$ is a \emph{partial} density operator if $\tr(\rho)\leq 1,$ where $\tr(\rho)$ is the matrix trace of $\rho$.
The set of partial density operators is denoted by $\D(\H)$.

The evolution of a quantum system can be characterized by a \emph{completely-positive} and \emph{trace-non-increasing} linear superoperator $\E$\footnote{A superoperator $\E$ is \emph{positive} if it maps from $\D(\H)$ to $\D(\H')$ for Hilbert spaces $\H, \H'$.
It is \emph{completely-positive} if for any Hilbert space $\A$, the superoperator $\E\otimes I_\A$ is positive.
It is \emph{trace-non-increasing} if for any initial state $\rho\in\D(\H)$, the final state $\E(\rho)\in \D(\H')$ satisfies $\tr(\E(\rho))\leq\tr(\rho)$.}, which is a mapping from $\D(\H)$ to $\D(\H')$ for Hilbert spaces $\H, \H'$.
We denote the set of such superoperators by $\CP(\H, \H').$
The special case when $\H'=\H$ is denoted by $\CP(\H)$.

For two superoperators $\E_1, \E_2\in\CP(\H),$ the composition is defined as $(\E_1\circ\E_2)(\rho)=\E_2(\E_1(\rho)).$ 
If there exists $\E$ and $\E_i\in\CP(\H)$ satisfying $\E(\rho)=\sum_i\E_i(\rho)$ for every $\rho\in\PO(\H),$ then we define $\E$ as $\sum_i\E_i$.
For every superoperator $\E\in\CP(\H, \H')$, by \cite{kraus1983states} there exists a set of Kraus operators $\{E_k\}_k$ such that $\E(\rho)=\sum_k E_k\rho E_k^\dag$ for any input $\rho\in\D(\H)$.
The Schr\"odinger-Heisenberg \emph{dual} of a superoperator $\E(\rho)=\sum_k E_k\rho E_k^\dag$ is $\E^\dag(\rho)=\sum_k E_k^\dag\rho E_k$.

A quantum \emph{measurement} on a system over Hilbert space $\H$ can be described by a set of linear operators $\{M_m\}_m$ where $\sum_m M_m^\dag M_m=I_\H$. 
The measurement outcome $m$ is observed with probability $p_m=\tr(M_m\rho M_m^\dag)$ for each $m$, which will collapse the pre-measure state $\rho$ to $\M_m(\rho)=M_m\rho M_m^\dag/p_m$.
A quantum measurement is \emph{projective} if $M_iM_j=M_i$ if $i=j$ and $O_{\H}$ otherwise. Namely, all $M_i$ are projective operators orthogonal to each other. 

\subsection{Extended Positive Operators} \label{sec:qPath:epo}

The set $\PO(\H)$ does not contain any infinity. We need to incorporate different infinities into it to distinguish different path sets which may lead to different divergent summations.

\begin{definition}
  A series of $\PO(\H)$ is a countable multiset of $\PO(\H)$, and can be written as $\biguplus_{i\in I} \rho_i,$ where $I$ is a countable index set. Symbol $\biguplus_{i\in I}$ enumerates every element $\rho_i$ in the multiset. The set of series of $\PO(\H)$ is denoted by $\S(\H)$.
  
  The \emph{union} of countably many series is denoted by:
  \begin{align}
    \biguplus\nolimits_{i\in I} \left(\biguplus\nolimits_{j\in J_i}\rho_{ij}\right)=\biguplus\nolimits_{(i, j):i\in I, j\in J_i}\rho_{ij}.
  \end{align}
  Note $\biguplus_{i\in I} \biguplus_{j\in J_i}\rho_{ij}\in\S(\H)$ since the index set is countable.
  
  A \emph{binary relation}  $\lesssim$ over $\S(\H)$ is defined by: $\biguplus_{i\in I}\rho_i \lesssim \biguplus_{j\in J}\sigma_j$ if and only if for every $\epsilon>0$ and finite $I'\subseteq I$, there exists a finite $J'\subseteq J$, such that
  \begin{align}
      \label{eq:pre-def}
      \sum\nolimits_{i\in I'}\rho_i\sqsubseteq \epsilon I_{\H}+\sum\nolimits_{j\in J'}\sigma_j.
  \end{align}
  
  We induce another binary relation $\sim$ from $\lesssim$ on $\S(\H)$ by:
  \begin{align*}
      \biguplus_{i\in I} \rho_i \sim \biguplus_{j\in J}\sigma_j \quad\Leftrightarrow\quad \biguplus_{i\in I}\rho_i \lesssim \biguplus_{j\in J}\sigma_j ~\wedge~ \biguplus_{j\in J}\sigma_j \lesssim \biguplus_{i\in I}\rho_i.
  \end{align*}
\end{definition}

Symbol $\biguplus_{i\in I}$ is employed to distinguish the series from the normal summation $\sum_{i\in I}$ over $\PO(\H)$. 
We will build connections between these two notions so that $\biguplus_{i\in I}$ can readily help us in the analysis of convergence, and more. 

We represent a finite series by enumerating its elements.
Like a series with one element $O_{\H}$, we denote it by $\{\!|O_{\H}|\!\}$. 

The definition of $\lesssim$ aims at a generalization to the L\"owner order in $\S(\H)$ that distinguishes the different infinities while preserving relations like $\{\!|I_{\H}|\!\}\lesssim \biguplus_{i>0} \frac{1}{2^i}I_{\H},$ whose correspondence in $\PO(\H)$ holds.

\begin{lemma}
  \label{lem:SHequiv}
  $\lesssim$ is a preorder, so $\sim$ is an equivalence relation.
\end{lemma}

The proof of this lemma along with several basic facts about $\S(\H)$ is in \app{SHfacts}.

\begin{definition}
  We define the extended positive operators $\Seqv(\H)=\S(\H)/\sim$ as the set of equivalence classes of $\sim$. Let the equivalence class including $\biguplus_{i\in I}\rho_i$ be
  \begin{align}
      \eqv{\biguplus\nolimits_{i\in I}\rho_i}=\left\{\biguplus\nolimits_{j\in J}\sigma_j ~\vline~ \biguplus\nolimits_{j\in J}\sigma_j\sim\biguplus\nolimits_{i\in I}\rho_i\right\},
  \end{align}
  where on the right hand side is a set of series.
  
  A partial order $\leq$ over $\Seqv(\H)$ is induced from the preorder $\lesssim$ over $\S(\H)$ by:
  \begin{align}
      \eqv{\biguplus\nolimits_{i\in I}\rho_i}\leq\eqv{\biguplus\nolimits_{j\in J}\sigma_j} ~~\Leftrightarrow~~ \biguplus\nolimits_{i\in I}\rho_i\lesssim\biguplus\nolimits_{j\in J}\sigma_j.
  \end{align}
  
  We define countable summation over $\Seqv(\H)$ from the union in $\S(\H)$ by
  \begin{align} \label{eqn:summation}
      \sum\nolimits_{i\in I} \eqv{\biguplus\nolimits_{j\in J_i}\rho_{ij}}=\eqv{\biguplus\nolimits_{i\in I} \biguplus\nolimits_{j\in J_i}\rho_{ij}}.
  \end{align}
\end{definition}

The summation defined above is independent of the choices of $\biguplus_{j\in J_i}\rho_{ij}$ because of \lem{S-biguplus-order}.

We slightly abuse notation, writing $\eqv{\rho}$ to represent $\eqv{\{\!|\rho|\!\}}$ for $\rho\in\PO(\H)$. A frequently used case of \eqref{eqn:summation} is to write the equivalence class of a series as  
\begin{align}
    \eqv{\biguplus\nolimits_{i\in I}\rho_i}=\sum\nolimits_{i\in I}\eqv{\rho_i},
\end{align}
where we can intuitively deem the countable summation over $\Seqv(\H)$ as a generalized summation over $\PO(\H).$ 
For example, we have $\sum_{i>0}\left[\frac{1}{2^i}I_{\H}\right]=\left[\sum_{i>0}\frac{1}{2^i}I_{\H}\right]=[I_{\H}]$ according to \lem{convergent-sum}.

\begin{remark}
$\PO(\H)$ is embedded in $\Seqv(\H)$ by $\rho\mapsto[\rho]$ as finite positive operators. Besides these, $\Seqv(\H)$ contains distinguishable divergent summations unattainable by $\PO(\H)$: e.g., $\sum_{i>0}[\ket{0}\bra{0}]$ is different from $\sum_{i>0}[\ket{1}\bra{1}]$, and less than $\sum_{i>0}[I_{\H_2}].$ These divergent summations are leveraged to depict the domain and the range of our extended quantum superoperators.
\end{remark}

\subsection{Quantum Actions}\label{sec:qPath:ap}

We are now ready to introduce quantum actions, a generalization of superoperators in the quantum path model, inspired by the path integral formulation of quantum mechanics. 

\begin{definition}
  A \emph{quantum action}, or \emph{action} for simplicity, over $\Seqv(\H)$ is a mapping from $\Seqv(\H)$ to $\Seqv(\H)$.
  
  A quantum action $\T$ is \emph{linear} if for series $\sum_{j\in J_i}\eqv{\rho_{ij}},$
  \begin{align}
      \T\left(\sum\nolimits_{i\in I} \sum\nolimits_{j\in J_i}\eqv{\rho_{ij}}\right)=\sum\nolimits_{i\in I}\T\left(\sum\nolimits_{j\in J_i}\eqv{\rho_{ij}}\right).
  \end{align}
  
  A quantum action $\T$ is \emph{monotone} if for any two series $\sum_{i\in I}\eqv{\rho_{i}}\leq\sum_{j\in J}\eqv{\sigma_j},$
  \begin{align}
      \T\left(\sum\nolimits_{i\in I}\eqv{\rho_i}\right)\leq\T\left(\sum\nolimits_{j\in J}\eqv{\sigma_j}\right).
  \end{align}
  
  We denote the set of linear and monotone quantum actions over $\Seqv(\H)$ by $\LT(\H)$ as the set of \emph{quantum path actions}. 
  
  The zero action $\O_{\H}$ maps everything to $\eqv{O_{\H}},$ and the identity action is denoted by $\I_{\H}.$ 
\end{definition}

A physical interpretation of quantum path actions in $\LT(\H)$ is the collection of quantum evolution along a single or many possible trajectories of the underlying system. 
Thus, one can readily define the composition and the sum of quantum path actions, as the concatenation and the union of trajectories.  

\begin{definition}
  We define the operations in $\LT(\H)$ by:
  \begin{align}
      \left(\sum_{i\in I} \T_i\right)\left(\sum_{j\in J}\eqv{\rho_j}\right)&=\sum_{i\in I} \T_i\left(\sum_{j\in J}\eqv{\rho_j}\right), \\
      (\T_1; \T_2)\left(\sum\nolimits_{j\in J}\eqv{\rho_j}\right)&=\T_2\left(\T_1\left(\sum\nolimits_{j\in J}\eqv{\rho_j}\right)\right), \\
      \T^*&=\sum\nolimits_{i\geq 0} \T^i.
  \end{align}
  Here $\T^i=\I_{\H}; \T; \T; \cdots; \T$ where $\T$ repeats $i$ times. 
  
  Additionally, we define $\T_1\diamond \T_2=\T_2; \T_1.$
  
  A point-wise partial order $\preceq$ in $\LT(\H)$ is induced point-wisely: $\T_1\preceq \T_2$ if and only if
  \begin{align}
       \forall \sum\nolimits_{i\in I}\eqv{\rho_i}, ~\T_1\left(\sum\nolimits_{i\in I}\eqv{\rho_i}\right)\leq\T_2\left(\sum\nolimits_{i\in I}\eqv{\rho_i}\right).
  \end{align}
\end{definition}

Our main result is that $\LT(\H)$ with the above partial order and operations satisfies the axioms of NKA. 
The proof is postponed to \app{PH-thms}.
Since infinite summations are well-defined over quantum path actions, any NKA derivation safely induces a derivation over quantum path actions without worrying about the infinity issue.

\begin{theorem}
  \label{thm:QS-NKA}
  The NKA axioms are sound for the \emph{quantum path model}, defined by $(\LT(\H), +, ~;, *, \preceq, \O_{\H}, \I_{\H})$. Here $+$ is the $\sum_i$ operation restricted on two operands.
\end{theorem}

\subsection{Embedding of $\CP(\H)$ in $\LT(\H)$}\label{sec:qPath:embed}

We mentioned the intuition that quantum path actions are generalizations of quantum superoperators in the quantum path model.  
We now make it precise by building an embedding from quantum superoperators to quantum path actions (and hence the quantum path model),
which allows us to prove superoperator equations via NKA theorems.

\begin{definition}
  \emph{Path lifting} is a mapping from $\E\in\CP(\H)$ to a quantum path action $\cp{\E}:\sum_{i\in I}\eqv{\rho_i}\mapsto\sum_{i\in I}\eqv{\E(\rho_i)}.$
\end{definition}

$\cp{\E}$ is well-defined (it does not depend on the choices of $\sum_{i\in I}[\rho_i]$) because of \lem{E-preorder}.

The path lifting embeds $\CP(\H)$ in $\LT(\H)$ by the following lemma, whose proof is routine and in \app{lift-embed}.

\begin{lemma}
\label{lem:PL-embed}
  The path lifting has the following properties:
  \begin{enumerate}[label=(\roman*), ref={\ref{lem:PL-embed}.(\roman*)}]
      \item $\cp{\E}\in \LT(\H),$ for $\E\in\CP(\H).$ 
      \item \label{lem:E-embedding}
      $\E_1=\E_2 \Leftrightarrow \cp{\E_1}=\cp{\E_2},$ for $\E_1, \E_2\in\CP(\H).$
      \item \label{lem:struct-preserving}
  operations $\circ$ and $\sum_i$ (when defined) in $\CP(\H)$ are preserved by path lifting as $;$ and $\sum_i$ operations in $\LT(\H)$.
  \end{enumerate}
\end{lemma}

%% file: QInt.tex
\section{Quantum Interpretation and Quantum Programs} \label{sec:qInt}

In this section, we link expressions, quantum path actions and quantum programs by quantum interpretation (\sec{qInt:interpretation}) and encoding (\sec{qInt:encoding}). 

\subsection{Quantum Interpretation} \label{sec:qInt:interpretation}

We endow equations in NKA with quantum interpretations.

\begin{definition} 
  A \emph{quantum interpretation setting} over an alphabet $\Sigma$ is a pair $\intp=(\H, \eval)$ where
  \begin{enumerate}
  \item $\H$ is a finite dimensional Hilbert space.
  \item $\eval:\Sigma \rightarrow \CP(\H)$ is a function to interpret symbols.
  \end{enumerate}
  
  The \emph{quantum interpretation} $\Qint$ w.r.t. a quantum interpretation setting $\intp$ is a mapping from $\expsig$ to $\LT(\H)$ where
  \begin{align*}
      &\Qint(0)=\O_{\H}, & \Qint(e+f)&=\Qint(e) + \Qint(f), \\
      &\Qint(1)=\I_{\H}, & \Qint(e\cdot f)&=\Qint(e); \Qint(f), \\
      &\Qint(a)=\cp{\eval(a)}, & \Qint(e^*)&=\Qint(e)^*.
  \end{align*}
  Here $a\in\Sigma,$ and $\cp{\eval(a)}$ is the path lifting of $\eval(a)$.
\end{definition}

\begin{theorem}
\label{thm:q-model-sound-complete}
  The axioms of NKA are sound and complete w.r.t. the quantum interpretation. That is, for any $e, f\in\expsig$,
  \begin{align}
    \NKAmodel e=f \quad \Leftrightarrow \quad \forall \intp, \Qint(e)=\Qint(f).
  \end{align}
\end{theorem}

The soundness comes directly from \thm{QS-NKA}. The completeness proof makes use of formal power series and is postponed to \app{q-model-sc}.
This result indicates that equations of NKA are all possible tautologies when atomic symbols are interpreted as any (lifted) quantum superoperator. These equations and interpretations do not necessarily correspond to quantum programs, so further exploitation of algebraic structures specifically for quantum programs is possible.

\begin{remark} \label{rm:completeness_NKA}
The completeness proof constructs interpretations with probabilistic processes only. It suggests that quantum processes have similar algebraic behaviors to probabilistic processes when probabilities are implicit (abstracted inside atomic operations). This is valid when measurements are not distinguished from other processes. We will discuss additional axioms for quantum measurements in \sec{NKAT}.
\end{remark}

Most of the derived rules in our applications rely on external hypotheses aside from the NKA axioms. A formula with inequalities as hypotheses is called a Horn clause. We present the relation of the Horn theorems of NKA and quantum interpretations by the following theorem. 

\begin{corollary}
\label{cor:primal-int}
  For expressions $\{e_i\}_{i=1}^n, \{f_i\}_{i=1}^n\subset\expsig$ and $e, f\in\expsig$, if
  \begin{align}
      \NKAmodel \left(\bigwedge_{i=1}^n e_i\leq f_i\right)\rightarrow e\leq f,
  \end{align}
  and $\intp=(\H, \eval)$ satisfies $\Qint(e_i)\preceq \Qint(f_i)$ for $1\leq i\leq n$, then $\Qint(e)\preceq\Qint(f).$
\end{corollary}

Note that the inequalities above can be replaced by equations, using the fact that $p=q\leftrightarrow p\leq q\wedge q\leq p$.

\begin{proof}
  The proof comes from \thm{QS-NKA} similarly. Along the derivation of $e\leq f$, we apply the NKA axioms and premises $e_i\leq f_i$ for $1\leq i\leq n$. The soundness of $e\leq f$ comes from the soundness of NKA axioms, proved in \thm{QS-NKA}, and the soundness of each premises, provided by the assumption $\Qint(e_i)\preceq \Qint(f_i)$ for each $e_i\leq f_i$.
\end{proof}

\subsection{Encoding of Quantum Programs} \label{sec:qInt:encoding}

The syntax of a \emph{\quwhile~program}, also called a program for simplicity, $P$ is defined as follows.
\footnote{The $\cskip$ statement does nothing and terminates. The $\cfail$ statement announces that the program fails, and halts the program without any result. Statement $q:=\ket{0}$ resets the register $q$ to $\ket{0}$, and $\overline{q}:=U[\overline{q}]$ applies a unitary operation on register set $\overline{q}$. These four statements' denotational semantics are called elementary superoperators. 
Note that there is no assignment statement due to the quantum no-cloning theorem~\cite{no-cloning}.
The loop $\qwhile{M[\overline{q}]=1}{P_1}$ executes repeatedly. Each time it measures $\overline{q}$ by $M$. If the measurement result is $1$, it executes $P_1$ and then starts over. Otherwise, it terminates.
When there are only two branches, we define syntax sugar $\mathbf{if~}M[\overline{q}]=1\mathbf{~then~}P_1\mathbf{~else~}P_2$ as an alternative to $\qif{M[\overline{q}]\rightarrow^{i} P_i}$. Moreover, if $P_2\equiv\cskip$, we write $\mathbf{if~}M[\overline{q}]=1\mathbf{~then~}P_1.$}
\begin{align*}
 P \enskip ::= & \enskip \cskip 
                \enskip | \enskip \cfail
                \enskip | \enskip q:=\ket{0} 
                \enskip | \enskip \overline{q}:=U[\overline{q}] \enskip | \enskip  P_1;P_2 \enskip | \enskip \\
              & \enskip \qif{M[\overline{q}]\xrightarrow{i} P_i} \enskip | \enskip \qwhile{M[\overline{q}]=1}{P_1}.
\end{align*}

The denotational semantics of $P$ is a quantum superoperator, denoted by $\sem{P}$. \citet{Ying16} proves that:
\begin{align*}
    &\sem{\cskip}(\rho)=\rho,  \quad \sem{\qif{M[\overline{q}]\xrightarrow{i} P_i}}=\sum_i\M_i\circ \sem{P_i}, \\
    &\sem{\cfail}(\rho)=O_{\H}, \quad \sem{q:=\ket{0}}(\rho)=\sum_{i}\ket{0}_q\bra{i}\rho\ket{i}_q\bra{0}, \\
    &\sem{P_1; P_2}=\sem{P_1}\circ\sem{P_2}, \quad \sem{\overline{q}:=U[\overline{q}]}(\rho)=U_{\overline{q}}\rho U_{\overline{q}}^{\dagger}, \\
    &\sem{\qwhile{M[\overline{q}]=1}{P}}=\sum_{n\geq 0} ((\M_1\circ \sem{P})^n\circ\M_0), 
\end{align*}
where for a quantum measurement $\{M_i\}_{i\in I}, \M_i$ is defined by $\M_i(\rho)=M_i\rho M_i^{\dagger}.$ Both $\circ$ and $\sum_i$ are operations over quantum superoperators. 

We formally define how to encode a quantum program as an expression, and how to recover the denotational semantics of a quantum program from an expression.

\begin{definition}
  An \emph{encoder setting} is a mapping $E$ from a finite subset of $\CP(\H)$ to $\Sigma$, that assigns a unique symbol in $\Sigma$ to the elementary superoperators (qubit resetting, unitary application, and measurement branches) in the target programs.

  The encoder $\enc$ of a program to $\expsig$ with respect to an encoder setting $E$ is defined inductively by:
  \begin{align*}
    &\enc(\cskip)=1; \qquad \ \ \enc(q:=\ket{0})=E(\sem{q:=\ket{0}}); \\
    &\enc(\cfail)=0; \qquad \enc(\overline{q}:=U[\overline{q}])=E(\sem{\overline{q}:=U[\overline{q}]});  \\
    &\enc(P_1;P_2)=\enc(P_1)\cdot \enc(P_2); \\
    &\enc(\qif{M[\overline{q}]\xrightarrow{i}P_i})=\sum_i E(\M_i)\cdot \enc(P_i); \label{eq:enc-case} \\
    &\enc(\qwhile{\!M[\overline{q}]\!=\!1\!\!}{ P \!})\!=\!(E(\M_1)\!\cdot\! \enc(P))^*\!\cdot\! E(\M_0),
  \end{align*}
  where $\Sigma_i$ in \eq{enc-case} is an abbreviation of expression summation.
\end{definition}

\begin{theorem}
  \label{thm:enc-recovery}
  For any quantum program $P$ and encoder setting $E$, let $\intp=(\H, E^{-1}),$ where $E^{-1}$ maps back the unique symbol for an elementary superoperator. Then
  \begin{align}
    \Qint(\enc(P))=\cp{\sem{P}}.
  \end{align}
\end{theorem}

A full proof by induction on $P$ is in \app{enc-recovery}.

Note that in real applications, we usually define the encoder setting $E$ jointly for multiple programs $\{P_i\}$ for technical convenience and easy comparison.

\vspace{3mm}
Now we have all the ingredients for \thm{main}.
\begin{proof}[Proof of \thm{main}]
We have constructed the quantum path model and proved it a sound model of NKA in \thm{QS-NKA}, leading to the soundness of Horn theorems by \cor{primal-int}. 
We also show an embedding of quantum superoperators into quantum path actions in \lem{E-embedding}, so Horn theorems are interpreted as quantum superoperator equivalences.
For each quantum program, we encode it with a symbolic expression whose interpretation corresponds to its denotational semantics, according to \thm{enc-recovery}.
Hence, if the NKA equivalence of quantum programs' encoding is derivable, the equivalence of their denotational semantics is induced.
\end{proof}
In the next sections, we show applications of \thm{main}.

%% file: Application.tex
\section{Validation of Quantum Compiler Optimizing Rules} \label{sec:app:optimization}

We demonstrate a few quantum compiler optimizing rules and their validation in NKA, in light of a similar application of KAT~\cite{KP00}. 
Note that many classical compiler optimizing rules do not hold or make sense in the quantum setting. 
We have carefully selected those rules with reasonable quantum counterparts, as well as quantum-specific rules found in real quantum applications. 

The validation of quantum program equivalence via NKA consists of three steps: (1) \emph{program encoding}: encode the programs as expressions over an alphabet; (2) \emph{condition formulation}: identify necessary hypotheses and construct a formula that encodes hypotheses and target equation; (3) \emph{NKA derivation}: derive the formula with the NKA axioms.

\subsection{Loop Unrolling}
\label{sec:loop-unroll}

Consider programs $\textsc{Unrolling1}$ and $\textsc{Unrolling2}$ in \fig{loop-unroll} with a program $P$ and a projective measurement $M$.

\vspace{1mm} \noindent \emph{Program Encoding}:
We encode the two programs by expressions $\enc(\textsc{Unrolling1})=(m_0p)^*m_1$ and $ \enc(\textsc{Unrolling2})=(m_0p(m_0p+m_1\cdot 1))^*m_1.$
The encoder setting is inferred easily.

\vspace{1mm} \noindent \emph{Condition Formulation}:
Because $M$ is a projective measurement, $\M_1\circ \M_1=\M_1$ and $\M_1\circ \M_0=\O_{\H}$ can be encoded by $m_1m_1=m_1$ and $m_1m_0=0.$ Their equivalence can be verified by the following formula:
\begin{align}
  \label{eq:loop-unroll}
  \NKAmodel ~~&m_1m_1=m_1\wedge m_1m_0=0\rightarrow \notag \\ &~~(m_0p)^*m_1=(m_0p(m_0p+m_1\cdot1))^*m_1.
\end{align}

\vspace{1mm} \noindent \emph{NKA Derivation}:
This formula can be derived in NKA by:
\begin{align}
  &(m_0p(m_0p+m_1\cdot1))^*m_1 \notag \\
  =~&(m_0pm_0p+m_0pm_1)^*m_1 \tag{distributive-law}\\
  =~&(m_0pm_0p)^*(m_0pm_1(m_0pm_0p)^*)^*m_1 \tag{\ref{denesting-rule}}\\
  =~&(m_0pm_0p)^*(m_0pm_1(1+m_0pm_0p(m_0pm_0p)^*))^*m_1 \tag{\ref{fp-rule}}\\
  =~&(m_0pm_0p)^*(m_0pm_1)^*m_1 \tag{$m_1m_0=0$}\\
  =~&(m_0pm_0p)^*(1+m_0pm_1(1+m_0pm_1(m_0pm_1)^*))m_1 \tag{\ref{fp-rule}}\\
  =~&(m_0pm_0p)^*(1+m_0pm_1)m_1 \tag{$m_1m_0=0$}\\
  =~&(m_0pm_0p)^*(1+m_0p)m_1 \tag{$m_1m_1=m_1$, distributive-law}\\
  =~&(m_0p)^*m_1. \tag{\ref{unr-rule}}
\end{align}

By \thm{main}, we have $\sem{\textsc{Unrolling1}}=\sem{\textsc{Unrolling2}}$.

\begin{figure}[t]
    \centering
    \scalebox{0.95}{
    \begin{subfigure}{.4\linewidth}
    \begin{align*}
    &\textsc{Unrolling1}\equiv~~\\
    &\mathbf{while~}M[q]=0\mathbf{~do}\\
    &\quad P \\
    &\mathbf{done.} \\
    \noalign{\vskip 3mm} 
    &\textsc{Boundary1}\equiv~~\\
    &\mathbf{while~} M[w]=0 \mathbf{~do~} \\
    &\quad q:=U[q]; \\
    &\quad P; \\
    &\quad q:=U^{-1}[q] \\
    &\mathbf{done}.
    \end{align*}
    \end{subfigure}
    \qquad
    \begin{subfigure}{.4\linewidth}
    \begin{align*}
    &\textsc{Unrolling2}\equiv~~\\
    &\mathbf{while~} M[q]=0 \mathbf{~do~} \\
    &\quad P; \enskip \mathbf{if~} M[q]=0 \mathbf{~then~} P \\
    &\mathbf{done}. \\
    \noalign{\vskip 3mm} 
    &\textsc{Boundary2}\equiv~~\\
    &q:=U[q]; \\
    &\mathbf{while~} M[w]=0 \mathbf{~do~} \\
    &\quad P; \\
    &\mathbf{done}; \\
    &q:=U^{-1}[q].
    \end{align*}
    \end{subfigure}
    }
    \caption{Two pairs of equivalent programs with conditions.
    }
    \label{fig:loop-unroll}
\end{figure}

\subsection{Loop Boundary}

This rule is quantum-specific because it makes use of the reversible property of quantum operations.
Consider programs $\textsc{Boundary1}$ and $\textsc{Boundary2}$ in \fig{loop-unroll}, where $P$ is an arbitrary program.
Here the unitary $U$ acting on $q$ does not affect the measurement on qubit $w$. In other words, quantum measurement $M_0$ and $M_1$ commute with $U$.

\vspace{1mm} \noindent \emph{Program Encoding}:
We encode these program by expressions $\enc(\textsc{Boundary1})=(m_0upu^{-1})^*m_1$ and $\enc(\textsc{Boundary2})=u(m_0p)^*m_1u^{-1},$
where the encoder setting $E$ can be inferred. 

\vspace{1mm} \noindent \emph{Condition Formulation}:
The reversibility property $UU^{-1}=U^{-1}U=I$ can be encoded by $uu^{-1}=u^{-1}u=1$ (at the level of quantum superoperators). Besides, the commutativity property of measurement and $U$ is encoded as $um_0=m_0u$ and $um_1=m_1u$. Then the formula we need to derive is
\begin{align}
  \label{eq:loop-boundary}
  \NKAmodel ~~& uu^{-1}=u^{-1}u=1\wedge um_0=m_0u\wedge um_1=m_1u\rightarrow \notag \\
  &\quad(m_0upu^{-1})^*m_1=u(m_0p)^*m_1u^{-1}.
\end{align}

\noindent \emph{NKA Derivation}:
The derivation of this formula in NKA is
\begin{align}
  &(m_0upu^{-1})^*m_1 \notag\\
  =~&(um_0pu^{-1})^*m_1 \tag{$um_0=m_0u$}\\
  =~&(1+u(m_0pu^{-1}u)^*m_0pu^{-1})m_1 \tag{\ref{product-star-rule}} \\
  =~&u(m_0p)^*m_1u^{-1}. \tag{$u^{-1}u=1$, \ref{fp-rule}}
\end{align}

Then $\sem{\textsc{Boundary1}}=\sem{\textsc{Boundary2}}$ by \thm{main}.

Due to space limitations, we showcase in \app{qsp} the use of the Loop Boundary rule to optimize, as observed in \cite{Childs9456}, one leading quantum Hamiltonian simulation algorithm called quantum signal processing (QSP)~\cite{low2017optimal}, as well as its algebraic verification. 

\section{Normal Form of Quantum Programs} \label{sec:app:normal}

Here we use NKA to prove a quantum counterpart of the classic B\"{o}hm-Jacopini theorem~\cite{BJ66}, namely, a normal form of quantum {\bf while} programs consisting of only a single loop. The normal form of classical programs depends on the folk operation, which copies the value of a variable to a new variable. However, in quantum programs, the no-cloning theorem prevents us from directly copying unknown states. Our approach is to store every measurement result in an augmented classical space and depends on the classical variables to manipulate the control flow of the program.
We note a quantum version of the B\"{o}hm-Jacopini theorem was recently shown in~\cite{2019arXiv190800158Y}, however, using a completely different and non-algebraic approach. 

Let us illustrate our idea with a simple example below first.
To unify the two \textbf{while} loops of \textsc{Original} into one, we redesign the control flow as in \textsc{Constructed} with a fresh classical guard variable $g\in\{0, 1, 2\}$.

\scalebox{0.85}{
\hspace{-6mm}
\begin{minipage}{.22\linewidth}
\begin{align*}
    &\textsc{Original}\equiv \\
    & \qwhile{M_1[p]=1}{P_1}; \\
    & \qwhile{M_2[p]=1}{P_2}; \\
    & g:=\ket{0}. \\
    & \\
    & \\
    & \\
    &
\end{align*}
\end{minipage}
\quad
\begin{minipage}{.26\linewidth}
\begin{align*}
    & \textsc{Constructed}\equiv \\
    &g:=\ket{1}; \\
    & \textbf{while }\text{Meas}[g]>0\textbf{ do} \\
    & \quad \textbf{if } \text{Meas}[g]>1 \textbf{ then } \\
    & \quad \quad \textbf{if } M_2[p]=1 \textbf{ then } P_2 \textbf{ else } g:=\ket{0}\\
    & \quad \textbf{else } \\
    & \quad \quad \textbf{if } M_1[p]=1 \textbf{ then } P_1 \textbf{ else } g:=\ket{2}\\
    & \textbf{done}.
\end{align*}
\end{minipage}
}
Here $\text{Meas}[g]$ is the computational basis measurement on variable $g$. When $g$ is classical, $\text{Meas}[g]$ returns the value of $g$, and does not modify $g$. The variable $g$ is used to store the measurement results and decide which branch the program executes in the next round. We prove $\sem{\textsc{Original}}=\sem{\textsc{Constructed}}$ via NKA, using the outline in \sec{app:optimization}.

\vspace{2mm} \noindent \emph{Program Encoding}:
We encode $g:=\ket{i}$ as $g^i$, and $\text{Meas}[g]>i$ as $g_{>i}$ and $g_{\leq i}$. Then the two programs are encoded as
\begin{align*}
    \enc(\textsc{Original})&=(m_{11}p_1)^*m_{10}(m_{21}p_2)^*m_{20}g^0, \\
    \enc(\textsc{Constructed})&=g^1(g_{>0}(g_{>1}(m_{21}p_2+m_{20}g^0) \\
    &\qquad \qquad +g_{\leq 1}(m_{11}p_1+m_{10}g^2)))^*g_{\leq 0}.
\end{align*}

\vspace{2mm} \noindent \emph{Condition Formulation}:
Since $g$ is fresh, operations on $g$ commutes with the quantum measurements $M_1, M_2$ and subprograms $P_1, P_2.$ This is encoded as $g^im_{jk}=m_{jk}g^i, g^ip_j=p_jg^i.$ Two consecutive assignment on $g$ will make the first one be overwritten, which is encoded as $g^ig^j=g^j.$ On top of these, $g^ig_{>j}=g^i$ if $i>j$ and $g^ig_{>j}=0$ if $i\leq j$. Similarly, $g^ig_{\leq j}=g^i$ if $i\leq j$ and $g^ig_{\leq j}=0$ if $i>j$.

\vspace{2mm} \noindent \emph{NKA derivation}:
To simplify the proof, let
\begin{align*}
    X&=g_{>0}g_{>1}(m_{21}p_2+m_{20}g^0), ~~
    Y=g_{>0}g_{\leq 1}(m_{11}p_1+m_{10}g^2).
\end{align*}
Then $\enc(\textsc{Constructed})$ is equivalent to $g^1(X+Y)^*g_{\leq 0}.$ We simplify $g^iX^*$ first.

\begin{align*}
    g^1X^*&=g^1(1+g_{>0}g_{>1}(m_{20}g^0+m_{21}p_2)X^*) \tag{\ref{fp-rule}} \\
    &=g^1, \tag{distributive-law}\\
    g^2X^*&=g^2(g_{>0}g_{>1}m_{21}p_2)^*(g_{>0}g_{>1}m_{20}g^0 \\
    &\qquad \cdot (1+g_{>0}g_{>1}m_{21}p_2(g_{>0}g_{>1}m_{21}p_2)^*))^* \tag{\ref{denesting-rule}, \ref{fp-rule}}\\
    &=(m_{21}p_2)^*g^2(g_{>0}g_{>1}m_{20}g^0)^* \tag{\ref{star-rew-rule}}\\
    &=(m_{21}p_2)^*g^2(1+g_{>0}g_{>1}m_{20}g^0)
    \tag{\ref{fp-rule}} \\
    &=(m_{21}p_2)^*(g^2+m_{20}g^0). \tag{distributive-law}
\end{align*}

Consider $g^1(X+Y)^*=g^1X^*(YX^*)^*=g^1(YX^*)^*,$ and then
\begin{align*}
    g^1(YX^*)^*
    =~&g^1(g_{>0}g_{\leq 1}m_{11}p_1X^*)^* \\
    &\cdot (g_{>0}g_{\leq 1}m_{10}g^2X^*(g_{>0}g_{\leq 1}m_{11}p_1X^*)^*)^* \tag{\ref{denesting-rule}}
\end{align*}
\begin{align*}
    =~&(m_{11}p_1)^*g^1(g_{>0}g_{\leq 1}m_{10}m_{21}^*(g^2+m_{20}g^0) \\
    &~~ \cdot (1+(g_{>0}g_{\leq 1}m_{11}p_1X^*)(g_{>0}g_{\leq 1}m_{11}p_1X^*)^*)) \tag{\ref{star-rew-rule}, \ref{fp-rule}} \\
    =~&(m_{11}p_1)^*m_{10}(m_{21}p_2)^*(g^2+m_{20}g^0).
\end{align*}

Insert the above equation into $g^1(X+Y)^*g_{\leq 0}.$
\begin{align*}
    g^1(X+Y)^*g_{\leq 0}
    =~&(m_{11}p_1)^*m_{10}(m_{21}p_2)^*(g^2+m_{20}g^0)g_{\leq 0} \\
    =~&(m_{11}p_1)^*m_{10}(m_{21}p_2)^*m_{20}g^0.
\end{align*}

This is exactly $\enc(\textsc{Constructed})=\enc(\textsc{Original})$. 
\thm{main} gives $\sem{\textsc{Constructed}}=\sem{\textsc{Original}}.$
Hence the two loops have been merged into one, with an additional classical space which is restored to 0 at the end.

We employ a similar idea to arbitrary programs by induction. Note that our above example corresponds to the case $S_1;S_2$ in induction. And our analysis above, which results in an equivalent program with one \textbf{while}-loop and additional classical space, constitutes a proof in that case. 
The more complicated cases are proved similarly, whose details are in \app{normal-form}.

\begin{theorem}
  \label{thm:normal-form}
  For any quantum while program $P$ over Hilbert space $\H$, there are a classical space $\C$ and a quantum while program \begin{align} P_0; ~\qwhile{M}{P_1}; ~p_{\C}:=\ket{0} \end{align} equivalent to $P; ~p_{\C}:=\ket{0}$ over $\H\otimes\C,$ where $P_0$, $P_1$ are while-free, $p_{\C}:=\ket{0}$ resets the classical variables in $\C$ to $\ket{0}$.
\end{theorem}

%% file: NKAT.tex
\section{Non-idempotent Kleene Algebra with Tests} \label{sec:NKAT}

As we stated before, NKA is not specifically designed for quantum programs: the measurements are treated as normal processes. Further characterization of measurements will grant finer algebraic structure. KAT introduces tests into KA, relying on the ability to simultaneously represent branching and predicates by Boolean algebra. However, for quantum programs, there is a gap between branching and predicates, which requires us to treat predicates and branching separately. We introduce effect algebra as a subalgebra of NKA to tackle quantum predicates in \sec{NKAT:effect}. As for branching, quantum measurements are abstracted as algebraic rules based on predicates. These lead to non-idempotent Kleene algebra with tests (NKAT) in \sec{NKAT:NKAT}. As an application, we show how propositional quantum Hoare logic is subsumed into algebraic rules of NKAT in \sec{NKAT:qht} and \sec{NKAT:pqhl}.

\subsection{Effect Algebra} \label{sec:NKAT:effect}

The notion of quantum predicates was defined in \cite{DP2006} as an operator $A\in\PO(\H)$ satisfying $\norm{A}\leq 1,$ and its negation $\neg{A}=I_{\H}-A.$ In the quantum foundations literature, quantum predicates are also called \emph{effects}. Their algebraic properties have been extensively studied as effect algebras. 

\begin{definition}[\cite{foulis1994effect}]
  An effect algebra (EA) is a $4$-tuple $(\L, \oplus, 0, e)$, where $0,e\in \L$, and $\oplus:\L\times \L\rightarrow \L$ is a partial binary operation satisfying the following properties: for any $a,b,c\in\L$, 
  \begin{enumerate}
  \item if $a\oplus b$ is defined then $b\oplus a$ is defined and $a\oplus b=b\oplus a$;
  \item if $a\oplus b$ and $(a\oplus b)\oplus c$ are defined, then $b\oplus c$ and $a\oplus (b\oplus c)$ are defined and $(a\oplus b)\oplus c=a\oplus (b\oplus c)$;
  \item if $a\oplus e$ is defined, then $a=0$;
  \item for every $a\in \L$ there exists a unique $\neg{a}\in \L$ such that $a\oplus\neg{a}=e$;
  \item for every $a\in\L$, $0\oplus a=a$.
  \end{enumerate}
\end{definition}

The fourth rule of the effect algebra defines a unary operator, the \emph{negation} over $\L$, denoted by $\bar{a}$ for $a\in\L$.

An effect algebra is easily embedded in NKA by viewing $\oplus$ as a restricted $+$ of NKA. Then we need to identify the correspondence of predicates in the quantum path model.

\begin{definition}
For a predicate $A$, a constant superoperator $\C_A\in\CP(\H)$ for $A\in\PO(\H)$ is defined by
\begin{align}
    \C_A(\rho)=\tr[\rho]A.
\end{align}

We let $\pred(\H)=\{\cp{\C_A}: A\in\PO(\H), \norm{A}\leq 1\}$ be the subset of $\LT(\H)$ containing the lifted constant superoperator.

A partial binary addition $\oplus$ over $\pred(\H)$ inherits from the addition in $\LT(\H)$, defined by:
\begin{align*}
  \cp{\C_A}\oplus \cp{\C_B}=\begin{cases} \cp{\C_A}+\cp{\C_B} & \cp{\C_A}+\cp{\C_B}\preceq \cp{\C_{I_{\H}}},\\ 
      \text{undefined} &\text{otherwise}. \end{cases}
\end{align*}
\end{definition}

\begin{lemma}
\label{lem:p-is-ea}
$(\pred(\H), \oplus, \O_{\H}, \cp{\C_{I_{\H}}})$ forms an effect algebra. Specifically, the negation of it satisfies $\neg{\cp{\C_A}}=\cp{\C_{\neg{A}}}.$
\end{lemma}

The proof is straightforward and in \app{ea-proofs}.

\subsection{Non-idempotent Kleene Algebras with Tests}
\label{sec:NKAT:NKAT}

We can characterize quantum measurements with the help of predicates, for which we propose \emph{partitions} algebraically.

\begin{definition}
\label{defn:nkat}
  An $\NKAT$ is a many-sorted algebra $(\K, \L, \N, \\ +, \cdot, *, \leq, 0, 1, e)$ such that
  \begin{enumerate}
  \item $(\K, +, \cdot, *, \leq, 0, 1)$ is an $\NKA$;
  \item $\L$ is a subset of $\K$, and $(\L, \oplus, 0, e)$ is an effect algebra, where $\oplus$ is the restriction of $+$ w.r.t. top element $e$ and partial order $\leq$; that is, for any $a,b\in\mathcal{L}$
    \begin{align}
      a\oplus b=\begin{cases} a+b & a+b\leq e,\\ 
      \text{undefined} &\text{otherwise}; \end{cases}
    \end{align}
  \item $\N$ is a set of tuples $(m_i)_{i\in I}$, where $I$ are finite index sets and $m_i\in\K$, satisfying:
  \begin{enumerate}
      \item each entry in the tuples satisfies $m_i\L\subseteq\L$; that is, for $a\in\L$, $m_i a\in\L.$
      \item for each tuple, $\sum_{i\in I} m_ie=e.$
  \end{enumerate}
  \end{enumerate}
  
  The tuples in $\N$ are called \emph{partitions}.
\end{definition}

We use $\L$ to characterize quantum predicates, and $\N$ to characterize branching in quantum programs. For a quantum measurement $\{M_i\}_{i\in I}$, its dual superoperators transform quantum predicates to quantum predicates: $\E_{M_i}^{\dagger}(A)=M_i^{\dagger}AM_i$. This is captured by $m_i\L\subseteq\L.$ Besides, general quantum measurements satisfies $\sum_{i\in I} M_i^{\dagger}M_i=I$, which is captured by $\sum_{i\in I}m_i e=e$, since $e$ represents predicate $I_{\H}.$
\footnote{
Our characterization of measurements matches positive-operator-valued measurements (POVM), the most general quantum measurements. We can further classify structures inside $\N$ to depict specific classes of quantum measurements. For example, projection-valued measurements (PVM) can be modeled as tuples $(m_i)_{i\in I}$ where $m_im_j=m_i$ if $i=j$, otherwise $m_im_j=0$.

Furthermore, a set of projective and pair-wise commutative measurement superoperators, defined by
$C(\H)=\{\E\in\CP(\H): \E(\rho)=D\rho D^{\dagger}, D\mathrm{~is~diagonal}, D^2=D \},$
represents the measurement superoperators in probabilistic programs. A Boolean algebra can be observed from it. It would be an interesting future direction to investigate the algebraic relation between NKAT and this Boolean algebra.
}

\begin{definition}
  {The set of quantum measurements lifted as quantum path actions in the dual sense is $\meas(\H)=$ \parfillskip=0pt\par}
  
  \noindent $\left\{\left(\cp{\M_i^{\dagger}}\right)_{i\in I} ~ : ~ \M_i(\rho)=M_i\rho M_i^{\dagger}, \sum_{i\in I}M_i^{\dagger}M_i=I_{\H}\right\}$.
\end{definition}

Then we augment the quantum path model in the NKAT framework, supporting quantum predicates ($\pred(\H)$) and quantum measurements ($\meas(\H)$).

\begin{theorem}
\label{thm:q-nkat}
The NKAT axioms are sound for the algebra $(\LT(\H), \pred(\H), \meas(\H), +, \diamond, *, \preceq, \O_{\H}, \I_{\H}, \cp{\C_{I_{\H}}})$.
\end{theorem}

Note we have substituted the right composition operation $\diamond$ for the left composition operation $;$ in $\LT(\H)$. This is mainly because our interpretation now uses dual superoperators. The verification of each axiom is standard. The detailed proofs are included in \app{ea-proofs}.

Several useful rules are derivable in NKAT, and their proofs are in \app{nkat-rules}.

\begin{lemma}
\label{lem:nkat-rules}
  The following formulae are derivable in NKAT. Here $I$ is a finite index set, $a, b, a_i$ are elements of the effect subalgebra, $(m_i)_{i\in I}$ is a partition.
  
  \begin{enumerate*}[series=MyList, before=\hspace{1ex}]
      \item ~$0\leq a\leq e$; $\qquad$
      \item ~$a+\neg{a}=e$; $\qquad$
      \item ~$\neg{\neg{a}}=a$;
  \end{enumerate*}
  \vspace{-1.5mm}
  \begin{enumerate}[resume=MyList]
      \item \label{negrev-rule} $a\leq b\rightarrow \neg{b}\leq \neg{a}$; \qquad \qquad \qquad \qquad (negation-reverse)
      \item \label{summa-rule} $\neg{\sum_{i\in I}m_ia_i}=\sum_{i\in I}m_i \neg{a_i}.$ \qquad \qquad (partition-transform)
  \end{enumerate}
\end{lemma}

\subsection{Encoding of Quantum Hoare Triples}
\label{sec:NKAT:qht}

A natural usage of classical predicates is reasoning via Hoare triples. With an algebraic representation of quantum predicates and programs, we can encode quantum Hoare triples as algebraic formulae.
A quantum Hoare triple is a judgment of the form $\{A\} P \{B\}$ where $A, B$ are quantum predicates and $P$ is a quantum program. It refers to partial correctness \cite{Ying16}, denoted by $\models_{par} \{A\}P\{B\}$, if for all input $\rho\in\D(\H)$ there is
\begin{align}
\tr(A\rho)\leq\tr(B\sem{P}(\rho)) +\tr(\rho)-\tr\left(\sem{P}(\rho)\right).
\end{align}

Then partial correctness $\models_{par} \{A\} P\{B\}$ can be encoded as an inequality $p\neg{b}\leq\neg{a}$, where $p$ is the encoding of program $P$, and effect algebra elements $a, b$ are the encoding of constant superoperators $\C_A$ and $\C_B$. This encoding can be interpreted by a dual interpretation $\Qintdag$.
\footnote{
A \emph{dual interpretation} $\Qint^{\dagger}$ is defined similar to $\Qint$ except for $\Qintdag(e\cdot f)=\Qintdag(e)\diamond\Qintdag(f)$ and $\Qintdag(a)=\cp{\eval(a)^{\dagger}}$. It describes the dual superoperators lifted to $\LT(\H).$ Properties of $\Qint$ like \thm{q-model-sound-complete}, \cor{primal-int} hold for the dual interpretation similarly. Analogous of \thm{enc-recovery}, $\Qintdag(\enc(P))=\cp{\sem{P}^{\dagger}}$, holds as well.
}
By setting any non-zero input for $\Qintdag(p\neg{b})\preceq\Qintdag(\neg{a})$ and \lem{E-embedding}, it turns to $\sem{P}^{\dagger}(I-B)\sqsubseteq I-A$, which is equivalent to $\models_{par} \{A\} P\{B\}$.

\subsection{Propositional Quantum Hoare Logic}
\label{sec:NKAT:pqhl}

An important feature of KAT is that KAT subsumes the deductive system of propositional Hoare logic, which contains the rules directly related to the control flow of classical while programs but not the rule for assignments \cite{K00a}. 
As a counterpart, quantum Hoare logic is an important tool in the verification and analysis of quantum programs.
A sound and (relatively) complete proof system for partial correctness of \quwhile~programs presented in \fig{pqhl} is discussed in \cite{Yin11}.
We aim to subsume in NKAT a fragment of quantum Hoare logic, called \emph{propositional} quantum Hoare logic. 

Due to the no-cloning of quantum information, the role of assignment is played by initialization and unitary transformation together in quantum programming. In quantum Hoare logic, the rule (Ax.In) and (Ax.UT) for them include atomic transformations, which cannot be captured by algebraic methods. 
As such, propositional quantum Hoare logic will treat these rules as atomic propositions and work with the following program syntax
\begin{align*}
 P \enskip ::= & \enskip p \enskip | \enskip \cskip 
                 \enskip | \enskip \cfail
                 \enskip | \enskip  P_1;P_2 \enskip | \enskip \\
               & \enskip \qif{M[\overline{q}]\xrightarrow{i} P_i} \enskip | \enskip \qwhile{M[\overline{q}]=1}{P_1}.
\end{align*}
Therefore, the deductive system of propositional quantum Hoare logic consists of the rules marked red in \fig{pqhl}. Its relative completeness and soundness can be proved similarly to the original quantum Hoare logic~\cite{Yin11} as a routine exercise. 

\begin{figure}[bt]
  \centering
  \scalebox{0.77}{
  \begin{minipage}{1\linewidth}
  \begin{align*}
      &(\mathrm{Ax.UT}) \ \ \triple{U^{\dagger}AU}{\ \overline{q}:=U[\overline{q}]\ }{A} 
      & &(\mathrm{Ax.In})\ \left\{\sum_i \ket{i}_q\bra{0}A\ket{0}_q\bra{i}\right\}q:=\ket{0}\ \{A\} \\
      &{\color{red} (\mathrm{Ax.Sk})}\qquad \ \ \triple{A}{\ \cskip\ }{A} 
      & &{\color{red} (\mathrm{R.OR})}\ \quad \frac{A \sqsubseteq A' \quad \triple{A'}{P}{B'} \quad B'\sqsubseteq B}{\triple{A}{P}{B}} \\
      &{\color{red} (\mathrm{Ax.Ab})}\quad \ \triple{I_{\H}}{\ \cfail\ }{O_{\H}} 
      & &{\color{red} (\mathrm{R.IF})}\ \ \  \frac{\triple{A_i}{P_i}{B}~\mathrm{for~all}~i~}{\triple{\sum_i\M_i^{\dagger}(A_i)}{\mathbf{case~}{M\xrightarrow{i} P_i}\mathbf{~end}}{B}}\\
      &{\color{red} (\mathrm{R.SC})}\quad \  \frac{\triple{A}{P_1}{B}\ \ \ \triple{B}{P_2}{C}}{\triple{A}{P_1;P_2}{C}} 
      & &{\color{red} (\mathrm{R.LP})}\ \ \frac{\triple{B}{P}{C} \quad C=\M_0^{\dagger}(A)+\M_1^{\dagger}(B)}{\triple{C}{\qwhile{M=1}{P}}{A}}
  \end{align*}
  \end{minipage}
  }
  \caption{A proof system for partial correctness of quantum programs. Propositional quantum Hoare logic includes the rules marked {\color{red} red} in this figure (the lower six rules).}
  \label{fig:pqhl}
\end{figure}

By the discussions in \sec{NKAT:qht}, the partial correctness of quantum Hoare triples can be encoded in NKAT. For a quantum measurement $\{M_i\}_{i\in I}$ we have an additional normalization rule $\sum_{i}M_i^{\dagger}M_i=I$, which is encoded as $\sum_i m_ie=e.$ Then the encoding of these rules is
\begin{align*}
  \begin{cases}
    {\rm (Ax.Sk):} & 1\neg{a}\leq \neg{a}, \\
    {\rm (Ax.Ab):} & 0\neg{0}\leq \neg{1}, \\
    {\rm (R.OR):} &  a\leq a'\wedge p\neg{b'}\leq\neg{a'} \wedge b'\leq b\rightarrow p\neg{b}\leq\neg{a}, \\
    {\rm (R.IF):} &  \left(\bigwedge_{i\in I} p_i\neg{b}\leq \neg{a_i}\right) \rightarrow (\sum_{i\in I} m_ip_i)\neg{b}\leq\neg{\sum_i m_ia_i}, \\
    {\rm (R.SC):} &  p_1\neg{b}\leq \neg{a} \wedge p_2\neg{c}\leq \neg{b} \rightarrow p_1p_2\neg{c}\leq\neg{a},\\
    {\rm (R.LP):} &  p\neg{m_0a+m_1b}\leq\neg{b}\rightarrow (m_1p)^*m_0\neg{a}\leq\neg{m_0a+m_1b}.
  \end{cases} \label{eq:pqhl}
\end{align*}
Here $I$ is a finite index set, $p, p_i\in\K$, elements $a, b, c, a', b', a_i\in\L$, and $(m_i)_{i\in I}$ are partitions.

\begin{theorem}
  With partitions $(m_i)_{i\in I}$, the formulae above are derivable in NKAT.
\end{theorem}

\begin{proof}
  ~
  \begin{enumerate}
  \item (Ax.Sk): $1\neg{a}=\neg{a}$.
  \item (Ax.Ab): $0\neg{0}=0\leq \neg{1}$ by \ref{pos-rule}.
  \item (R.OR): By \hyperref[negrev-rule]{negation-reverse}, we have $\neg{a'}\leq\neg{a}$ and $\neg{b}\leq\neg{b'}.$ So $p\neg{b}\leq p\neg{b'}\leq \neg{a'}\leq \neg{a}.$  \item (R.IF): Applying \hyperref[summa-rule]{partition-transform}, $(\sum_{i\in I} m_ip_i)\neg{b}=\sum_{i\in I} m_ip_i\neg{b}\leq \sum_{i\in I} m_i\neg{a_i}=\neg{\sum_i m_ia_i}.$
  \item (R.SC): $p_1(p_2\neg{c})\leq p_1\neg{b}\leq\neg{a}.$
  \item (R.LP): By \hyperref[summa-rule]{partition-transform}, $\neg{m_0a+m_1b}=m_0\neg{a}+m_1\neg{b}$. With $p\neg{m_0a+m_1b}\leq\neg{b},$ we have \begin{align*}
      m_0\neg{a}+m_1p\neg{m_0a+m_1b}\leq m_0\neg{a}+m_1\neg{b}=\neg{m_0a+m_1b}.
  \end{align*}
  Then $(m_1p)^*m_0\neg{a}\leq\neg{m_0a+m_1b}$ is concluded by applying $q+pr\leq r\rightarrow p^*q\leq r$.
  \end{enumerate}
\end{proof}

It is clear that the NKAT subsumes the encoding of propositional quantum Hoare logic. 

\section*{Acknowledgement}
We thank the anonymous reviewers and our shepherd, Calvin Smith, for their helpful feedback. 
Y.P. and X.W. were supported by the U.S. Department of Energy, Office of Science, Office of Advanced Scientific Computing Research, Quantum Testbed Pathfinder Program under Award
Number DE-SC0019040 and the U.S. Air Force Office of Scientific Research MURI grant FA9550-16-1-0082. M.Y. was supported by the National Key R\&D Program of China (2018YFA0306701) and the
National Natural Science Foundation of China (61832015).

%% file: Appendix.tex
\begin{center}
 \LARGE \bf Appendix
\end{center}

\section{From NKA to Formal and Rational Power Series} \label{app:NKA:FPS}

Researches on formal power series date back to \cite{SCHUTZENBERGER1961245}, and see also some recent references \cite{berstel2011noncommutative, droste2009handbook}. 

Formal power series generalize formal languages by weighing strings with the extended natural number $\natinf$.

\begin{definition}
The \emph{extended set of natural numbers} is $\natinf=\nat\cup\{\infty\}$, where $\infty$ is an added top element.
The calculation in this semiring follows the correspondences in $\nat$, and: 
\begin{align*}
    &0+\infty=\infty, \quad 0\cdot\infty=\infty\cdot 0=0, \quad \ \ 0^*=1; \\
    \forall n\in\natinf\backslash\{0\}: \quad &n+\infty=\infty, \quad  n\cdot\infty=\infty\cdot n=\infty, \quad n^*=\infty.
\end{align*}

A countable summation $\sum_{i\in I} n_i$ for $n_i\in\natinf$ is defined to be $\infty$ if there exists an $i_0\in I$ such that $n_{i_0}=\infty$, or if there exists infinitely many non-zero $n_i$'s. In other cases, it degenerates to a finite summation and the definition follows naturally.

The partial order in $\natinf$ extends the natural partial order in $\nat$ by $\forall n\in\natinf, n\leq \infty$.  
\end{definition}


\begin{definition}[\cite{droste2009handbook, berstel2011noncommutative}]
  For a finite alphabet $\Sigma$, a \emph{formal power series} $\fps{f}$ over $\Sigma$ is a function $\fps{f}: \Sigma^*\rightarrow \natinf$, and can be represented by $\fps{f}=\sum_{w\in\Sigma^*} \fpsc{f}{w}w$
  where $\fpsc{f}{w}\in\natinf$ is the \emph{coefficient} of string $w$.
  We denote the set of the formal power series over $\Sigma$ by $\ps.$
  
  
\end{definition}

For example, the zero mapping in $\ps$ is represented by $\fps{f}=0$. The unit mapping $\fps{f}=1\empstr$ maps the empty string $\empstr$ to $1$, and the others to $0$. The mapping represented by $\fps{f}=1a$ for $a\in\Sigma$ maps $a$ to $1$, and the others to $0$.



\begin{definition}
  Addition, multiplication and the star operation are defined on $\ps$ by
  \begin{align}
    (\fps{f}+\fps{g})[w]&=\fpsc{f}{w}+\fpsc{g}{w}, \\
    (\fps{f}\cdot \fps{g})[w]&=\sum_{uv=w} \fpsc{f}{u}\fpsc{g}{v}, \\
    \label{eq:fps-star}
    (\fps{f}^*)[w]&=\sum_{n\geq 0} \sum_{u_1\cdots u_n=w} \fpsc{f}{u_1}\cdots\fpsc{f}{u_n}w. 
  \end{align}
  Here $uv$ is the concatenation of strings in $\Sigma^*$, and $u_i$ can be the empty string $\empstr$ in \eq{fps-star}.
  Note also that $\fps{f}^*=\sum_{n\geq 0} \fps{f}^n.$ 

  The partial order in $\ps$ is defined by:
  \begin{align}
    \fps{f}\leq \fps{g} \leftrightarrow \forall w\in\Sigma^*, \fpsc{f}{w}\leq \fpsc{g}{w}.
  \end{align}
\end{definition}


With these operations in formal power series, it is possible to interpret expressions over $\Sigma$ as formal power series over $\Sigma$ by a semantic mapping $\rpssem{-}$. 

\begin{definition}
  $\rpssem{-}:\expsig\rightarrow \ps$ is defined by
  \begin{align*}
     \rpssem{0}&=0, & \rpssem{a}&=1a,  & \rpssem{e+f}&=\rpssem{e}+\rpssem{f}, \\
     \rpssem{1}&=1\empstr, & \rpssem{e^*}&=\rpssem{e}^*, & \rpssem{e\cdot f}&=\rpssem{e}\cdot\rpssem{f},
  \end{align*}
  where $a\in\Sigma$, and $e, f\in\expsig$.
\end{definition}

Then we are able to define rational power series as an analogue to regular languages.

\begin{definition}[\cite{droste2009handbook, berstel2011noncommutative}]
  The set of \emph{rational power series}, denoted by $\ratps$, is the smallest subset of \, $\ps$ containing: $(1)~\fps{f}=0$; $(2)~\fps{f}=1\empstr$; $(3)~\fps{f}=1a$ for all $a\in \Sigma$, and is closed under $+, \cdot, *$.
\end{definition}

A series of works from \citet{esik2004inductive, beal2005equivalence, beal2006conjugacy, bloom2009axiomatizing} demonstrates the rational power series as a pivotal model for the NKA axioms.

\begin{theorem}[\cite{esik2004inductive, bloom2009axiomatizing}]
  \label{thm:rps-sound-complete}
  The NKA axioms are sound and complete for $(\ratps, +, \cdot, *, \\ \leq, 0, 1\empstr)$. 
  Namely, for any expression $e$ and $f$ over $\Sigma$, we have 
  \begin{align}
    \NKAmodel e=f ~~\Leftrightarrow~~ 
    \rpssem{e}=\rpssem{f}.
  \end{align}
\end{theorem}

\section{Optimizing Quantum Signal Processing}
\label{app:qsp}

\begin{figure}[b]
    \centering
    \begin{minipage}{.495\linewidth}
    \begin{align*}
    &\textsc{qsp}[q]\equiv \\
    & c:=\ket{n}; ~p:=\ket{+}; \tag{$c_0p_0$}\\
    & r:=\ket{G}; \tag{$r_0$}\\
    & \mathbf{while~} M[c]=1 \mathbf{~do~} \tag{\{$m_i$\}}\\
    & \quad c, p:=\Phi[c, p]; \tag{$\phi$}\\
    & {\color{red} \quad r:=S[r];} \tag{$s$}\\
    & \quad p, r, q:=\mathrm{C}_W[p, r, q]; \tag{$w_c$}\\
    & {\color{red} \quad r:=S^{-1}[r];} \tag{$s^{-1}$}\\
    & \quad c, p:=\Phi^{-1}[c, p]; \tag{$\phi^{-1}$}\\
    & \quad c:=\mathrm{Dec}[c] \tag{$d$}\\
    & \mathbf{done;} \\
    & \mathbf{if~} M_{\ket{+}\ket{G}}[p, r]=0 \\
    & \mathbf{then~abort} \tag{$\tau_0 0+\tau_1 1$}
    \end{align*}
    \end{minipage}
    \begin{minipage}{.495\linewidth}
    \begin{align*}
    &\textsc{qsp'}[q]\equiv \\
    & c:=\ket{n}; ~p:=\ket{+}; \tag{$c_0p_0$}\\
    & r:=\ket{G}; \tag{$r_0$}\\
    & \mathbf{while~} M[c]=1 \mathbf{~do~} \tag{\{$m_i$\}}\\
    & \quad c, p:=\Phi[c, p]; \tag{$\phi$}\\
    & \quad p, r, q:=\mathrm{C}_W[p, r, q]; \tag{$w_c$}\\
    & \quad c, p:=\Phi^{-1}[c, p]; \tag{$\phi^{-1}$}\\
    & \quad c:=\mathrm{Dec}[c] \tag{$d$}\\
    & \mathbf{done;} \\
    & \mathbf{if~} M_{\ket{+}\ket{G}}[p, r]=0 \\
    & \mathbf{then~abort} \tag{$\tau_0 0+\tau_1 1$}\\
    & \\
    \end{align*}
    \end{minipage}
    \caption{The program $\textsc{qsp}$ and $\textsc{qsp'}$. The measurement $M[c]$ is $\{M_1=\ket{0}\bra{0}, M_0=I_c-M_1\}$ on register $c$. The measurement $M_{\ket{+}\ket{G}}[p, r]$ is $\{M_1=\ket{+}\bra{+}\otimes\ket{G}\bra{G}, M_0=I_{p, r}-M_1\}$ on register $p$ and $r$ jointly.} 
    \label{fig:qsp}
\end{figure}

Quantum signal processing (QSP) \cite{low2017optimal} is an advanced quantum algorithm for Hamiltonian simulation problem. In \cite{Childs9456} an optimization is observed by canceling adjacent sub-processes. The QSP implementation before ($\textsc{qsp}$) and after ($\textsc{qsp'}$) the optimization is illustrated in \fig{qsp}.
The algorithm QSP simulates the Hamiltonian $H=\sum_{l=1}^L \alpha_l H_l$ on qubit register $q$ with high probability.
Let us explain the components in QSP briefly, whose details imply some commutativity conditions for our purpose. 
$\ket{G}=1/\sqrt{\sum_{l=1}^L \alpha_l}\sum_{l=1}^L\sqrt{\alpha_l}\ket{l}$ is a state defined by $H$. $\Phi=\sum_{j=1}^{n}\ket{j}\bra{j}\otimes e^{-i\phi_j\sigma^Z/2}$ is an operation rotating qubit $p$ with a pre-defined angle $\phi_j$. Unitary $S=(1-i)\ket{G}\bra{G}-I$ is a partial reflection operator about state $\ket{G}$, and $W=-i((2\ket{G}\bra{G}-I)\otimes I)\sum_{l=1}^L\ket{l}\bra{l}\otimes H_l$, which defines $\mathrm{C}_W=\ket{+}\bra{+}\otimes I+\ket{-}\bra{-}\otimes W$. $\mathrm{Dec}=\ket{n}\bra{0}+\sum_{j=1}^n\ket{j-1}\bra{j}$ is the unitary implementing $j \mapsto (j-1) \mod n$. 

\vspace{2mm} \noindent \emph{Program Encoding}:
We encode the programs in \fig{qsp} as
\begin{align*}
    \enc(\textsc{qsp})&=c_0p_0r_0(m_1\varphi s w_c s^{-1} \varphi^{-1} d)^*m_0(\tau_0 0+\tau_1 1), \\
    \enc(\textsc{qsp'})&=c_0p_0r_0(m_1\varphi w_c \varphi^{-1} d)^*m_0(\tau_0 0+\tau_1 1).
\end{align*}
The detailed encoder setting is self-explanatory.

\vspace{2mm} \noindent \emph{Condition Formulation}:
One can derive commutative conditions because 
 $c, p:=\Phi[c, p]$ and $r:=S[r]$, similarly  $r:=S^{-1}[r]$ and $c, p:=\Phi^{-1}[c, p]; c:=\mathrm{Dec}[c]$, apply on different quantum variables and hence commute. 
Algebraically, we hence have $\varphi s=s\varphi$, and $\varphi^{-1}ds^{-1}=s^{-1}\varphi^{-1}d$.
Moreover, $M[c]$ is commutable to $r:=S[r]$, so $m_1 s=s m_1$ and $m_0 s=s m_0.$ 
Since $S\ket{G}\bra{G}S^{\dagger}=\ket{G}\bra{G}$, we have $r_0s=r_0.$ Similarly the Kraus operator $(\ket{+}\bra{+}\otimes\ket{G}\bra{G})\cdot (I_p\otimes((1+i)\ket{G}\bra{G}-I_r))=i\ket{+}\bra{+}\otimes\ket{G}\bra{G},$ and the phases are cancelled when represented by superoperator. This is encoded as $s^{-1}\tau_1=\tau_1.$ Then we need to show $\enc(\textsc{qsp})=\enc(\textsc{qsp'})$ with these hypotheses and the NKA axioms.

\vspace{2mm} \noindent \emph{NKA derivation}:
By \eq{loop-boundary}, we have 
\begin{align}
  &c_0p_0r_0(m_1\varphi s w_c s^{-1} \varphi^{-1} d)^*m_0(\tau_0 0+\tau_1 1) \notag \\
  =&c_0p_0r_0(sm_1\varphi w_c \varphi^{-1}d s^{-1})^*m_0\tau_1 \tag{commutativity}\\
  =&c_0p_0r_0s(m_1\varphi w_c \varphi^{-1}d)^*m_0s^{-1}\tau_1 \tag{\eq{loop-boundary}}\\
  =&c_0p_0r_0(m_1\varphi w_c \varphi^{-1}d)^*m_0\tau_1, \tag{absorption-hypotheses}\\
  =&c_0p_0r_0(m_1\varphi w_c \varphi^{-1}d)^*m_0(\tau_0 0+\tau_1 1). \notag
\end{align}
Notice that $m_1$ and $\varphi$ do not commute, so we cannot apply \eq{loop-boundary} further. By \cor{primal-int}, \thm{enc-recovery} and \lem{E-embedding}, $\sem{\textsc{qsp}}=\sem{\textsc{qsp'}}$. Note that in $\textsc{qsp'}$, $S$ and $S^{-1}$ vanish, which could largely reduce the total gate count.

\section{Proofs of Technical Results}

We call the last two star laws ($q+pr\leq r\rightarrow p^*q\leq r$ and $q+rp\leq r\rightarrow qp^*\leq r$) the inductive star laws. They are ubiquitous in the proofs.

\subsection{Detailed Proof of \lem{iter-star-rules}}
\label{app:nka-rules}

\begin{proof}[Proof of \lem{iter-star-rules}]
  We rewrite the proofs in \cite{esik2004inductive} for the rules in \fig{NKAtheorems}.
  \begin{itemize}
      \item ($1+pp^*=p^*$): By star laws there is $1+pp^*\leq p^*$, so we only need to prove the other side. Because $\leq$ is monotone, we multiply $p$ and then plus $1$ on the both sides, leading to
      \begin{align*}
          1+p(1+pp^*)\leq 1+pp^*.
      \end{align*}
      Applying the inductive star law gives $p^*\leq 1+pp^*$.
      \item ($1+p^*p=p^*$): First we show $\geq$ side. Notice that
      \begin{align*}
          1+p(1+p^*p)=1+p+pp^*p=1+(1+pp^*)p=1+p^*p.
      \end{align*}
      Applying the inductive star law, we have $p^*\leq 1+p^*p$.
      
      Then we show $\leq$ side. Applying star law,
      \begin{align*}
          p+ppp^*=p(1+pp^*) \leq pp^*.
      \end{align*}
      So $p^*p\leq pp^*$ holds. Because $\leq$ is preserved by $+$, we conclude $1+p^*p\leq 1+pp^*\leq p^*.$
      \item ($p\leq q\rightarrow p^*\leq q^*$): We multiply $q^*$ and add $1$ on both sides, which gives
      \begin{align*}
          1+pq^*\leq 1+qq^*\leq q^*.
      \end{align*}
      By star laws, there is $p^*\leq q^*$.
      \item ($1+p(qp)^*q=(pq)^*$): We show $\geq$ side first. By semiring laws there is
      \begin{align*}
          1+(pq)(1+p(qp)^*q)=1+p(1+qp(qp)^*)q =1+p(qp^*)q.
      \end{align*}
      Because of the inductive star law, we get $(pq)^*\leq 1+p(qp)^*q$.
      
      Similarly for $\leq$ side, we consider
      \begin{align*}
          q+qpq(pq)^*=q(1+pq(pq)^*=q(pq)^*.
      \end{align*}
      We know that $(qp)^*q\leq q(pq)^*.$ Multiplying $p$ and adding $1$ on the both sides give 
      \begin{align*}
        1+p(qp)^*q\leq 1+pq(pq)^*\leq(pq)^*.
      \end{align*}
      \item ($(pq)^*p=p(qp)^*$): Multiplying $p$ on \ref{product-star-rule} results in
      \begin{align*}
          (pq)^*p=p+p(qp)^*qp=p(qp)^*.
      \end{align*}
      \item ($(p+q)^*=(p^*q)^*p^*$): To show $(p+q)^*\leq (p^*q)^*p^*$, we apply \ref{sliding-rule} twice, \ref{fp-rule} twice, followed by \ref{sliding-rule} once :
      \begin{align*}
          1+(p+q)(p^*q)^*p^*&=1+p(p^*q)^*p^*+q(p^*q)^*p^* \\
          &=pp^*(qp^*)^*+(1+(qp^*)^*qp^*) \\
          &=pp^*(qp^*)^*+(qp^*)^* \\
          &=(1+pp^*)(qp^*)^* \\
          &=p^*(qp^*)^* \\
          &=(p^*q)^*p^*
      \end{align*} 
      Then by the inductive star law there is $(p+q)^*\leq (p^*q)^*p^*.$
      
      The other side is by 
      \begin{align*}
          (p+q)^*=1+(p+q)(p+q)^*=(1+q(p+q)^*)+p(p+q)^*.
      \end{align*}
      Because of the inductive star law there is
      \begin{align*}
          p^*+p^*q(p+q)^*=p^*(1+q(p+q)^*)\leq (p+q)^*.
      \end{align*}
      Apply it once more, we eventually get $(p^*q)^*p^*\leq (p+q)^*.$
      \item ($(p+q)^*=p^*(qp^*)^*$): By \ref{sliding-rule} there is $p^*(qp^*)^*=(p^*q)^*p^*=(p+q)^*.$
      \item ($0\leq p$): Note that $0+1\cdot p=p\leq p$. Apply the inductive star law, and we have $0=1^*\cdot 0\leq p.$
  \end{itemize}
  
  Rules in \fig{optlem} can be derived by:
  \begin{itemize}
  \item (\ref{unr-rule}): For $\leq$ side, applying \ref{fp-rule} twice on $p^*$, we have $$1+p+(pp)p^*=p^*.$$ Applying the inductive star law, we have $(pp)^*(1+p)\leq p^*$. \\
    For $\geq$ side, applying \ref{fp-rule} on $(pp)^*$ we have $$1+(pp)^*(1+p)p=(pp)^*p+(1+(pp)^*pp)=(pp)^*(1+p).$$ Then by the inductive law, we have $p^*\leq (pp)^*(1+p).$
  \item (\ref{comstar-rule}): Applying \ref{fp-rule} there is $$p^*q=q+p^*pq=q+p^*qp.$$ By the inductive star law there is $qp^*\leq p^*q$. \\
  Similarly, the other side is by $$qp^*=q+qpp^*=q+pqp^*,$$ which leads to $p^*q\leq qp^*$. 
  \item (\ref{star-rew-rule}): By \ref{fp-rule} there is $$r^*p=p+r^*rp=p+r^*pq.$$ Applying the inductive star law there is $pq^*\leq r^*p.$ \\
  To prove the other side, note that with \ref{fp-rule} there is $$pq^*=p+pqq^*=p+rpq^*.$$ The inductive star law gives $r^*p\leq pq^*.$
  \end{itemize}
\end{proof}

\subsection{Detailed Proofs of \lem{SHequiv} and Several Facts about $\S(\H)$}
\label{app:SHfacts}

\begin{proof}[Proof of \lem{SHequiv}]
 Reflexivity is proved by choosing $J'=I'$ in the definition.
 To prove transitivity, we assume $\biguplus_{i\in I}\rho_i\lesssim\biguplus_{j\in J}\sigma_j$ and $\biguplus_{j\in J}\sigma_j\lesssim \biguplus_{k\in K}\gamma_k.$ For $\epsilon>0$ and finite $I'\subseteq I$, there exists a finite $J'\subseteq J$ such that $\sum_{i\in I'}\rho_i\sqsubseteq \frac{\epsilon}{2} I_{\H}+\sum_{j\in J'}\sigma_j.$ Then there exists a finite $K'\subseteq K$ such that $\sum_{j\in J'}\sigma_j\sqsubseteq \frac{\epsilon}{2} I_{\H}+\sum_{k\in K'}\gamma_k$ as well. Because $\sqsubseteq$ is monotone with respect to $+$, we have $\sum_{i\in I'}\rho_i\sqsubseteq\epsilon I_{\H}+\sum_{k\in K'}\gamma_k.$
\end{proof}

\begin{lemma}
  \label{lem:SHfacts}
  We demonstrate several basic facts about $\S(\H).$
  \begin{enumerate}[label=(\roman*), ref={\ref{lem:SHfacts}.(\roman*)}]
    \item \label{lem:S-biguplus-order} If for all $i\in I$, $\biguplus_{j\in J_i}\rho_{ij}\lesssim\biguplus_{k\in K_i}\sigma_{ik},$ then
      \begin{align}
        \biguplus_{i\in I} \biguplus_{j\in J_i}\rho_{ij}\lesssim \biguplus_{i\in I} \biguplus_{k\in K_i}\sigma_{ik}.
      \end{align}
    \item \label{lem:S-natinf} Let $n_i\in\natinf$ for all $i\in I.$ Then for all $\biguplus_{j\in J}\rho_j\in\S(\H),$ there is
      \begin{align}
        \biguplus_{i\in I}\biguplus_{0\leq k<n_i}\biguplus_{j\in J_i}\rho_j\sim\biguplus_{0\leq k<\sum_{i\in I}n_i}\biguplus_{j\in J_i}\rho_j.
      \end{align}
    Here $\{k:0\leq k<\infty\}=\nat.$
    \item \label{lem:convergent-sum} If $\sum_{i\in I} \rho_i$ converges in $\PO(\H)$, then
      \begin{align}
        \biguplus_{i\in I}\rho_i\sim\left\{\!\left|\sum_{i\in I}\rho_i\right|\!\right\}.
      \end{align}
    \item \label{lem:S-*-continuity} For a series $\biguplus_{i\in\nat}\biguplus_{j\in J_i}\rho_{ij}\in\S(\H)$, if there exists $\biguplus_{k\in K} \sigma_k$ such that for all $n\geq 0,$
      \begin{align}
        \biguplus_{0\leq i<n} \biguplus_{j\in J_i}\rho_{ij}\lesssim \biguplus_{k\in K} \sigma_k,
      \end{align}
      then $\biguplus_{i\in\nat} \biguplus_{j\in J_i}\rho_{ij}\lesssim \biguplus_{k\in K} \sigma_k.$
    \item \label{lem:E-preorder} If $\biguplus_{i\in I}\rho_i\lesssim\biguplus_{j\in J}\sigma_j,$ then for $\E\in\CP(\H),$ there is
      \begin{align}
        \biguplus_{i\in I}\E(\rho_i)\lesssim \biguplus_{j\in J}\E(\sigma_j).
      \end{align}
  \end{enumerate}
\end{lemma}

\begin{proof}[Proof of \lem{SHfacts}]
  W.l.o.g. we assume the index sets to be subsets of $\nat$.
  \begin{enumerate}[label=(\roman*)]
      \item For any $\epsilon>0$ and any finite subseries $\biguplus_{i\in I, j\in J'_i}\rho_{ij}$ of $\biguplus_{i\in I} \biguplus_{j\in J_i}\rho_{ij}$, there exists an $N$ such that for $i\geq N$, there is $J'_i=\phi$. When $N=0$, then $\{(i, j):i\in I, j\in J'_i\}=\phi$ and the inequality holds with an empty subset chosen on the right hand side. Otherwise let $\epsilon'=\frac{\epsilon}{N},$ so there exist finite index set $K_i'$ for each $0\leq i<N$ such that $\sum_{j\in J_i'}\rho_{ij}\sqsubseteq \epsilon'I+\sum_{k\in K_i'}\sigma_{ik}$. Adding them up gives $\sum_{0\leq i<N, j\in J_i'} \rho_{ij}\sqsubseteq \epsilon I + \sum_{0\leq i < N, k\in K_i'} \sigma_{ik}$. This concludes $\biguplus_i \biguplus_{j\in J_i}\rho_{ij}\lesssim \biguplus_i \biguplus_{k\in K_i}\sigma_{ik}.$
      \item By reordering the multisets it holds apparently.
      \item ($\lesssim$): Notice that for any finite $I'\subseteq I$, $\sum_{i\in I'} \rho_i\sqsubseteq \sum_{i\in I}\rho_i.$ Then this direction comes from the definition.
  
      ($\gtrsim$): Since $\sum_{i\in I}\rho_i$ converges, for any $\epsilon>0$ there is an $N>0$ such that $\norm{\sum_{i\in I, i > N}\rho_i}\leq \epsilon,$ where $\norm{\cdot}$ is the spectral norm. Hence $\sum_{i\in I}\rho_i\sqsubseteq \epsilon I_{\H}+\sum_{i\in I, i \leq N}\rho_i.$ This gives $\gtrsim$ direction.
      \item Consider any finite subseries $\biguplus_{i\in\nat, j\in J_i'}\rho_{ij}$ selected from $\biguplus_{i\geq 0} \biguplus_{j\in J_i}\rho_{ij}.$ There exists $N$ such that for all $i\geq N, J_i'=\phi$. Let $n=N$ in the assumption, then we know that for any $\epsilon>0$ there exists a finite $K'\subseteq K$ such that $\sum_{0\leq i<N, j\in J_i'}\rho_{ij}\sqsubseteq \epsilon I_{\H}+\sum_{k\in K'}\sigma_k,$ and this concludes the proof.
      \item If $\E(I_{\H})=O_{\H},$ then $\E\equiv O_{\H}$, and we are done by definition. Now we assume $\E(I_{\H})\neq O_{\H}.$ For every finite $I'\subseteq I$ and $\epsilon>0$, there exists $J'\subseteq J$ such that $\sum_{i\in I'} \rho_i\sqsubseteq \frac{\epsilon}{\norm{\E(I_{\H})}}I_{\H}+\sum_{j\in J'}\sigma_j.$ Then \begin{align*}
        \sum_{i\in I'}\E(\rho_i)=\E\left(\sum_{i\in I'}\rho_i\right)&\sqsubseteq \E\left(\frac{\epsilon}{\norm{\E(I_{\H})}}I_{\H}+\sum_{j\in J'}\sigma_j\right) \\ &\sqsubseteq \epsilon I_{\H}+\sum_{j\in J'}\E(\sigma_j).
      \end{align*}
      Here $\norm{\cdot}$ is the spectral norm. This leads to $\biguplus_{i\in I}\E(\rho_i)\lesssim \biguplus_{j\in J}\E(\sigma_j).$
  \end{enumerate}
\end{proof}

\subsection{Detailed Proofs of \lem{P-oper-closed} and \thm{QS-NKA}}
\label{app:PH-thms}

\begin{lemma}
  \label{lem:P-oper-closed}
  $\sum_i, ;$ and $*$ operations are closed in $\LT(\H)$.
\end{lemma}

\begin{proof}[Proof of \lem{P-oper-closed}]
  The monotone of $\sum_i$ follows \lem{S-biguplus-order}, and the monotone of $;$ follows the definition. It suffices to verify the linearity of them.
  
  For $\sum_i,$ notice that
  \begin{align*}
      \left(\sum_k\T_k\right)\left(\sum_i \sum_{j\in J_i}\eqv{\rho_{ij}}\right)&=\sum_k\T_k\left(\sum_i \sum_{j\in J_i}\eqv{\rho_{ij}}\right) \\
      &=\sum_k\sum_i\T_k\left(\sum_{j\in J_i}\eqv{\rho_{ij}}\right) \\
      &=\sum_i\sum_k\T_k\left(\sum_{j\in J_i}\eqv{\rho_{ij}}\right) \\
      &=\sum_i\left(\sum_k\T_k\right)\left(\sum_{j\in J_i}\eqv{\rho_{ij}}\right).
  \end{align*}
  
  For $;$ operation, it is directly proved by
  \begin{align*}
      (\T_1;\T_2)\left(\sum_i \sum_{j\in J_i}\eqv{\rho_{ij}}\right)&=\T_2\left(\sum_i\T_1\left(\sum_{j\in J_i}\eqv{\rho_{ij}}\right)\right) \\
      &=\sum_i\T_2\left(\T_1\left(\sum_{j\in J_i}\eqv{\rho_{ij}}\right)\right) \\
      &=\sum_i(\T_1;\T_2)\left(\sum_{j\in J_i}\eqv{\rho_{ij}}\right).
  \end{align*}
\end{proof}

\begin{proof}[Proof of \thm{QS-NKA}]
  The proofs of monotone of $+$ and $;$ operations, the star laws are presented here.
    \begin{itemize}
      \item $p\leq q\wedge r\leq s\rightarrow p+r\leq q+s$: First we show that $+$ and $\leq$ over $\Seqv(\H)$ follow this rule.
      
      Let $\uplus$ be an abbreviation of $\biguplus_i$ where there are only two operands.
      
      For $\sum_{i\in I}\eqv{\rho_i}\leq \sum_{j\in J}\eqv{\sigma_j}$ and $\sum_{k\in K}\eqv{\gamma_k}\leq \sum_{l\in L}\eqv{\chi_l}$, notice that $\biguplus_{i\in I}\rho_i\lesssim\biguplus_{j\in J}\sigma_j$ and $\biguplus_{k\in K}\gamma_k\lesssim\biguplus_{l\in L}\chi_l.$ By \lem{S-biguplus-order} there is $\biguplus_{i\in I}\rho_i\uplus \biguplus_{k\in K}\gamma_k\lesssim \biguplus_{j\in J}\sigma_j\uplus \biguplus_{l\in L}\chi_l.$ Hence 
      \begin{align*}
          &\sum_{i\in I}\eqv{\rho_i} + \sum_{k\in K}\eqv{\gamma_k}=\eqv{\biguplus_{i\in I}\rho_i\uplus \biguplus_{k\in K}\gamma_k} \\
          \leq &\eqv{\biguplus_{j\in J}\sigma_j\uplus \biguplus_{l\in L}\chi_l} =\sum_{j\in J}\eqv{\sigma_j} + \sum_{l\in L}\eqv{\chi_l}.
      \end{align*}
      Then at $\LT(\H)$ level, the inequality holds by definition.
      \item $p\leq q\wedge r\leq s\rightarrow pr\leq qs$: Because $\T\in\LT(\H)$ is monotone, by definition this law holds.
      \item $1+pp^*\leq p^*$: For any $\T\in\LT(\H)$, there is
      \begin{align*}
          \I_{\H}+(\T; \T^*)&=\T^0+\left(\T; \sum_{i\geq 0}\T^i\right) \\
          &=\T^0 + \sum_{i\geq 0}(\T; \T^i)  \\
          &=\sum_{i\geq 0}\T^i =\T^*.
      \end{align*}
      The second equality comes from the definition of $;$ operation.
      \item $*$-continuity: the $*$-continuity condition is defined as 
  \begin{align*}
      \left(\forall n\in\nat, \sum_{0\leq i\leq n} pq^ir\leq s\right)\rightarrow pq^*r\leq s. 
  \end{align*}
\lem{S-*-continuity} leads to the $*$-continuity in $\Seqv(\H)$: for $\sum_{i\in\nat}\sum_{j\in J_i}\eqv{\rho_{ij}},$ if there exists $\sum_{k\in K}\eqv{\sigma_k}$ such that for all $n\geq 0: \sum_{0\leq i<n} \sum_{j\in J_i}\eqv{\rho_{ij}}\leq \sum_{k\in K}\eqv{\sigma_k},$ then $\sum_{i\in\nat} \sum_{j\in J_i}\eqv{\rho_{ij}}\leq \sum_{k\in K}\eqv{\sigma_k}.$
      
      Eventually we show the $*$-continuity of $\LT(\H)$. For $\T_p, \T_q, \T_r, \T_s$ satisfying $\sum_{0\leq i<n} (\T_p; \T_q^i; \T_r)\preceq\T_s$ for all $n\geq 0$, there is $\sum_{0\leq i<n}(\T_p; \T_q^i; \T_r)(\sum_{j\in J}\eqv{\rho_j})\leq \T_s(\sum_{j\in J}\eqv{\rho_j})$ for every $\sum_{j\in J}\eqv{\rho_j}\in\Seqv(\H).$ By the $*$-continuity in $\Seqv(\H)$ and linearity, inequality
      \begin{align*}
          &(\T_p; \T_q^*; \T_r)\left(\sum_{j\in J}\eqv{\rho_j}\right) \\
          =~&\sum_{i\geq 0}(\T_p; \T_q^i; \T_r)\left(\sum_{j\in J}\eqv{\rho_j}\right)\leq \T_s\left(\sum_{j\in J}\eqv{\rho_j}\right)
      \end{align*}
      holds for every $\sum_{j\in J}\eqv{\rho_j}\in\Seqv(\H).$ This concludes the $*$-continuity rule in $\LT(\H).$
      Easily we have $\O_{\H}\preceq \T$ for any $\T\in\LT(\H).$
      
      To derive the other star laws, we make use of $0\leq p$ and the $*$-continuity. For $q+pr\leq r\rightarrow p^*q\leq r,$ note $\sum_{0\leq i\leq n} 1p^iq\leq p^{n+1}r+\sum_{0\leq i\leq n} p^iq=q+p(q+p(...q+p(q+pr)...))\leq r.$ Then $*$-continuity gives $p^*q\leq r.$ The other side follows similarly.
  \end{itemize}
\end{proof}

\subsection{Detailed Proof of \lem{PL-embed}}
\label{app:lift-embed}

\begin{proof}[Proof of \lem{PL-embed}]
~

  \begin{enumerate}[label=(\roman*)]
  \item By \lem{E-preorder}, $\cp{\E}$ is monotone. Linearity is from
  \begin{align*}
  \cp{\E}\left(\sum_i \sum_{j\in J_i}\eqv{\rho_{ij}}\right)
      =\sum_i \sum_{j\in J_i}\eqv{\E(\rho_{ij})} 
      =\sum_i \cp{\E}\left(\sum_{j\in J_i}\eqv{\rho_{ij}}\right).
  \end{align*}
  \item $(\Rightarrow)$: By definition this direction holds.
  
  $(\Leftarrow)$: To prove the injectivity of path lifting, we assume $\E_1\neq\E_2$ while $\cp{\E_1}=\cp{\E_2},$ then there exists $\rho\in\PO(\H)$ such that $\E_1(\rho)\neq\E_2(\rho).$ $\cp{\E_1}=\cp{\E_2}$ indicates that 
  \begin{align*}
      &\eqv{\E_1(\rho)}=\cp{\E_1}\left(\eqv{\rho}\right) = \cp{\E_2}\left(\eqv{\rho}\right)=\eqv{\E_2(\rho)}.
  \end{align*}
  Hence $\{\E_1(\rho)\}\sim\{\E_2(\rho)\}.$ If $\E_1(\rho)=O_{\H}$, then for every $\epsilon>0,$ there is $\E_2(\rho)\sqsubseteq \epsilon I_{\H},$ resulting in $\E_2(\rho)=O_{\H}=\E_1(\rho),$ which is a contradiction. If $\E_1(\rho)\neq O_{\H},$ for every $0<\epsilon<\norm{\E_1(\rho)},$ there is $\E_1(\rho)\sqsubseteq \epsilon I_{\H}+\E_2(\rho).$ Hence $\E_1(\rho)\sqsubseteq\E_2(\rho).$ Similarly we have $\E_2(\rho)\sqsubseteq\E_1(\rho).$ So $\E_1(\rho)=\E_2(\rho)$ is the contradiction.
  \item For $\E_1, \E_2\in\CP(\H)$ and $\sum_{i\in I}\eqv{\rho_i}\in\Seqv(\H),$ there is
  \begin{align*}
      (\cp{\E_1}; \cp{\E_2})\left(\sum_{i\in I}\eqv{\rho_i}\right) &= \cp{\E_2}\left(\sum_{i\in I}\eqv{\E_1(\rho_i)}\right) \\
      &= \sum_{i\in I}\eqv{\E_2(\E_1(\rho_i))} \\
      &= \cp{\E_1\circ\E_2}\left(\sum_{i\in I}\eqv{\rho_i}\right).
  \end{align*}
  
  Similarly, if $\sum_i\E_i$ is defined in $\CP(\H)$, then $\sum_i\E_i(\rho)$ converges for any $\rho\in\PO(\H).$ By \lem{convergent-sum}, for every $\sum_{j\in J}\eqv{\rho_j}\in\Seqv(\H),$ there is
  \begin{align*}
      \left(\sum_i \cp{\E_i}\right)\left(\sum_{j\in J}\eqv{\rho_j}\right) &= \sum_i \cp{\E_i}\left(\sum_{j\in J}\eqv{\rho_j}\right) \\
      &= \sum_j\sum_i \eqv{\E_i(\rho_j)} \\
      &= \sum_j\eqv{\sum_i\E_i(\rho_j)} \\
      &= \cp{\sum_i\E_i}\left(\sum_{j\in J}\eqv{\rho_j}\right).
  \end{align*}
  \end{enumerate}
\end{proof}

\subsection{Detailed Proof of \thm{q-model-sound-complete}}
\label{app:q-model-sc}

\begin{proof}[Proof of \thm{q-model-sound-complete}]
  ($\Rightarrow$): Formally we prove it by induction on the derivation of $\NKAmodel e=f.$ Practically it suffices to prove the soundness of the NKA axioms on the quantum path model, which is proved in \thm{QS-NKA}.

  
  
  ($\Leftarrow$): We will establish $\NKAmodel e=f$ by first showing $\rpssem{e}=\rpssem{f}$ and then applying \thm{rps-sound-complete}. 
  To that end, let us consider the case of any fixed $n\in \nat$, and show that for string $w$ with length less than $n$, there is $\rpssem{e}[w]=\rpssem{f}[w].$
  
  Let $S=\{s\in\Sigma^*: |s|\leq n\}.$ Because $\Sigma$ and $n$ are finite, $S$ is a finite set. We set $\H=\mathrm{span}\{\ket{s}:s\in S\}$ which is finite dimensional, and $\eval(a)(\rho)=\sum_{s\in S}K_{a,s}\rho K_{a, s}^{\dagger},$ where $K_{a, s}=\frac{1}{\sqrt{\#_a}}\ket{sa}\bra{s}$ for $sa\in S$, $K_{a, s}=O_{\H}$ for $sa\not\in S.$
  Here $\#_a=|\{s:sa\in S\}|$ is a normalization factor to make sure $\eval(a)\in\CP(\H)$. For $s=a_1a_2\cdots a_l,$ we set $\#_s=\prod_{i=1}^l \#_{a_i}.$
  
  Let $\intp=(\H, \eval)$.
  We claim for $s\in S$ and $r\in\mathbb{R}$, there is
  \begin{align}
    \label{eq:completeness-claim}
    \Qint(e)(\eqv{r\cdot\ket{s}\bra{s}})=\sum_{st\in S}\sum_{k=1}^{\rpssem{e}[t]}\eqv{r/\#_t\cdot\ket{st}\bra{st}}.
  \end{align}
  The proof is based on the induction on expression $e$, and its proof is left to the last.
  
  Then we consider two expressions $e, f$ such that $\Qint(e)=\Qint(f)$. We apply this action on $\empstr$ and $r=1$, resulting in 
  \begin{align*}
     &\sum_{s\in S}\sum_{k=1}^{\rpssem{e}[s]}\eqv{1/\#_s\cdot\ket{s}\bra{s}}
     = \sum_{s\in S}\sum_{k=1}^{\rpssem{f}[s]}\eqv{1/\#_s\cdot\ket{s}\bra{s}}.
  \end{align*}
  If there exists $t\in S: \rpssem{e}[t]<\rpssem{f}[t],$ then there exists $m\in \nat$ such that $\rpssem{e}[t]<m\leq\rpssem{f}[t].$ By selecting $I'=\{(t, k):0\leq k<m\}$ in the definition of $\biguplus_{s\in S}\biguplus_{k=1}^{\rpssem{f}[s]} 1/\#_s\cdot\ket{s}\bra{s}\lesssim\biguplus_{s\in S}\biguplus_{k=1}^{\rpssem{e}[s]} 1/\#_s\cdot\ket{s}\bra{s},$ it is impossible to find a $J'$ to satisfy definition inequality \eq{pre-def}, because there are at most $\rpssem{e}[t]$ operators that are non-zero in basis $\ket{t}\bra{t}$. The cases where $\rpssem{e}[s]>\rpssem{f}[s]$ can be ruled out similarly. Then $\forall s\in S, \rpssem{e}[s]=\rpssem{f}[s].$
  
  Notice that the above argument holds for any $n\in\nat$. Hence $\rpssem{e}=\rpssem{f}$. By \thm{rps-sound-complete}, $\NKAmodel e=f.$
  
  Now we come back to \eq{completeness-claim}.
  Let us prove it by induction on $e$. For the base cases, notice that
  \begin{align*}
    &\Qint(0)=\O_{\H}, \qquad \Qint(1)=\I_{\H}, \\
    &\Qint(a)(\eqv{r\cdot\ket{s}\bra{s}})=\begin{cases}\eqv{r/\#_a\cdot\ket{sa}\bra{sa}}, & sa\in S, \\ \eqv{O_{\H}}, & sa\not\in S.\end{cases}
  \end{align*}
  Combined with $\rpssem{0}=0, \rpssem{1}=1\empstr$ and $\rpssem{a}=1a$, the equation holds for the base cases.

  Consider the case $e+f$. For any $s\in S$ and $r\in \mathbb{R}$, by inductive hypotheses and \lem{S-natinf},
  \begin{align*}
    &\Qint(e+f)(\eqv{r\cdot\ket{s}\bra{s}}) \\
    =~&\Qint(e)(\eqv{r\cdot\ket{s}\bra{s}})+\Qint(f)(\eqv{r\cdot\ket{s}\bra{s}})\\
    =~&\sum_{st\in S}\left(\sum_{k=1}^{\rpssem{e}[t]}\eqv{r/\#_t\cdot\ket{st}\bra{st}}+\sum_{k=1}^{\rpssem{f}[t]}\eqv{r/\#_t\cdot\ket{st}\bra{st}}\right)\\
    =~&\sum_{st\in S}\sum_{k=1}^{\rpssem{e+f}[t]}\eqv{r/\#_t\cdot\ket{st}\bra{st}}.
  \end{align*}

  Consider the case $e\cdot f$. For any $s\in S$ and $r\in \mathbb{R}$, by inductive hypotheses and \lem{S-natinf},
  \begin{align*}
    &\Qint(e\cdot f)(\eqv{r\cdot\ket{s}\bra{s}})\\
    =~&\Qint(f)(\Qint(e)(\eqv{r\cdot\ket{s}\bra{s}}))\\
    =~&\Qint(f)\left(\sum_{st\in S}\sum_{k=1}^{\rpssem{e}[t]}\eqv{r/\#_t\cdot\ket{st}\bra{st}}\right)\\
    =~&\sum_{stw\in S}\sum_{k=1}^{\rpssem{e}[t]}\sum_{l=1}^{\rpssem{f}[w]}\eqv{r/(\#_t\cdot\#_w)\cdot\ket{stw}\bra{stw}}\\
    =~&\sum_{st\in S}\sum_{k=1}^{\rpssem{e\cdot f}[t]} \eqv{r/\#_t\cdot\ket{st}\bra{st}}.
  \end{align*}

  Consider the case $e^*$. For any $s\in S$, by inductive hypothesis, \lem{S-natinf} and the above proofs for $e+f$ and $e\cdot f$,
  \begin{align*}
    &\Qint(e^*)(\eqv{r\cdot\ket{s}\bra{s}}) \\
    =~&\Qint(e)^*(\eqv{r\cdot\ket{s}\bra{s}})\\
    =~&\sum_{i\geq 0}\Qint(e)^i(\eqv{r\cdot\ket{s}\bra{s}})\\
    =~&\sum_{i\geq 0}\Qint(e^i)(\eqv{r\cdot\ket{s}\bra{s}})\\
    =~&\sum_{i\geq 0}\sum_{st\in S}\sum_{k=1}^{\rpssem{e^i}[t]}\eqv{r/\#_t\cdot\ket{st}\bra{st}}\\
    =~&\sum_{st\in S}\sum_{k=1}^{\rpssem{e^*}[t]} \eqv{r/\#_t\cdot\ket{st}\bra{st}}.
  \end{align*}
\end{proof}

\subsection{Detailed Proof of \thm{enc-recovery}}
\label{app:enc-recovery}

\begin{proof}[Proof of \thm{enc-recovery}]
  We prove them by induction on $P$.
  
  \begin{itemize}
    \item For the base cases $P\equiv\cskip, \cfail$, the equation holds by definition. For $P\equiv q:=\ket{0}$ and $\overline{q}:=U[\overline{q}]$, we know $\enc(P)\in\Sigma$ by the encoder setting $E$. With $E^{-1}(\enc(P))=\cp{\sem{P}},$ the equation holds. 
    \item For $P=P_1;P_2$, by inductive hypotheses there are $\Qint(\enc(P_1))=\cp{\sem{P_1}}$ and $\Qint(\enc(P_2))=\cp{\sem{P_2}}.$ 
    Then by \lem{struct-preserving},
    \begin{align*}
      \Qint(\enc(P))&=\Qint(\enc(P_1)); \Qint(\enc(P_2)) \\
      &=\cp{\sem{P_1}}; \cp{\sem{P_2}}
      =\cp{\sem{P_1}\circ\sem{P_2}}.
    \end{align*}
    \item For $P\equiv \qif{M[\overline{q}]\xrightarrow{i}P_i}$, the inductive hypotheses are $\Qint(\enc(P_i))=\cp{\sem{P_i}}.$
    Then by \lem{struct-preserving},
    \begin{align*}
      \Qint(\enc(P))&=\sum_i(\Qint(E(\M_i)); \Qint(\enc(P_i)) \\
      &=\sum_i(\cp{\M_i}; \cp{\sem{P_i}})
      =\sum_i(\cp{\M_i\circ\sem{P_i}}) \\
      &=\cp{\sum_i(\M_i\circ\sem{P_i})}.
    \end{align*}
    \item For $P\equiv \qwhile{M[\overline{q}]=1}{S}$, the inductive hypothesis becomes $\Qint(\enc(S))=\cp{\sem{S}}$. By \cite{Ying16} $\sum_{n\geq 0}((\M_1\circ\sem{S})^n\circ\M_0)$ exists in $\CP(\H)$, so by \lem{struct-preserving} and linearity of transformations in $\LT(\H)$,
    \begin{align*}
      \Qint(\enc(P)) &=(\Qint(E(\M_1)); \Qint(\enc(S)))^*\Qint(E(\M_0)) \\
      &=(\cp{\M_1}; \cp{\sem{S}})^*; \cp{\M_0}  \\
      &=\left(\sum_{n\geq 0} (\cp{\M_1}; \cp{\sem{S}})^n\right); \cp{\M_0} \\
      &=\sum_{n\geq 0} ((\cp{\M_1}; \cp{\sem{S}})^n; \cp{\M_0})\\
      &=\sum_{n\geq 0} \cp{(\M_1\circ \sem{S})^n\circ \M_0} \\
      &=\cp{\sum_{n\geq 0}((\M_1\circ \sem{S})^n\circ \M_0)}.
    \end{align*}
  \end{itemize}
\end{proof}

\subsection{Detailed Proof of \thm{normal-form}}
\label{app:normal-form}

\begin{proof}[Proof of \thm{normal-form}]
  We prove the normal form theorem by induction on the program $P$. For each step we introduce a classical guard variable $g$ whose value is limited in a finite set $\{0, 1, ..., n-1\}$, and denote the space of $g$ by $\C_n$. We encode $g:=\ket{i}$ as $g^i$, the measurement $\mathrm{Meas}[g]=i$ as $g_i$ and the reset of space $\C$ as $c$. Each time $g$ is independent of the existing space, so the following assumptions hold for any $i, j$ in the value set:
  \begin{itemize}
  \item $g^i$ commutes with every elements except for $g^j.$
  \item $g^ig_j=\delta_{ij}g^i,$ where $\delta_{ij}=1$ when $i=j,$ and $\delta_{ij}=0$ when $i\neq j$.
  \item $g^ig^j=g^j.$
  \end{itemize}

  \bigskip
  (a) For the base case where $P=\cskip ~|~ \cfail ~|~ q:=\ket{0} ~|~ \overline{q}=U[\overline{q}],$ they are while-free. Let $\C=\C_1$ the space with only one value. We claim $P; g:=\ket{0}; \qwhile{\text{Meas}[g]=1}{\cskip}; g:=\ket{0}$ is equivalent to $P; g:=\ket{0}$. The NKA encoding of these two programs are $pg^0(g_11)^*g_0g^0$ and $pg^0$. This motivates the following derivation:
  \begin{align*}
    g^0(g_11)^*g_0=~&g^0g_0+g^0g_1g_1^*g_0 =g^0.
  \end{align*}
  Hence $pg^0(g_11)^*g_0g^0=pg^0g^0=pg^0.$

  \bigskip
  (b) For the $S_1; S_2$ case, by inductive hypothesis we have two external space $\C^1$ and $\C^2$ such that $S_i; p_{\C^i}:=\ket{0}$ is equivalent to $P_{i0}; \qwhile{M_i}{P_{i1}}; p_{\C^i}:=\ket{0}$, where $P_{ij}$ is while-free. We claim $S_1; S_2; p_{\C^1\otimes\C^2\otimes\C_3}:=\ket{0}$ and
  \begin{align*}
    &P_{10}; g:=\ket{1}; \\
    &\mathbf{while~}\mathrm{Meas}[g]>0\mathbf{~do} \\
    &\quad \mathbf{if~}\mathrm{Meas}[g]=1\mathbf{~then} \\
    &\quad \quad \mathbf{if~}M_1\mathbf{~then~} P_{11} \\
    &\quad \quad \mathbf{else~} P_{20}; g:=\ket{2} \\
    &\quad \mathbf{else} \\
    &\quad \quad \mathbf{if~}M_2\mathbf{~then~} P_{21} \\
    &\quad \quad \mathbf{else~} g:=\ket{0} \\
    &\mathbf{done}; \\
    &p_{\C^1\otimes\C^2\otimes\C_3}:=\ket{0},
  \end{align*}
  are equivalent, whose encodings are $s_1s_2c_1c_2g^0$ and $p_{10}g^1((g_1+g_2)(g_1(m_{11}p_{11}+m_{12}p_{20}g^2)+(g_0+g_2)(m_{21}p_{21}+m_{22}g^0)))^*g_0c_1c_2g^0.$

  Notice that $c_1$ acts on $\C^1$, so $c_1$ is commutable to those operators acting on $\H\otimes\C^2\otimes\C_3$. By inductive hypothesis, there is $s_ic_i=p_{i0}(m_{i1}p_{i1})^*m_{i2}c_i$, so
  \begin{align*}
    s_1s_2c_1c_2g^0&=s_1c_1s_2c_2g^0\\
    &=p_{10}(m_{11}p_{11})^*m_{12}c_1p_{20}(m_{21}p_{21})^*m_{22}c_2g^0\\
    &=p_{10}(m_{11}p_{11})^*m_{12}p_{20}(m_{21}p_{21})^*m_{22}c_1c_2g^0.
  \end{align*}

  Let $X=(g_1+g_2)(g_1(m_{11}p_{11}+m_{12}p_{20}g^2)+(g_0+g_2)(m_{21}p_{21}+m_{22}g^0))=g_1(m_{11}p_{11}+m_{12}p_{20}g^2)+g_2(m_{21}p_{21}+m_{22}g^0),$ and $Y=g_1(m_{11}p_{11}+m_{12}p_{20}g^2).$ Then by denesting rule:
  \begin{align*}
    g^1X^* =~& g^1(g_1(m_{11}p_{11}+m_{12}p_{20}g^2))^* \\
    &\quad \cdot (g_2(m_{21}p_{21}+m_{22}g^0)(g_1(m_{11}p_{11}+m_{12}p_{20}g^2))^*)^* \\
    =~& g^1Y^*(g_2(m_{21}p_{21}+m_{22}g^0)Y^*)^*. \\
    g^1Y^* =~& g^1(g_1m_{11}p_{11})^*(g_1m_{12}p_{20}g^2(g_1m_{11}p_{11})^*)^* \\
    =~& (m_{11}p_{11})^*g^1 \\
    &\quad \cdot(g_1m_{12}p_{20}g^2+g_1m_{12}p_{20}g^2g_1m_{11}p_{11}(g_1m_{11}p_{11})^*))^* \\
    =~& (m_{11}p_{11})^*g^1(g_1m_{12}p_{20}g^2)^* \\
    =~& (m_{11}p_{11})^*g^1(1+g_1m_{12}p_{20}g^2 \\
    &\qquad \qquad \quad  +g_1m_{12}p_{20}g^2g_1m_{12}p_{20}g^2(g_1m_{12}p_{20}g^2)^*) \\
    =~& (m_{11}p_{11})^*(g^1+m_{12}p_{20}g^2). \\
    g^2Y^* =~& g^2.
  \end{align*}
  By \ref{star-rew-rule} , we have:
  \begin{align*}
    & g^2(g_2(m_{21}p_{21}+m_{22}g^0)Y^*)^* \\
    =~& g^2(g_2m_{21}p_{21}Y^*)^*(g_2m_{22}g^0Y^*(g_2m_{21}p_{21}Y^*)^*)^* \\
    =~& g^2(g_2m_{21}p_{21})^*(g_2m_{22}g^0+g_2m_{22}g^0g_2m_{21}p_{21}Y^*(g_2m_{21}p_{21}Y^*)^*) \\
    =~& (m_{21}p_{21})^*g^2(g_2m_{22}g^0)^* \\
    =~& (m_{21}p_{21})^*g^2(1+g_2m_{22}g^0+g_2m_{22}g^0(g_2m_{22}g^0)^*) \\
    =~& (m_{21}p_{21})^*(g^2+m_{22}g^0).
  \end{align*}

  Hence we have:
  \begin{align*}
    &  p_{10}g^1((g_1+g_2)(g_1(m_{11}p_{11}+m_{12}p_{20}g^2) \\
    & \qquad \qquad \qquad +(g_0+g_2)(m_{21}p_{21}+m_{22}g^0)))^*g_0c_1c_2g^0 \\
    =~& p_{10}(m_{11}p_{11})^*(g^1+m_{12}p_{20}g^2)(g_2(m_{21}p_{21}+m_{22}g^0)Y^*)^*g_0c_1c_2g^0 \\
    =~& p_{10}(m_{11}p_{11})^*g^1g_0c_1c_2g^0 \\
    &+p_{10}(m_{11}p_{11})^*m_{12}p_{20}g^2(g_2(m_{21}p_{21}+m_{22}g^0)Y^*)^*g_0c_1c_2g^0 \\
    =~& p_{10}(m_{11}p_{11})^*m_{12}p_{20}(m_{21}p_{21})^*(g^2+m_{22}g^0)g_0c_1c_2g^0 \\
    =~& p_{10}(m_{11}p_{11})^*m_{12}p_{20}(m_{21}p_{21})^*m_{22}c_1c_2g^0 \\
    =~& s_1s_2c_1c_2g^0.
  \end{align*}

  \bigskip
  (c) For the $\qif{M\xrightarrow{i}S_i}$ case, w.l.o.g. we assume the measurement results are $\{1, 2, ..., n\}$. By inductive hypothesis we have two external spaces $\{\C^i\}_{1\leq i\leq n}$ such that $S_i; p_{\C^i}:=\ket{0}$ is equivalent to $P_{i0}; \qwhile{M_i}{P_{i1}}; p_{\C^i}:=\ket{0}$, where $P_{ij}$ is while-free. Let $\C=\left(\bigotimes_{1\leq i\leq n}\C^i\right)\otimes \C_{n+1}$. We claim $\qif{M\xrightarrow{i}S_i}; p_C=\ket{0}$ and
  \begin{align*}
    &\qif{M\xrightarrow{i} P_{i0}; g:=\ket{i}} \\
    &\mathbf{while~}\mathrm{Meas}[g]>0\mathbf{~do} \\
    &\quad \mathbf{case~}\mathrm{Meas}[g]\xrightarrow{i>0} \\
    &\quad \quad \mathbf{if~}M_i\mathbf{~then~}P_{i1} \\
    &\quad \quad \mathbf{else~}g:=\ket{0} \\
    &\quad \mathbf{end} \\
    &\mathbf{done}; \\
    &p_{\C}:=\ket{0}
  \end{align*}
  are equivalent, whose encodings are $(\sum_{i=1}^n m_is_i)(\prod_{i=1}^n c_i)g^0$ and 
  \begin{align*}
     \left(\sum_{i=1}^n m_ip_{i0}g^i\right)\left(\left(\sum_{i=1}^n g_i\right)\left(\sum_{i=1}^n g_i(m_{i1}p_{i1}+m_{i2}g^0)\right)\right)^*g_0\left(\prod_{i=1}^n c_i\right)g^0.
  \end{align*}
  
  First we show $\qif{M\xrightarrow{i}S_i}; p_C=\ket{0}$ is equivalent to
  \begin{align*}
    &\qif{M\xrightarrow{i}S_i; p_{C^i}:=\ket{0}}; \\
    &p_C=\ket{0}.
  \end{align*}
  $\left(\sum_{1\leq i\leq n} m_is_i\right)\left(\prod_{i=1}^n c_i\right)g^0=\left(\sum_{1\leq i\leq n} m_is_ic_i\right)\left(\prod_{i=1}^n c_i\right)g^0$ is what we need to derive. Because each $\C^i$ and $C_3$ are disjoint, we have $c_i$ commutes each other for $1\leq i\leq n$, and $c_ic_i=c_i$. With these assumptions added, the two expressions are equivalent by distributive law.

  Then we could apply inductive hypothesis $p_{i0}(m_{i1}p_{i1})^*m_{i2}c_i=s_ic_i$ on each branch. Let $X=\left(\sum_{i=1}^n g_i\right)\left(\sum_{i=1}^n g_i(m_{i1}p_{i1}+m_{i2}g^0)\right)=\sum_{i=1}^n g_i(m_{i1}p_{i1}+m_{i2}g^0), Y_i=g_im_{i1}p_{i1}+g_im_{i2}g^0$ for convenience. By denesting rule:
  \begin{align*}
    g^iX^*&=g^iY_i^*\left(\left(\sum_{j\neq i}g_j(m_{j1}p_{j1}+m_{j2}g^0)\right)Y_i^*\right)^*.
  \end{align*}
  Notice that for $1\leq i\leq n$,
  \begin{align*}
    g^iY_i^* =~& g^i(g_im_{i1}p_{i1})^*(g_im_{i2}g^0(g_im_{i1}p_{i1})^*)^* \\
    =~& (m_{i1}p_{i1})^*g^i(g_im_{i2}g^0+g_im_{i2}g^0g_im_{i1}p_{i1}(g_im_{i1}p_{i1})^*)^* \\
    =~& (m_{i1}p_{i1})^*g^i(g_im_{i2}g^0)^* \\
    =~& (m_{i1}p_{i1})^*g^i(1+g_im_{i2}g^0+g_im_{i2}g^0g_im_{i2}g^0(g_im_{i2}g^0)^*) \\
    =~& (m_{i1}p_{i1})^*(g^i+m_{i2}g^0),
  \end{align*}
  Meanwhile, for all $1\leq i\leq n$,
  \begin{align*}
    g^0&\left(\left(\sum_{j\neq i}g_j(m_{j1}p_{j1}+m_{j2}g^0)\right)Y_i^*\right)^* = g^0, \\
    g^i&\left(\left(\sum_{j\neq i}g_j(m_{j1}p_{j1}+m_{j2}g^0)\right)Y_i^*\right)^* = g^i. 
  \end{align*}
  Combining them up results in $g^iX^* = (m_{i1}p_{i1})^*(g^i+m_{i2}g^0)$ for $1\leq i\leq n$. Thus
  \begin{align*}
      &\left(\sum_{i=1}^n m_ip_{i0}g^i\right)\left(\left(\sum_{i=1}^n g_i\right)\left(\sum_{i=1}^n g_i(m_{i1}p_{i1}+m_{i2}g^0)\right)\right)^* \\
      &~\cdot g_0\left(\prod_{i=1}^n c_i\right)g^0 \\
      =~&\left(\sum_{i=1}^n m_ip_{i0}g^iX^*\right)g_0\left(\prod_{i=1}^n c_i\right)g^0 \\
      =~&\left(\sum_{i=1}^n m_ip_{i0}(m_{i1}p_{i1})^*(g^i+m_{i2}g^0)\right)g_0\left(\prod_{i=1}^n c_i\right)g^0\\
      =~&\left(\sum_{i=1}^n m_ip_{i0}(m_{i1}p_{i1})^*m_{i2}g^0\right)g_0\left(\prod_{i=1}^n c_i\right)g^0 \\
      =~&\left(\sum_{i=1}^n m_ip_{i0}(m_{i1}p_{i1})^*m_{i2}c_i\right)\left(\prod_{i=1}^n c_i\right)g^0 \\
      =~&\left(\sum_{i=1}^n m_is_ic_i\right)\left(\prod_{i=1}^n c_i\right)g^0 \\
      =~&\left(\sum_{i=1}^n m_is_i\right)\left(\prod_{i=1}^n c_i\right)g^0.
  \end{align*}

  \bigskip
  (d) For the $\qwhile{M_1}{S}$ case, by inductive hypothesis we have $\C$ such that $S; p_{\C}:=\ket{0}$ is equivalent to $P_1; \qwhile{M_2}{P_2}; p_{\C}:=\ket{0}$, where $P_i$ is while-free.

  We claim $\qwhile{M_1}{S}; p_{\C\otimes\C_3}:=\ket{0}$ and
  \begin{align*} 
    &g:=\ket{1}; \\
    &\mathbf{while~}\mathrm{Meas}[g]>0\mathbf{~do} \\
    &\quad \mathbf{if~}\mathrm{Meas}[g]=1\mathbf{~then} \\
    &\quad \quad \mathbf{if~} M_1 \mathbf{~then~} P_1; g:=\ket{2} \\
    &\quad \quad \mathbf{else~} g:=\ket{0} \\
    &\quad \mathbf{else~} \\
    &\quad \quad \mathbf{if~} M_2 \mathbf{~then~} P_2 \\
    &\quad \quad \mathbf{else~} g:=\ket{1} \\
    &p_{\C\otimes\C_3}:=\ket{0},
  \end{align*}
  are equivalent, whose encodings are $(m_{11}s)^*m_{12}cg^0$ and $g^1((g_1+g_2)(g_1(m_{11}p_1g^2+m_{12}g^0)+(g_0+g_2)(m_{21}p_2+m_{22}g^1)))^*g_0cg^0.$

  Similarly to the above case, utilizing inductive hypothesis, we have $sc=p_1(m_{21}p_2)^*m_{22}c.$ Let $X=(g_1+g_2)(g_1(m_{11}p_1g^2+m_{12}g^0)+(g_0+g_2)(m_{21}p_2+m_{22}g^1))=g_1(m_{11}p_1g^2+m_{12}g^0)+g_2(m_{21}p_2+m_{22}g^1).$ By denesting rule:
  \begin{align*}
    g^1X^* &= g^1(g_1(m_{11}p_1g^2+m_{12}g^0))^* \\
    &\quad \cdot (g_2(m_{21}p_2+m_{22}g^1)(g_1(m_{11}p_1g^2+m_{12}g^0))^*)^* 
  \end{align*}
  Let $Y=g_1(m_{11}p_1g^2+m_{12}g^0), Z=m_{11}p_1(m_{21}p_2)^*m_{22}.$ So
  \begin{align*}
    g^1X^*=~&g^1Y^*(g_2(m_{21}p_2+m_{22}g^1)Y^*)^*. \\
    g^1Y^* =~& g^1(1+g_1(m_{11}p_1g^2+m_{12}g^0)(g_1(m_{11}p_1g^2+m_{12}g^0))^*) \\
    =~& g^1+m_{12}g^0+m_{11}p_1g^2 \\
    g^2Y^* =~& g^2 \\
  \end{align*}
  Hence $g^2(g_2m_{21}p_2Y^*)^* = g^2(g_2m_{21}p_2)^* = (m_{21}p_2)^*g^2.$ By \ref{star-rew-rule}, there is
  \begin{align*}
    & g^2(g_2(m_{21}p_2+m_{22}g^1)Y^*)^*g_0 \\
    =~& g^2(g_2m_{21}p_2Y^*)^*(g_2m_{22}g^1Y^*(g_2m_{21}p_2Y^*)^*)^*g_0 \\
    =~& (m_{21}p_2)^*g^2(g_2m_{22}(g^1+m_{12}g^0+m_{11}p_1g^2)(g_2m_{21}p_2Y^*)^*)^*g_0 \\
    =~& (m_{21}p_2)^*g^2(g_2m_{22}(g^1+m_{12}g^0)+g_2m_{22}m_{11}p_1(m_{21}p_2)^*g^2)^*g_0 \\
    =~& (m_{21}p_2)^*g^2(g_2m_{22}m_{11}p_1(m_{21}p_2)^*g^2)^* \\
    & \quad \cdot (g_2m_{22}(g^1+m_{12}g^0)(g_2m_{22}m_{11}p_1(m_{21}p_2)^*g^2)^*)^*g_0 \\
    =~& (m_{21}p_2)^*(m_{22}m_{11}p_1(m_{21}p_2)^*)^*g^2(g_2m_{22}(g^1+m_{12}g^0))^*g_0
  \end{align*}
  Expand the star expression twice:
  \begin{align*}
    & g^2(g_2m_{22}(g^1+m_{12}g^0))^*g_0 \\
    =~& g^2[1+g_2m_{22}(g^1+m_{12}g^0) \\
    &+g_2m_{22}(g^1+m_{12}g^0)g_2m_{22}(g^1+m_{12}g^0)(g_2m_{22}(g^1+m_{12}g^0))^*]g_0 \\
    =~& g^2(1+g_2m_{22}(g^1+m_{12}g^0))g_0 \\
    =~& m_{22}m_{12}g^0. 
  \end{align*}
  By \ref{sliding-rule}
  \begin{align*}
    & g^2(g_2(m_{21}p_2+m_{22}g^1)Y^*)^*g_0 \\
    =~& (m_{21}p_2)^*(m_{22}m_{11}p_1(m_{21}p_2)^*)^*m_{22}m_{12}g^0\\
    =~& (m_{21}p_2)^*m_{22}(m_{11}p_1(m_{21}p_2)^*m_{22})^*m_{12}g^0 \\
    =~& (m_{21}p_2)^*m_{22}Z^*m_{12}g^0.
  \end{align*}
  Combining them up gives
  \begin{align*}
    g^1X^*g_0 =~& (g^1+m_{12}g^0+m_{11}p_1g^2)(g_2(m_{21}p_2+m_{22}g^1)Y^*)^*g_0 \\
    =~& m_{12}g^0+m_{11}p_1g^2(g_2(m_{21}p_2+m_{22}g^1)Y^*)^*g_0 \\
    =~& (1+m_{11}p_1(m_{21}p_2)^*m_{22}Z^*)m_{12}g^0 \\
    =~& (1+ZZ^*)m_{12}g^0 \\
    =~& Z^*m_{12}g^0
  \end{align*}
  Hence we have
  \begin{align*}
    &  g^1((g_1+g_2)(g_1(m_{11}p_1g^2+m_{12}g^0) \\
    & \quad +(g_0+g_2)(m_{21}p_2+m_{22}g^1)))^*g_0cg^0 \\
    =~& g^1X^*g_0cg^0 \\
    =~& Z^*m_{12}cg^0 \\
    =~& (m_{11}p_1(m_{21}p_2)^*m_{22})^*cm_{12}g^0 \\
    =~& (m_{11}p_1(m_{21}p_2)^*m_{22}c)^*cm_{12}g^0 \\
    =~& (m_{11}sc)^*cm_{12}g^0 \\
    =~& (m_{11}s)^*m_{12}cg^0.
  \end{align*}
\end{proof}

\subsection{Proof of \lem{p-is-ea}, \thm{q-nkat}}
\label{app:ea-proofs}

\begin{proof}[Proof of \lem{p-is-ea}]
~

  \begin{enumerate}
  \item If $\cp{\C_A}\oplus \cp{\C_B}$ is defined, then $\cp{\C_A}+\cp{\C_B}\preceq\cp{\C_{I_{\H}}}.$ Commutativity of addition makes $\cp{\C_B}+\cp{\C_A}\preceq\cp{\C_{I_{\H}}}$ and leads to $\cp{\C_B}\oplus \cp{\C_A}=\cp{\C_B}+\cp{\C_A}.$
  \item If $\cp{\C_A}\oplus \cp{\C_B}$ and $(\cp{\C_A}\oplus \cp{\C_B})\oplus \cp{\C_C}$ are defined, then $\cp{\C_A}+\cp{\C_B}\preceq\cp{\C_{I_{\H}}}$ and $(\cp{\C_A}+\cp{\C_B})+\cp{\C_C}\preceq\cp{\C_{I_{\H}}}.$ Hence $\cp{\C_B}+\cp{\C_C}\preceq\cp{\C_{I_{\H}}}$ and $\cp{\C_A}+(\cp{\C_B}+\cp{\C_C})\preceq\cp{\C_{I_{\H}}}.$ By definition, $\cp{\C_B}+\cp{\C_C}$ and $\cp{\C_A}+(\cp{\C_B}+\cp{\C_C})$ are defined, and $\cp{\C_A}+(\cp{\C_B}+\cp{\C_C})=(\cp{\C_A}+\cp{\C_B})+\cp{\C_C}.$
  \item If $\cp{\C_A}\oplus \cp{\C_{I_{\H}}}$ is defined in $\pred(\H)$, we assume $\cp{\C_A}+\cp{\C_{I_{\H}}}=\cp{\C_B}.$ Apply the quantum actions on $[\ket{0}\bra{0}],$ we have $[A+I_{\H}]=[B]$. Meanwhile, $\norm{A}, \norm{B}\leq 1.$ This forces $A=0$ so $\cp{\C_A}=\cp{\C_{O_{\H}}}=\O_{\H}.$
  \item For $\cp{\C_A}\in\pred(\H),$ there is $\cp{\C_A}+\cp{\C_{\neg{A}}}=\cp{\C_{I_{\H}}},$ hence $\cp{\C_A}\oplus\cp{\C_{\neg{A}}}=\cp{\C_{I_{\H}}}.$ Meanwhile, if $\cp{\C_A}\oplus\cp{\C_B}=\cp{\C_{I_{\H}}},$ we apply these quantum actions on $[\ket{0}\bra{0}],$ resulting in $[A+B]=[I_{\H}].$ Hence $B=I-A=\neg{A}.$ That is, $\neg{\cp{\C_A}}=\cp{\C_{\neg{A}}}$ is the unique negation of $\cp{\C_A}$ in $\pred(\H).$
  \item For $\cp{\C_A}\in\pred(\H)$, $\O_{\H}+\cp{\C_A}=\cp{\C_A}\preceq\cp{\C_{I_{\H}}}$, whose left hand side then equals to $\O_{\H}\oplus \cp{\C_A}$ by definition. 
  \end{enumerate}
\end{proof}

\begin{proof}[Proof of \thm{q-nkat}]
~
  Notice that the NKA axioms are symmetric for operands of $\cdot$. That is, if we define $a\star b=b\cdot a,$ any axiom substituting $\star$ for $\cdot$ has a corresponding axiom. Hence if $(\K, +, \cdot, *, 0, 1)$ forms an NKA, $(\K, +, \star, *, 0, 1)$ also forms an NKA. Hence \thm{QS-NKA} has verified (1) in \defn{nkat}.
  
  Meanwhile, \lem{p-is-ea} has verified (2) in \defn{nkat}. We only need to verify (3) here.
  \begin{itemize}
      \item By definition, $\cp{\M_i^{\dagger}}$ are elements of $\LT(\H).$
      \item Note $\cp{\M_i^{\dagger}}\diamond\cp{\C_A}(\sum_j[\rho_j])=\sum_j[\tr(\rho_j)\M_i^{\dagger}A\M_i]=\cp{\C_{\M_i^{\dagger}A\M_i}}(\sum_j[\rho_j]).$ Hence $\cp{\M_i^{\dagger}}\diamond\cp{\C_A}=\cp{\C_{\M_i^{\dagger}A\M_i}}$ and it is in $\pred(\H).$
      \item Similarly, we have
      \begin{align*}
          &\left(\sum_{i\in I}\cp{\M_i^{\dagger}}\diamond\cp{\C_{I_{\H}}}\right)\left(\sum_j[\rho_j]\right) \\
          =~~&\sum_j\left[\tr(\rho_j)\sum_{i\in I}\M_i^{\dagger}\M_i\right] \\
          =~~&\sum_j\left[\tr(\rho_j)I_{\H}\right] = \cp{\C_{I_{\H}}}\left(\sum_j[\rho_j]\right).
      \end{align*}
      This gives $\left(\sum_{i\in I}\cp{\M_i^{\dagger}}\diamond\cp{\C_{I_{\H}}}\right)=\cp{\C_{I_{\H}}}.$
  \end{itemize}
\end{proof}

\subsection{Proof of \lem{nkat-rules}}
\label{app:nkat-rules}

\begin{proof}[Proof of \lem{nkat-rules}]
~

  \begin{itemize}
   \item ($0\leq a\leq e$): Notice that $0\oplus a=a$ is defined by the definition of effect algebra. There is $0\leq 0+a=a\leq e.$
   \item ($a+\neg{a}=e$): Because $a\oplus \neg{a}=e$ is defined, we have $e=a\oplus \neg{a}=a+\neg{a}.$
   \item ($\neg{\neg{a}}=a$): Notice that there exists a unique $\neg{a}\in \L$ satisfying $a\oplus\neg{a}=e.$ Then there exists a unique $\neg{\neg{a}}$ satisfying $\neg{\neg{a}}\oplus\neg{a}=e.$ Therefore $a=\neg{\neg{a}}.$
  \item (\hyperref[negrev-rule]{negation-reverse}): Because $a\leq b$, $0\leq a+\neg{b}\leq b+\neg{b}=e.$ Hence $a\oplus\neg{b}\in\L.$ Let $c=\neg{a\oplus\neg{b}}\in\L,$ there is $0\leq c.$ So $a\oplus\neg{b}\oplus c=e=a\oplus\neg{a}.$ Thence $\neg{a}=\neg{b}\oplus c=\neg{b}+c,$ and $\neg{a}\leq\neg{b}.$
  \item (\hyperref[summa-rule]{partition-transform}): By $0\leq a_i\leq e$, monotone properties and $m_i a_i\in\L$, $\sum_{i\in I}m_ia_i\leq\sum_{i\in I}m_ie=e.$ So $\bigoplus_{i\in I}m_ia_i=\sum_{i\in I}m_ia_i\in \L$ by the definition of $\oplus$. Similarly $\bigoplus_{i\in I}m_i\neg{a_i}=\sum_{i\in I}m_i\neg{a_i}\in \L.$ Adding them together, $e=\sum_{i\in I}m_i e=\sum_{i\in I}m_i(a_i+\neg{a_i})=\sum_{i\in I}m_ia_i+\sum_{i\in I}m_i\neg{a_i}.$ Hence $\neg{\sum_{i\in I}m_ia_i}=\sum_{i\in I}m_i\neg{a_i}.$
  \end{itemize}
\end{proof}

%% file: main.bbl

\begin{thebibliography}{73}


\ifx \showCODEN    \undefined \def \showCODEN     #1{\unskip}     \fi
\ifx \showDOI      \undefined \def \showDOI       #1{#1}\fi
\ifx \showISBNx    \undefined \def \showISBNx     #1{\unskip}     \fi
\ifx \showISBNxiii \undefined \def \showISBNxiii  #1{\unskip}     \fi
\ifx \showISSN     \undefined \def \showISSN      #1{\unskip}     \fi
\ifx \showLCCN     \undefined \def \showLCCN      #1{\unskip}     \fi
\ifx \shownote     \undefined \def \shownote      #1{#1}          \fi
\ifx \showarticletitle \undefined \def \showarticletitle #1{#1}   \fi
\ifx \showURL      \undefined \def \showURL       {\relax}        \fi
\providecommand\bibfield[2]{#2}
\providecommand\bibinfo[2]{#2}
\providecommand\natexlab[1]{#1}
\providecommand\showeprint[2][]{arXiv:#2}

\bibitem[\protect\citeauthoryear{Abhari, Faruque, Dousti, Svec, Catu,
  Chakrabati, Chiang, Vanderwilt, Black, Chong, Martonosi, Suchara, Brown,
  Pedram, and Brun}{Abhari et~al\mbox{.}}{2012}]%
        {Sca12}
\bibfield{author}{\bibinfo{person}{Ali~Javadi Abhari}, \bibinfo{person}{Arvin
  Faruque}, \bibinfo{person}{Mohammad~Javad Dousti}, \bibinfo{person}{Lukas
  Svec}, \bibinfo{person}{Oana Catu}, \bibinfo{person}{Amlan Chakrabati},
  \bibinfo{person}{Chen-Fu Chiang}, \bibinfo{person}{Seth Vanderwilt},
  \bibinfo{person}{John Black}, \bibinfo{person}{Fred Chong},
  \bibinfo{person}{Margaret Martonosi}, \bibinfo{person}{Martin Suchara},
  \bibinfo{person}{Ken Brown}, \bibinfo{person}{Massoud Pedram}, {and}
  \bibinfo{person}{Todd Brun}.} \bibinfo{year}{2012}\natexlab{}.
\newblock \bibinfo{booktitle}{\emph{Scaffold: Quantum Programming Language}}.
\newblock \bibinfo{type}{{T}echnical {R}eport} TR-934-12.
  \bibinfo{institution}{Princeton University}.
\newblock


\bibitem[\protect\citeauthoryear{Aleksandrowicz, Alexander, Barkoutsos, Bello,
  Ben-Haim, Bucher, et~al\mbox{.}}{Aleksandrowicz et~al\mbox{.}}{2019}]%
        {aleksandrowicz2019qiskit}
\bibfield{author}{\bibinfo{person}{Gadi Aleksandrowicz},
  \bibinfo{person}{Thomas Alexander}, \bibinfo{person}{Panagiotis Barkoutsos},
  \bibinfo{person}{Luciano Bello}, \bibinfo{person}{Yael Ben-Haim},
  \bibinfo{person}{David Bucher}, {et~al\mbox{.}}}
  \bibinfo{year}{2019}\natexlab{}.
\newblock \bibinfo{title}{{Qiskit: An Open-source Framework for Quantum
  Computing}}.
\newblock
\newblock


\bibitem[\protect\citeauthoryear{Anderson, Foster, Guha, Jeannin, Kozen,
  Schlesinger, and Walker}{Anderson et~al\mbox{.}}{2014}]%
        {AFGJKSW13a}
\bibfield{author}{\bibinfo{person}{Carolyn~Jane Anderson},
  \bibinfo{person}{Nate Foster}, \bibinfo{person}{Arjun Guha},
  \bibinfo{person}{Jean-Baptiste Jeannin}, \bibinfo{person}{Dexter Kozen},
  \bibinfo{person}{Cole Schlesinger}, {and} \bibinfo{person}{David Walker}.}
  \bibinfo{year}{2014}\natexlab{}.
\newblock \showarticletitle{{NetKAT}: Semantic Foundations for Networks}. In
  \bibinfo{booktitle}{\emph{Proc. 41st ACM SIGPLAN-SIGACT Symp. Principles of
  Programming Languages (POPL'14)}}. ACM, \bibinfo{address}{San Diego,
  California, USA}, \bibinfo{pages}{113--126}.
\newblock


\bibitem[\protect\citeauthoryear{Angus and Kozen}{Angus and Kozen}{2001}]%
        {AK01a}
\bibfield{author}{\bibinfo{person}{Allegra Angus} {and} \bibinfo{person}{Dexter
  Kozen}.} \bibinfo{year}{2001}\natexlab{}.
\newblock \bibinfo{booktitle}{\emph{Kleene Algebra with Tests and Program
  Schematology}}.
\newblock \bibinfo{type}{{T}echnical {R}eport} TR2001-1844.
  \bibinfo{institution}{Computer Science Department, Cornell University}.
\newblock


\bibitem[\protect\citeauthoryear{Baltag and Smets}{Baltag and Smets}{2011}]%
        {Baltag2011}
\bibfield{author}{\bibinfo{person}{Alexandru Baltag} {and}
  \bibinfo{person}{Sonja Smets}.} \bibinfo{year}{2011}\natexlab{}.
\newblock \showarticletitle{Quantum logic as a dynamic logic}.
\newblock \bibinfo{journal}{\emph{Synthese}} \bibinfo{volume}{179},
  \bibinfo{number}{2} (\bibinfo{year}{2011}), \bibinfo{pages}{285--306}.
\newblock


\bibitem[\protect\citeauthoryear{Barthe, Hsu, Ying, Yu, and Zhou}{Barthe
  et~al\mbox{.}}{2019}]%
        {popl20-relational}
\bibfield{author}{\bibinfo{person}{Gilles Barthe}, \bibinfo{person}{Justin
  Hsu}, \bibinfo{person}{Mingsheng Ying}, \bibinfo{person}{Nengkun Yu}, {and}
  \bibinfo{person}{Li Zhou}.} \bibinfo{year}{2019}\natexlab{}.
\newblock \showarticletitle{Relational Proofs for Quantum Programs}.
\newblock \bibinfo{journal}{\emph{Proc. ACM Program. Lang.}}
  \bibinfo{volume}{4}, \bibinfo{number}{POPL}, Article \bibinfo{articleno}{21}
  (\bibinfo{date}{Dec.} \bibinfo{year}{2019}), \bibinfo{numpages}{29}~pages.
\newblock
\urldef\tempurl%
\url{https://doi.org/10.1145/3371089}
\showDOI{\tempurl}


\bibitem[\protect\citeauthoryear{B{\'e}al, Lombardy, and Sakarovitch}{B{\'e}al
  et~al\mbox{.}}{2005}]%
        {beal2005equivalence}
\bibfield{author}{\bibinfo{person}{Marie-Pierre B{\'e}al},
  \bibinfo{person}{Sylvain Lombardy}, {and} \bibinfo{person}{Jacques
  Sakarovitch}.} \bibinfo{year}{2005}\natexlab{}.
\newblock \showarticletitle{On the Equivalence of $\mathbb{Z}$-Automata}. In
  \bibinfo{booktitle}{\emph{International Colloquium on Automata, Languages,
  and Programming}}. Springer, \bibinfo{pages}{397--409}.
\newblock


\bibitem[\protect\citeauthoryear{B{\'e}al, Lombardy, and Sakarovitch}{B{\'e}al
  et~al\mbox{.}}{2006}]%
        {beal2006conjugacy}
\bibfield{author}{\bibinfo{person}{Marie-Pierre B{\'e}al},
  \bibinfo{person}{Sylvain Lombardy}, {and} \bibinfo{person}{Jacques
  Sakarovitch}.} \bibinfo{year}{2006}\natexlab{}.
\newblock \showarticletitle{Conjugacy and equivalence of weighted automata and
  functional transducers}. In \bibinfo{booktitle}{\emph{International Computer
  Science Symposium in Russia}}. Springer, \bibinfo{pages}{58--69}.
\newblock


\bibitem[\protect\citeauthoryear{Berstel and Reutenauer}{Berstel and
  Reutenauer}{2011}]%
        {berstel2011noncommutative}
\bibfield{author}{\bibinfo{person}{Jean Berstel} {and}
  \bibinfo{person}{Christophe Reutenauer}.} \bibinfo{year}{2011}\natexlab{}.
\newblock \bibinfo{booktitle}{\emph{Noncommutative rational series with
  applications}}. Vol.~\bibinfo{volume}{137}.
\newblock \bibinfo{publisher}{Cambridge University Press}.
\newblock


\bibitem[\protect\citeauthoryear{Bloom and {\'E}sik}{Bloom and
  {\'E}sik}{2009}]%
        {bloom2009axiomatizing}
\bibfield{author}{\bibinfo{person}{Stephen~L Bloom} {and}
  \bibinfo{person}{Zolt{\'a}n {\'E}sik}.} \bibinfo{year}{2009}\natexlab{}.
\newblock \showarticletitle{Axiomatizing rational power series over natural
  numbers}.
\newblock \bibinfo{journal}{\emph{Information and Computation}}
  \bibinfo{volume}{207}, \bibinfo{number}{7} (\bibinfo{year}{2009}),
  \bibinfo{pages}{793--811}.
\newblock


\bibitem[\protect\citeauthoryear{B\"{o}hm and Jacopini}{B\"{o}hm and
  Jacopini}{1966}]%
        {BJ66}
\bibfield{author}{\bibinfo{person}{Corrado B\"{o}hm} {and}
  \bibinfo{person}{Giuseppe Jacopini}.} \bibinfo{year}{1966}\natexlab{}.
\newblock \showarticletitle{Flow Diagrams, Turing Machines and Languages with
  Only Two Formation Rules}.
\newblock \bibinfo{journal}{\emph{Commun. ACM}} \bibinfo{volume}{9},
  \bibinfo{number}{5} (\bibinfo{date}{May} \bibinfo{year}{1966}),
  \bibinfo{pages}{366–371}.
\newblock
\showISSN{0001-0782}
\urldef\tempurl%
\url{https://doi.org/10.1145/355592.365646}
\showDOI{\tempurl}


\bibitem[\protect\citeauthoryear{Bonchi, Bonsangue, Boreale, Rutten, and
  Silva}{Bonchi et~al\mbox{.}}{2012}]%
        {BONCHI201277}
\bibfield{author}{\bibinfo{person}{Filippo Bonchi}, \bibinfo{person}{Marcello
  Bonsangue}, \bibinfo{person}{Michele Boreale}, \bibinfo{person}{Jan Rutten},
  {and} \bibinfo{person}{Alexandra Silva}.} \bibinfo{year}{2012}\natexlab{}.
\newblock \showarticletitle{A coalgebraic perspective on linear weighted
  automata}.
\newblock \bibinfo{journal}{\emph{Information and Computation}}
  \bibinfo{volume}{211} (\bibinfo{year}{2012}), \bibinfo{pages}{77 -- 105}.
\newblock
\showISSN{0890-5401}
\urldef\tempurl%
\url{https://doi.org/10.1016/j.ic.2011.12.002}
\showDOI{\tempurl}


\bibitem[\protect\citeauthoryear{Brunet and Jorrand}{Brunet and
  Jorrand}{2004}]%
        {BJ04}
\bibfield{author}{\bibinfo{person}{Olivier Brunet} {and}
  \bibinfo{person}{Philippe Jorrand}.} \bibinfo{year}{2004}\natexlab{}.
\newblock \showarticletitle{Dynamic Quantum Logic for Quantum Programs}.
\newblock \bibinfo{journal}{\emph{International Journal of Quantum
  Information}} \bibinfo{volume}{2}, \bibinfo{number}{1}
  (\bibinfo{year}{2004}).
\newblock


\bibitem[\protect\citeauthoryear{Cacciapuoti, Caleffi, Tafuri, Cataliotti,
  Gherardini, and Bianchi}{Cacciapuoti et~al\mbox{.}}{2020}]%
        {qnetwork}
\bibfield{author}{\bibinfo{person}{Angela~Sara Cacciapuoti},
  \bibinfo{person}{Marcello Caleffi}, \bibinfo{person}{Francesco Tafuri},
  \bibinfo{person}{Francesco~Saverio Cataliotti}, \bibinfo{person}{Stefano
  Gherardini}, {and} \bibinfo{person}{Giuseppe Bianchi}.}
  \bibinfo{year}{2020}\natexlab{}.
\newblock \showarticletitle{Quantum Internet: Networking Challenges in
  Distributed Quantum Computing}.
\newblock \bibinfo{journal}{\emph{IEEE Network}} \bibinfo{volume}{34},
  \bibinfo{number}{1} (\bibinfo{year}{2020}), \bibinfo{pages}{137--143}.
\newblock
\urldef\tempurl%
\url{https://doi.org/10.1109/MNET.001.1900092}
\showDOI{\tempurl}


\bibitem[\protect\citeauthoryear{Chadha, Mateus, and Sernadas}{Chadha
  et~al\mbox{.}}{2006}]%
        {CHADHA200619}
\bibfield{author}{\bibinfo{person}{Rohit Chadha}, \bibinfo{person}{Paulo
  Mateus}, {and} \bibinfo{person}{Am{\'i}lcar Sernadas}.}
  \bibinfo{year}{2006}\natexlab{}.
\newblock \showarticletitle{Reasoning About Imperative Quantum Programs}.
\newblock \bibinfo{journal}{\emph{Electronic Notes in Theoretical Computer
  Science}}  \bibinfo{volume}{158} (\bibinfo{year}{2006}).
\newblock


\bibitem[\protect\citeauthoryear{Childs, Maslov, Nam, Ross, and Su}{Childs
  et~al\mbox{.}}{2018}]%
        {Childs9456}
\bibfield{author}{\bibinfo{person}{Andrew~M. Childs}, \bibinfo{person}{Dmitri
  Maslov}, \bibinfo{person}{Yunseong Nam}, \bibinfo{person}{Neil~J. Ross},
  {and} \bibinfo{person}{Yuan Su}.} \bibinfo{year}{2018}\natexlab{}.
\newblock \showarticletitle{Toward the first quantum simulation with quantum
  speedup}.
\newblock \bibinfo{journal}{\emph{Proceedings of the National Academy of
  Sciences}} \bibinfo{volume}{115}, \bibinfo{number}{38}
  (\bibinfo{year}{2018}), \bibinfo{pages}{9456--9461}.
\newblock
\showISSN{0027-8424}


\bibitem[\protect\citeauthoryear{Cohen, Kozen, and Smith}{Cohen
  et~al\mbox{.}}{1996}]%
        {CKS96a}
\bibfield{author}{\bibinfo{person}{Ernie Cohen}, \bibinfo{person}{Dexter
  Kozen}, {and} \bibinfo{person}{Frederick Smith}.}
  \bibinfo{year}{1996}\natexlab{}.
\newblock \bibinfo{booktitle}{\emph{The complexity of {K}leene algebra with
  tests}}.
\newblock \bibinfo{type}{{T}echnical {R}eport} TR96-1598.
  \bibinfo{institution}{Computer Science Department, Cornell University}.
\newblock


\bibitem[\protect\citeauthoryear{D'Hondt and Panangaden}{D'Hondt and
  Panangaden}{2006}]%
        {DP2006}
\bibfield{author}{\bibinfo{person}{Ellie D'Hondt} {and}
  \bibinfo{person}{Prakash Panangaden}.} \bibinfo{year}{2006}\natexlab{}.
\newblock \showarticletitle{Quantum Weakest Preconditions}.
\newblock \bibinfo{journal}{\emph{Mathematical Structures in Computer Science}}
  \bibinfo{volume}{16}, \bibinfo{number}{3} (\bibinfo{year}{2006}).
\newblock


\bibitem[\protect\citeauthoryear{Droste, Kuich, and Vogler}{Droste
  et~al\mbox{.}}{2009}]%
        {droste2009handbook}
\bibfield{author}{\bibinfo{person}{Manfred Droste}, \bibinfo{person}{Werner
  Kuich}, {and} \bibinfo{person}{Heiko Vogler}.}
  \bibinfo{year}{2009}\natexlab{}.
\newblock \bibinfo{booktitle}{\emph{Handbook of weighted automata}}.
\newblock \bibinfo{publisher}{Springer Science \& Business Media}.
\newblock


\bibitem[\protect\citeauthoryear{Eilenberg}{Eilenberg}{1974}]%
        {eilenberg1974automata}
\bibfield{author}{\bibinfo{person}{Samuel Eilenberg}.}
  \bibinfo{year}{1974}\natexlab{}.
\newblock \bibinfo{booktitle}{\emph{Automata, languages, and machines}}.
\newblock \bibinfo{publisher}{Academic press}.
\newblock


\bibitem[\protect\citeauthoryear{{\'E}sik and Kuich}{{\'E}sik and
  Kuich}{2004}]%
        {esik2004inductive}
\bibfield{author}{\bibinfo{person}{Zolt{\'a}n {\'E}sik} {and}
  \bibinfo{person}{Werner Kuich}.} \bibinfo{year}{2004}\natexlab{}.
\newblock \showarticletitle{Inductive *-semirings}.
\newblock \bibinfo{journal}{\emph{Theoretical Computer Science}}
  \bibinfo{volume}{324}, \bibinfo{number}{1} (\bibinfo{year}{2004}),
  \bibinfo{pages}{3--33}.
\newblock


\bibitem[\protect\citeauthoryear{Feng, Duan, Ji, and Ying}{Feng
  et~al\mbox{.}}{2007}]%
        {Feng:2007}
\bibfield{author}{\bibinfo{person}{Yuan Feng}, \bibinfo{person}{Runyao Duan},
  \bibinfo{person}{Zhengfeng Ji}, {and} \bibinfo{person}{Mingsheng Ying}.}
  \bibinfo{year}{2007}\natexlab{}.
\newblock \showarticletitle{Proof Rules for the Correctness of Quantum
  Programs}.
\newblock \bibinfo{journal}{\emph{Theoretical Computer Science}}
  \bibinfo{volume}{386}, \bibinfo{number}{1-2} (\bibinfo{year}{2007}).
\newblock


\bibitem[\protect\citeauthoryear{Fischer and Ladner}{Fischer and
  Ladner}{1979}]%
        {FISCHER1979194}
\bibfield{author}{\bibinfo{person}{Michael~J. Fischer} {and}
  \bibinfo{person}{Richard~E. Ladner}.} \bibinfo{year}{1979}\natexlab{}.
\newblock \showarticletitle{Propositional dynamic logic of regular programs}.
\newblock \bibinfo{journal}{\emph{J. Comput. System Sci.}}
  \bibinfo{volume}{18}, \bibinfo{number}{2} (\bibinfo{year}{1979}),
  \bibinfo{pages}{194 -- 211}.
\newblock
\showISSN{0022-0000}
\urldef\tempurl%
\url{https://doi.org/10.1016/0022-0000(79)90046-1}
\showDOI{\tempurl}


\bibitem[\protect\citeauthoryear{Foster, Kozen, Mamouras, Reitblatt, and
  Silva}{Foster et~al\mbox{.}}{2016}]%
        {FKMRS15a}
\bibfield{author}{\bibinfo{person}{Nate Foster}, \bibinfo{person}{Dexter
  Kozen}, \bibinfo{person}{Konstantinos Mamouras}, \bibinfo{person}{Mark
  Reitblatt}, {and} \bibinfo{person}{Alexandra Silva}.}
  \bibinfo{year}{2016}\natexlab{}.
\newblock \showarticletitle{Probabilistic {NetKAT}}. In
  \bibinfo{booktitle}{\emph{25th European Symposium on Programming (ESOP
  2016)}} \emph{(\bibinfo{series}{Lecture Notes in Computer Science},
  Vol.~\bibinfo{volume}{9632})}, \bibfield{editor}{\bibinfo{person}{Peter
  Thiemann}} (Ed.). \bibinfo{publisher}{Springer}, \bibinfo{address}{Eindhoven,
  The Netherlands}, \bibinfo{pages}{282--309}.
\newblock
\urldef\tempurl%
\url{https://doi.org/10.1007/978-3-662-49498-1_12}
\showDOI{\tempurl}


\bibitem[\protect\citeauthoryear{Foster, Kozen, Milano, Silva, and
  Thompson}{Foster et~al\mbox{.}}{2015}]%
        {FKMST15a}
\bibfield{author}{\bibinfo{person}{Nate Foster}, \bibinfo{person}{Dexter
  Kozen}, \bibinfo{person}{Matthew Milano}, \bibinfo{person}{Alexandra Silva},
  {and} \bibinfo{person}{Laure Thompson}.} \bibinfo{year}{2015}\natexlab{}.
\newblock \showarticletitle{A Coalgebraic Decision Procedure for {NetKAT}}. In
  \bibinfo{booktitle}{\emph{Proc. 42nd ACM SIGPLAN-SIGACT Symp. Principles of
  Programming Languages (POPL'15)}}. ACM, \bibinfo{address}{Mumbai, India},
  \bibinfo{pages}{343--355}.
\newblock


\bibitem[\protect\citeauthoryear{Foulis and Bennett}{Foulis and
  Bennett}{1994}]%
        {foulis1994effect}
\bibfield{author}{\bibinfo{person}{David~J Foulis} {and}
  \bibinfo{person}{Mary~K Bennett}.} \bibinfo{year}{1994}\natexlab{}.
\newblock \showarticletitle{Effect algebras and unsharp quantum logics}.
\newblock \bibinfo{journal}{\emph{Foundations of physics}}
  \bibinfo{volume}{24}, \bibinfo{number}{10} (\bibinfo{year}{1994}),
  \bibinfo{pages}{1331--1352}.
\newblock


\bibitem[\protect\citeauthoryear{Gay}{Gay}{2006}]%
        {Gay:2006}
\bibfield{author}{\bibinfo{person}{Simon~J. Gay}.}
  \bibinfo{year}{2006}\natexlab{}.
\newblock \showarticletitle{Quantum Programming Languages: Survey and
  Bibliography}.
\newblock \bibinfo{journal}{\emph{Mathematical Structures in Computer Science}}
  \bibinfo{volume}{16}, \bibinfo{number}{4} (\bibinfo{year}{2006}).
\newblock


\bibitem[\protect\citeauthoryear{{Google}}{{Google}}{2018}]%
        {Cirq18}
\bibfield{author}{\bibinfo{person}{{Google}}.} \bibinfo{year}{2018}\natexlab{}.
\newblock \bibinfo{howpublished}{\url{https://github.com/quantumlib/Cirq}}.
\newblock


\bibitem[\protect\citeauthoryear{Grattage}{Grattage}{2005}]%
        {AG05}
\bibfield{author}{\bibinfo{person}{Jonathan Grattage}.}
  \bibinfo{year}{2005}\natexlab{}.
\newblock \showarticletitle{A Functional Quantum Programming Language}. In
  \bibinfo{booktitle}{\emph{LICS}}.
\newblock


\bibitem[\protect\citeauthoryear{Green, Lumsdaine, Ross, Selinger, and
  Valiron}{Green et~al\mbox{.}}{2013}]%
        {Green2013}
\bibfield{author}{\bibinfo{person}{Alexander~S Green},
  \bibinfo{person}{Peter~LeFanu Lumsdaine}, \bibinfo{person}{Neil~J Ross},
  \bibinfo{person}{Peter Selinger}, {and} \bibinfo{person}{Beno{\^\i}t
  Valiron}.} \bibinfo{year}{2013}\natexlab{}.
\newblock \showarticletitle{Quipper: a scalable quantum programming language}.
  In \bibinfo{booktitle}{\emph{PLDI}}. \bibinfo{pages}{333--342}.
\newblock


\bibitem[\protect\citeauthoryear{Hietala, Rand, Hung, Wu, and Hicks}{Hietala
  et~al\mbox{.}}{2019}]%
        {hietala2019verified}
\bibfield{author}{\bibinfo{person}{Kesha Hietala}, \bibinfo{person}{Robert
  Rand}, \bibinfo{person}{Shih-Han Hung}, \bibinfo{person}{Xiaodi Wu}, {and}
  \bibinfo{person}{Michael Hicks}.} \bibinfo{year}{2019}\natexlab{}.
\newblock \showarticletitle{A verified optimizer for quantum circuits}.
\newblock \bibinfo{journal}{\emph{arXiv preprint arXiv:1912.02250}}
  (\bibinfo{year}{2019}).
\newblock


\bibitem[\protect\citeauthoryear{Hoare}{Hoare}{1969}]%
        {Hoare:1969}
\bibfield{author}{\bibinfo{person}{C.~A.~R. Hoare}.}
  \bibinfo{year}{1969}\natexlab{}.
\newblock \showarticletitle{An Axiomatic Basis for Computer Programming}.
\newblock \bibinfo{journal}{\emph{Commun. ACM}} \bibinfo{volume}{12},
  \bibinfo{number}{10} (\bibinfo{date}{Oct.} \bibinfo{year}{1969}),
  \bibinfo{pages}{576--580}.
\newblock
\showISSN{0001-0782}


\bibitem[\protect\citeauthoryear{Kakutani}{Kakutani}{2009}]%
        {Kaku09}
\bibfield{author}{\bibinfo{person}{Yoshihiko Kakutani}.}
  \bibinfo{year}{2009}\natexlab{}.
\newblock \showarticletitle{A Logic for Formal Verification of Quantum
  Programs}. In \bibinfo{booktitle}{\emph{ASIAN 2009}}.
  \bibinfo{pages}{79--93}.
\newblock


\bibitem[\protect\citeauthoryear{Kiefer, Murawski, Ouaknine, Wachter, and
  Worrell}{Kiefer et~al\mbox{.}}{2013}]%
        {kiefer2013complexity}
\bibfield{author}{\bibinfo{person}{Stefan Kiefer}, \bibinfo{person}{Andrzej
  Murawski}, \bibinfo{person}{Jo{\"e}l Ouaknine}, \bibinfo{person}{Bj{\"o}rn
  Wachter}, {and} \bibinfo{person}{James Worrell}.}
  \bibinfo{year}{2013}\natexlab{}.
\newblock \showarticletitle{On the complexity of equivalence and minimisation
  for Q-weighted automata}.
\newblock \bibinfo{journal}{\emph{arXiv preprint arXiv:1302.2818}}
  (\bibinfo{year}{2013}).
\newblock


\bibitem[\protect\citeauthoryear{Kleene}{Kleene}{1956}]%
        {Kleene56}
\bibfield{author}{\bibinfo{person}{S.~C. Kleene}.}
  \bibinfo{year}{1956}\natexlab{}.
\newblock \bibinfo{booktitle}{\emph{Representation of Events in Nerve Nets and
  Finite Automata}}.
\newblock \bibinfo{publisher}{Princeton University Press},
  \bibinfo{address}{Princeton}, \bibinfo{pages}{3 -- 42}.
\newblock
\urldef\tempurl%
\url{https://doi.org/10.1515/9781400882618-002}
\showDOI{\tempurl}


\bibitem[\protect\citeauthoryear{Kozen}{Kozen}{1990}]%
        {kozen1990completeness}
\bibfield{author}{\bibinfo{person}{Dexter Kozen}.}
  \bibinfo{year}{1990}\natexlab{}.
\newblock \bibinfo{booktitle}{\emph{A completeness theorem for Kleene algebras
  and the algebra of regular events}}.
\newblock \bibinfo{type}{{T}echnical {R}eport}. \bibinfo{institution}{Cornell
  University}.
\newblock


\bibitem[\protect\citeauthoryear{Kozen}{Kozen}{1997}]%
        {K97c}
\bibfield{author}{\bibinfo{person}{Dexter Kozen}.}
  \bibinfo{year}{1997}\natexlab{}.
\newblock \showarticletitle{Kleene algebra with tests}.
\newblock \bibinfo{journal}{\emph{ACM Trans. Programming Languages and Systems
  (TOPLAS)}} \bibinfo{volume}{19}, \bibinfo{number}{3} (\bibinfo{date}{May}
  \bibinfo{year}{1997}), \bibinfo{pages}{427--443}.
\newblock
\urldef\tempurl%
\url{https://doi.org/10.1145/256167.256195}
\showDOI{\tempurl}


\bibitem[\protect\citeauthoryear{Kozen}{Kozen}{2000}]%
        {K00a}
\bibfield{author}{\bibinfo{person}{Dexter Kozen}.}
  \bibinfo{year}{2000}\natexlab{}.
\newblock \showarticletitle{On {H}oare logic and {K}leene algebra with tests}.
\newblock \bibinfo{journal}{\emph{Trans. Computational Logic}}
  \bibinfo{volume}{1}, \bibinfo{number}{1} (\bibinfo{date}{July}
  \bibinfo{year}{2000}), \bibinfo{pages}{60--76}.
\newblock


\bibitem[\protect\citeauthoryear{Kozen}{Kozen}{2017}]%
        {K17a}
\bibfield{author}{\bibinfo{person}{Dexter Kozen}.}
  \bibinfo{year}{2017}\natexlab{}.
\newblock \showarticletitle{On the Coalgebraic Theory of {K}leene Algebra with
  Tests}.
\newblock In \bibinfo{booktitle}{\emph{Rohit Parikh on Logic, Language and
  Society}}, \bibfield{editor}{\bibinfo{person}{Can Ba{\c{s}}kent},
  \bibinfo{person}{Lawrence~S. Moss}, {and} \bibinfo{person}{Ramaswamy
  Ramanujam}} (Eds.). \bibinfo{series}{Outstanding Contributions to Logic},
  Vol.~\bibinfo{volume}{11}. \bibinfo{publisher}{Springer},
  \bibinfo{pages}{279--298}.
\newblock


\bibitem[\protect\citeauthoryear{Kozen and Patron}{Kozen and Patron}{2000}]%
        {KP00}
\bibfield{author}{\bibinfo{person}{Dexter Kozen} {and}
  \bibinfo{person}{Maria-Cristina Patron}.} \bibinfo{year}{2000}\natexlab{}.
\newblock \showarticletitle{Certification of compiler optimizations using
  {K}leene algebra with tests}. In \bibinfo{booktitle}{\emph{Proc. 1st Int.
  Conf. Computational Logic (CL2000)}} (London) \emph{(\bibinfo{series}{Lecture
  Notes in Artificial Intelligence}, Vol.~\bibinfo{volume}{1861})},
  \bibfield{editor}{\bibinfo{person}{John Lloyd}, \bibinfo{person}{Veronica
  Dahl}, \bibinfo{person}{Ulrich Furbach}, \bibinfo{person}{Manfred Kerber},
  \bibinfo{person}{Kung-Kiu Lau}, \bibinfo{person}{Catuscia Palamidessi},
  \bibinfo{person}{Luis~Moniz Pereira}, \bibinfo{person}{Yehoshua Sagiv}, {and}
  \bibinfo{person}{Peter~J. Stuckey}} (Eds.).
  \bibinfo{publisher}{Springer-Verlag}, \bibinfo{address}{London},
  \bibinfo{pages}{568--582}.
\newblock


\bibitem[\protect\citeauthoryear{Kozen and Smith}{Kozen and Smith}{1996}]%
        {KS96a}
\bibfield{author}{\bibinfo{person}{Dexter Kozen} {and}
  \bibinfo{person}{Frederick Smith}.} \bibinfo{year}{1996}\natexlab{}.
\newblock \showarticletitle{Kleene algebra with tests: Completeness and
  decidability}. In \bibinfo{booktitle}{\emph{Proc. 10th Int. Workshop Computer
  Science Logic (CSL'96)}} \emph{(\bibinfo{series}{Lecture Notes in Computer
  Science}, Vol.~\bibinfo{volume}{1258})},
  \bibfield{editor}{\bibinfo{person}{D.~van Dalen} {and}
  \bibinfo{person}{M.~Bezem}} (Eds.). \bibinfo{publisher}{Springer-Verlag},
  \bibinfo{address}{Utrecht, The Netherlands}, \bibinfo{pages}{244--259}.
\newblock


\bibitem[\protect\citeauthoryear{Kozlowski and Wehner}{Kozlowski and
  Wehner}{2019}]%
        {q-internet}
\bibfield{author}{\bibinfo{person}{Wojciech Kozlowski} {and}
  \bibinfo{person}{Stephanie Wehner}.} \bibinfo{year}{2019}\natexlab{}.
\newblock \showarticletitle{Towards Large-Scale Quantum Networks}. In
  \bibinfo{booktitle}{\emph{Proceedings of the Sixth Annual ACM International
  Conference on Nanoscale Computing and Communication}} (Dublin, Ireland)
  \emph{(\bibinfo{series}{NANOCOM '19})}. \bibinfo{publisher}{Association for
  Computing Machinery}, \bibinfo{address}{New York, NY, USA}, Article
  \bibinfo{articleno}{3}, \bibinfo{numpages}{7}~pages.
\newblock
\showISBNx{9781450368971}
\urldef\tempurl%
\url{https://doi.org/10.1145/3345312.3345497}
\showDOI{\tempurl}


\bibitem[\protect\citeauthoryear{Kraus, B{\"o}hm, Dollard, and Wootters}{Kraus
  et~al\mbox{.}}{1983}]%
        {kraus1983states}
\bibfield{author}{\bibinfo{person}{Karl Kraus}, \bibinfo{person}{Arno
  B{\"o}hm}, \bibinfo{person}{John~D Dollard}, {and} \bibinfo{person}{WH
  Wootters}.} \bibinfo{year}{1983}\natexlab{}.
\newblock \showarticletitle{States, effects, and operations: fundamental
  notions of quantum theory. Lectures in mathematical physics at the University
  of Texas at Austin}.
\newblock \bibinfo{journal}{\emph{Lecture notes in physics}}
  \bibinfo{volume}{190} (\bibinfo{year}{1983}).
\newblock


\bibitem[\protect\citeauthoryear{Kuich and Salomaa}{Kuich and Salomaa}{1985}]%
        {kuich2012semirings}
\bibfield{author}{\bibinfo{person}{Werner Kuich} {and} \bibinfo{person}{Arto
  Salomaa}.} \bibinfo{year}{1985}\natexlab{}.
\newblock \bibinfo{booktitle}{\emph{Semirings, Automata, Languages}}.
\newblock \bibinfo{publisher}{Springer-Verlag}, \bibinfo{address}{Berlin,
  Heidelberg}.
\newblock
\showISBNx{3540137165}


\bibitem[\protect\citeauthoryear{Li and Ying}{Li and Ying}{2017}]%
        {Li:2017}
\bibfield{author}{\bibinfo{person}{Yangjia Li} {and} \bibinfo{person}{Mingsheng
  Ying}.} \bibinfo{year}{2017}\natexlab{}.
\newblock \showarticletitle{Algorithmic Analysis of Termination Problems for
  Quantum Programs}.
\newblock  \bibinfo{volume}{2}, \bibinfo{number}{POPL}, Article
  \bibinfo{articleno}{35} (\bibinfo{date}{Dec.} \bibinfo{year}{2017}),
  \bibinfo{numpages}{29}~pages.
\newblock


\bibitem[\protect\citeauthoryear{Low and Chuang}{Low and Chuang}{2017}]%
        {low2017optimal}
\bibfield{author}{\bibinfo{person}{Guang~Hao Low} {and}
  \bibinfo{person}{Isaac~L Chuang}.} \bibinfo{year}{2017}\natexlab{}.
\newblock \showarticletitle{Optimal Hamiltonian simulation by quantum signal
  processing}.
\newblock \bibinfo{journal}{\emph{Physical review letters}}
  \bibinfo{volume}{118}, \bibinfo{number}{1} (\bibinfo{year}{2017}),
  \bibinfo{pages}{010501}.
\newblock


\bibitem[\protect\citeauthoryear{Mislove}{Mislove}{2006}]%
        {MISLOVE2006}
\bibfield{author}{\bibinfo{person}{Mike Mislove}.}
  \bibinfo{year}{2006}\natexlab{}.
\newblock \showarticletitle{On Combining Probability and Nondeterminism}.
\newblock \bibinfo{journal}{\emph{Electronic Notes in Theoretical Computer
  Science}}  \bibinfo{volume}{162} (\bibinfo{year}{2006}), \bibinfo{pages}{261
  -- 265}.
\newblock
\newblock
\shownote{Proceedings of the Workshop Essays on Algebraic Process Calculi (APC
  25)}.


\bibitem[\protect\citeauthoryear{Nielsen and Chuang}{Nielsen and
  Chuang}{2010}]%
        {nielsen_chuang_2010}
\bibfield{author}{\bibinfo{person}{Michael~A. Nielsen} {and}
  \bibinfo{person}{Isaac~L. Chuang}.} \bibinfo{year}{2010}\natexlab{}.
\newblock \bibinfo{booktitle}{\emph{Quantum Computation and Quantum
  Information: 10th Anniversary Edition}}.
\newblock \bibinfo{publisher}{Cambridge University Press}.
\newblock
\urldef\tempurl%
\url{https://doi.org/10.1017/CBO9780511976667}
\showDOI{\tempurl}


\bibitem[\protect\citeauthoryear{\"{O}mer}{\"{O}mer}{2003}]%
        {Om03}
\bibfield{author}{\bibinfo{person}{Bernhard \"{O}mer}.}
  \bibinfo{year}{2003}\natexlab{}.
\newblock \emph{\bibinfo{title}{Structured Quantum Programming}}.
\newblock \bibinfo{thesistype}{Ph.\,D. Dissertation}. \bibinfo{school}{Vienna
  University of Technology}.
\newblock


\bibitem[\protect\citeauthoryear{Paykin, Rand, and Zdancewic}{Paykin
  et~al\mbox{.}}{2017}]%
        {PRZ2017}
\bibfield{author}{\bibinfo{person}{Jennifer Paykin}, \bibinfo{person}{Robert
  Rand}, {and} \bibinfo{person}{Steve Zdancewic}.}
  \bibinfo{year}{2017}\natexlab{}.
\newblock \showarticletitle{{QWIRE}: A Core Language for Quantum Circuits}
  \emph{(\bibinfo{series}{POPL 2017})}. \bibinfo{pages}{846–858}.
\newblock
\showISBNx{9781450346603}


\bibitem[\protect\citeauthoryear{Pous}{Pous}{2015}]%
        {Pous15}
\bibfield{author}{\bibinfo{person}{Damien Pous}.}
  \bibinfo{year}{2015}\natexlab{}.
\newblock \showarticletitle{Symbolic Algorithms for Language Equivalence and
  Kleene Algebra with Tests}. In \bibinfo{booktitle}{\emph{Proceedings of the
  42nd Annual ACM SIGPLAN-SIGACT Symposium on Principles of Programming
  Languages}} (Mumbai, India) \emph{(\bibinfo{series}{POPL ’15})}.
  \bibinfo{publisher}{Association for Computing Machinery},
  \bibinfo{address}{New York, NY, USA}, \bibinfo{pages}{357–368}.
\newblock
\showISBNx{9781450333009}
\urldef\tempurl%
\url{https://doi.org/10.1145/2676726.2677007}
\showDOI{\tempurl}


\bibitem[\protect\citeauthoryear{{Rigetti}}{{Rigetti}}{2018}]%
        {ForestRig}
\bibfield{author}{\bibinfo{person}{{Rigetti}}.}
  \bibinfo{year}{2018}\natexlab{}.
\newblock \bibinfo{howpublished}{\url{https://www.rigetti.com/forest}}.
\newblock


\bibitem[\protect\citeauthoryear{Sabry}{Sabry}{2003}]%
        {Sabry-Haskel}
\bibfield{author}{\bibinfo{person}{Amr Sabry}.}
  \bibinfo{year}{2003}\natexlab{}.
\newblock \showarticletitle{Modeling Quantum Computing in Haskell}. In
  \bibinfo{booktitle}{\emph{The Haskell Workshop}}.
\newblock


\bibitem[\protect\citeauthoryear{Sanders and Zuliani}{Sanders and
  Zuliani}{2000}]%
        {SZ00}
\bibfield{author}{\bibinfo{person}{Jeff~W. Sanders} {and}
  \bibinfo{person}{Paolo Zuliani}.} \bibinfo{year}{2000}\natexlab{}.
\newblock \showarticletitle{Quantum Programming}. In
  \bibinfo{booktitle}{\emph{MPC}}.
\newblock


\bibitem[\protect\citeauthoryear{Schützenberger}{Schützenberger}{1961}]%
        {SCHUTZENBERGER1961245}
\bibfield{author}{\bibinfo{person}{M.P. Schützenberger}.}
  \bibinfo{year}{1961}\natexlab{}.
\newblock \showarticletitle{On the definition of a family of automata}.
\newblock \bibinfo{journal}{\emph{Information and Control}}
  \bibinfo{volume}{4}, \bibinfo{number}{2} (\bibinfo{year}{1961}),
  \bibinfo{pages}{245 -- 270}.
\newblock
\showISSN{0019-9958}
\urldef\tempurl%
\url{https://doi.org/10.1016/S0019-9958(61)80020-X}
\showDOI{\tempurl}


\bibitem[\protect\citeauthoryear{Selinger}{Selinger}{2004a}]%
        {Selinger04}
\bibfield{author}{\bibinfo{person}{Peter Selinger}.}
  \bibinfo{year}{2004}\natexlab{a}.
\newblock \showarticletitle{A Brief Survey of Quantum Programming Languages}.
  In \bibinfo{booktitle}{\emph{FLOPS}}.
\newblock


\bibitem[\protect\citeauthoryear{Selinger}{Selinger}{2004b}]%
        {Se04}
\bibfield{author}{\bibinfo{person}{Peter Selinger}.}
  \bibinfo{year}{2004}\natexlab{b}.
\newblock \showarticletitle{Towards a Quantum Programming Language}.
\newblock \bibinfo{journal}{\emph{Mathematical Structures in Computer Science}}
  \bibinfo{volume}{14}, \bibinfo{number}{4} (\bibinfo{year}{2004}).
\newblock


\bibitem[\protect\citeauthoryear{Silva}{Silva}{2010}]%
        {silva2010kleene}
\bibfield{author}{\bibinfo{person}{Alexandra Silva}.}
  \bibinfo{year}{2010}\natexlab{}.
\newblock \emph{\bibinfo{title}{Kleene coalgebra}}.
\newblock \bibinfo{thesistype}{Ph.\,D. Dissertation}. \bibinfo{school}{Radboud
  University Nijmegen}.
\newblock


\bibitem[\protect\citeauthoryear{Smolka, Foster, Hsu, Kapp\'{e}, Kozen, and
  Silva}{Smolka et~al\mbox{.}}{2019}]%
        {GKAT}
\bibfield{author}{\bibinfo{person}{Steffen Smolka}, \bibinfo{person}{Nate
  Foster}, \bibinfo{person}{Justin Hsu}, \bibinfo{person}{Tobias Kapp\'{e}},
  \bibinfo{person}{Dexter Kozen}, {and} \bibinfo{person}{Alexandra Silva}.}
  \bibinfo{year}{2019}\natexlab{}.
\newblock \showarticletitle{Guarded Kleene Algebra with Tests: Verification of
  Uninterpreted Programs in Nearly Linear Time}.
\newblock \bibinfo{journal}{\emph{Proc. ACM Program. Lang.}}
  \bibinfo{volume}{4}, \bibinfo{number}{POPL}, Article \bibinfo{articleno}{61}
  (\bibinfo{date}{Dec.} \bibinfo{year}{2019}), \bibinfo{numpages}{28}~pages.
\newblock
\urldef\tempurl%
\url{https://doi.org/10.1145/3371129}
\showDOI{\tempurl}


\bibitem[\protect\citeauthoryear{Staton}{Staton}{2015}]%
        {staton2015algebraic}
\bibfield{author}{\bibinfo{person}{Sam Staton}.}
  \bibinfo{year}{2015}\natexlab{}.
\newblock \showarticletitle{Algebraic Effects, Linearity, and Quantum
  Programming Languages}. In \bibinfo{booktitle}{\emph{Proceedings of the 42nd
  Annual ACM SIGPLAN-SIGACT Symposium on Principles of Programming Languages}}
  (Mumbai, India) \emph{(\bibinfo{series}{POPL '15})}.
  \bibinfo{publisher}{Association for Computing Machinery},
  \bibinfo{address}{New York, NY, USA}, \bibinfo{pages}{395–406}.
\newblock
\showISBNx{9781450333009}
\urldef\tempurl%
\url{https://doi.org/10.1145/2676726.2676999}
\showDOI{\tempurl}


\bibitem[\protect\citeauthoryear{Stockmeyer and Meyer}{Stockmeyer and
  Meyer}{1973}]%
        {stockmeyer1973word}
\bibfield{author}{\bibinfo{person}{Larry~J Stockmeyer} {and}
  \bibinfo{person}{Albert~R Meyer}.} \bibinfo{year}{1973}\natexlab{}.
\newblock \showarticletitle{Word problems requiring exponential time
  (preliminary report)}. In \bibinfo{booktitle}{\emph{Proceedings of the fifth
  annual ACM symposium on Theory of computing}}. \bibinfo{pages}{1--9}.
\newblock


\bibitem[\protect\citeauthoryear{Svore, Geller, Troyer, Azariah, Granade, Heim,
  Kliuchnikov, Mykhailova, Paz, and Roetteler}{Svore et~al\mbox{.}}{2018}]%
        {Svore:2018}
\bibfield{author}{\bibinfo{person}{Krysta Svore}, \bibinfo{person}{Alan
  Geller}, \bibinfo{person}{Matthias Troyer}, \bibinfo{person}{John Azariah},
  \bibinfo{person}{Christopher Granade}, \bibinfo{person}{Bettina Heim},
  \bibinfo{person}{Vadym Kliuchnikov}, \bibinfo{person}{Mariia Mykhailova},
  \bibinfo{person}{Andres Paz}, {and} \bibinfo{person}{Martin Roetteler}.}
  \bibinfo{year}{2018}\natexlab{}.
\newblock \showarticletitle{Q\#: Enabling Scalable Quantum Computing and
  Development with a High-level {DSL}}. In \bibinfo{booktitle}{\emph{RWDSL}}.
\newblock


\bibitem[\protect\citeauthoryear{Unruh}{Unruh}{2019}]%
        {Unruh-POPL-19}
\bibfield{author}{\bibinfo{person}{Dominique Unruh}.}
  \bibinfo{year}{2019}\natexlab{}.
\newblock \showarticletitle{Quantum Relational {Hoare} Logic}.
\newblock  \bibinfo{volume}{3}, \bibinfo{number}{POPL}, Article
  \bibinfo{articleno}{33} (\bibinfo{date}{Jan.} \bibinfo{year}{2019}),
  \bibinfo{numpages}{31}~pages.
\newblock


\bibitem[\protect\citeauthoryear{Varacca and Winskel}{Varacca and
  Winskel}{2006}]%
        {varacca_winskel_2006}
\bibfield{author}{\bibinfo{person}{Daniele Varacca} {and}
  \bibinfo{person}{Glynn Winskel}.} \bibinfo{year}{2006}\natexlab{}.
\newblock \showarticletitle{Distributing probability over non-determinism}.
\newblock \bibinfo{journal}{\emph{Mathematical Structures in Computer Science}}
  \bibinfo{volume}{16}, \bibinfo{number}{1} (\bibinfo{year}{2006}),
  \bibinfo{pages}{87–113}.
\newblock
\urldef\tempurl%
\url{https://doi.org/10.1017/S0960129505005074}
\showDOI{\tempurl}


\bibitem[\protect\citeauthoryear{Watrous}{Watrous}{2018}]%
        {watrous_2018}
\bibfield{author}{\bibinfo{person}{John Watrous}.}
  \bibinfo{year}{2018}\natexlab{}.
\newblock \bibinfo{booktitle}{\emph{The Theory of Quantum Information}}.
\newblock \bibinfo{publisher}{Cambridge University Press}.
\newblock
\urldef\tempurl%
\url{https://doi.org/10.1017/9781316848142}
\showDOI{\tempurl}


\bibitem[\protect\citeauthoryear{Wootters and Zurek}{Wootters and
  Zurek}{1982}]%
        {no-cloning}
\bibfield{author}{\bibinfo{person}{W.~K. Wootters} {and} \bibinfo{person}{W.~H.
  Zurek}.} \bibinfo{year}{1982}\natexlab{}.
\newblock \showarticletitle{A single quantum cannot be cloned}.
\newblock \bibinfo{journal}{\emph{Nature}} \bibinfo{volume}{299},
  \bibinfo{number}{5886} (\bibinfo{year}{1982}), \bibinfo{pages}{802--803}.
\newblock


\bibitem[\protect\citeauthoryear{Ying}{Ying}{2011}]%
        {Yin11}
\bibfield{author}{\bibinfo{person}{Mingsheng Ying}.}
  \bibinfo{year}{2011}\natexlab{}.
\newblock \showarticletitle{Floyd--{Hoare} Logic for Quantum Programs}.
\newblock \bibinfo{journal}{\emph{ACM Transactions on Programming Languages and
  Systems}} \bibinfo{volume}{33}, \bibinfo{number}{6} (\bibinfo{year}{2011}).
\newblock


\bibitem[\protect\citeauthoryear{Ying}{Ying}{2016}]%
        {Ying16}
\bibfield{author}{\bibinfo{person}{Mingsheng Ying}.}
  \bibinfo{year}{2016}\natexlab{}.
\newblock \bibinfo{booktitle}{\emph{Foundations of Quantum Programming}}.
\newblock \bibinfo{publisher}{Morgan Kaufmann}.
\newblock


\bibitem[\protect\citeauthoryear{Ying}{Ying}{2019}]%
        {Ying2019}
\bibfield{author}{\bibinfo{person}{Mingsheng Ying}.}
  \bibinfo{year}{2019}\natexlab{}.
\newblock \showarticletitle{Toward automatic verification of quantum programs}.
\newblock \bibinfo{journal}{\emph{Formal Aspects of Computing}}
  \bibinfo{volume}{31}, \bibinfo{number}{1} (\bibinfo{date}{01 Feb}
  \bibinfo{year}{2019}), \bibinfo{pages}{3--25}.
\newblock


\bibitem[\protect\citeauthoryear{Ying, Ying, and Wu}{Ying
  et~al\mbox{.}}{2017}]%
        {YYW17}
\bibfield{author}{\bibinfo{person}{Mingsheng Ying}, \bibinfo{person}{Shenggang
  Ying}, {and} \bibinfo{person}{Xiaodi Wu}.} \bibinfo{year}{2017}\natexlab{}.
\newblock \showarticletitle{Invariants of Quantum Programs: Characterisations
  and Generation} \emph{(\bibinfo{series}{POPL 2017})}.
  \bibinfo{pages}{818–832}.
\newblock


\bibitem[\protect\citeauthoryear{{Yu}}{{Yu}}{2019}]%
        {2019arXiv190800158Y}
\bibfield{author}{\bibinfo{person}{Nengkun {Yu}}.}
  \bibinfo{year}{2019}\natexlab{}.
\newblock \showarticletitle{{Quantum Temporal Logic}}.
\newblock \bibinfo{journal}{\emph{arXiv e-prints}}, Article
  \bibinfo{articleno}{arXiv:1908.00158} (\bibinfo{date}{July}
  \bibinfo{year}{2019}), \bibinfo{numpages}{arXiv:1908.00158}~pages.
\newblock
\showeprint[arxiv]{1908.00158}~[cs.LO]


\bibitem[\protect\citeauthoryear{Yu and Palsberg}{Yu and Palsberg}{2021}]%
        {YuPalsberg21}
\bibfield{author}{\bibinfo{person}{Nengkun Yu} {and} \bibinfo{person}{Jens
  Palsberg}.} \bibinfo{year}{2021}\natexlab{}.
\newblock \showarticletitle{Quantum Abstract Interpretation}. In
  \bibinfo{booktitle}{\emph{Proceedings of the 42nd ACM SIGPLAN International
  Conference on Programming Language Design and Implementation}} (Virtual,
  Canada) \emph{(\bibinfo{series}{PLDI 2021})}. \bibinfo{publisher}{Association
  for Computing Machinery}, \bibinfo{address}{New York, NY, USA},
  \bibinfo{pages}{542–558}.
\newblock
\showISBNx{9781450383912}
\urldef\tempurl%
\url{https://doi.org/10.1145/3453483.3454061}
\showDOI{\tempurl}


\bibitem[\protect\citeauthoryear{Zhou, Yu, and Ying}{Zhou
  et~al\mbox{.}}{2019}]%
        {ZYY19}
\bibfield{author}{\bibinfo{person}{Li Zhou}, \bibinfo{person}{Nengkun Yu},
  {and} \bibinfo{person}{Mingsheng Ying}.} \bibinfo{year}{2019}\natexlab{}.
\newblock \showarticletitle{An Applied Quantum {Hoare} Logic}
  \emph{(\bibinfo{series}{PLDI 2019})}. \bibinfo{pages}{1149–1162}.
\newblock


\end{thebibliography}
